\documentclass[11pt]{article}
\pdfoutput=1
\usepackage[ascii]{inputenc}
\usepackage[bookmarksnumbered,hypertexnames=false,colorlinks=false,linkcolor=blue,urlcolor=blue,citecolor=blue,anchorcolor=green,pdfusetitle]{hyperref}
\usepackage{bbm,braket,microtype,mathrsfs,amsmath,amssymb,amsthm,mathtools,fullpage,graphicx,enumitem,caption,overpic,caption,booktabs,xr,framed}
\usepackage[svgnames]{xcolor}
\usepackage[numbers,sort&compress]{natbib}
\usepackage[T1]{fontenc}
\usepackage[USenglish]{babel}
\usepackage[capitalize]{cleveref}
\usepackage[noblocks]{authblk}
\usepackage{thmtools}
\usepackage{thm-restate}
\usepackage[T1]{fontenc}
\usepackage{lmodern}
\newtheorem{thm}{Theorem}\crefname{thm}{Theorem}{Theorems}
\numberwithin{thm}{section}
\newtheorem*{approxthm*}{Approximation Theorem}
\newtheorem*{classthm*}{Classification Theorem}
\newtheorem{lem}[thm]{Lemma}\crefname{lem}{Lemma}{Lemmas}
\newtheorem{prop}[thm]{Proposition}\crefname{prop}{Proposition}{Propositions}
\newtheorem{cor}[thm]{Corollary}\crefname{cor}{Corollary}{Corollaries}
\theoremstyle{definition}
\newtheorem{dfn}[thm]{Definition}\crefname{def}{Definition}{Definitions}
\newtheorem{rmk}[thm]{Remark}\crefname{rem}{Remark}{Remarks}
\newtheorem{ex}[thm]{Example}
\numberwithin{equation}{section}
\setlist[enumerate,1]{label=(\roman*)}
\setlist[enumerate,2]{label=(\alph*)}

\newcommand{\A}{{\mathcal A}}
\newcommand{\B}{{\mathcal B}}
\newcommand{\C}{{\mathcal C}}
\newcommand{\D}{{\mathcal D}}
\renewcommand{\L}{{\mathcal L}}
\newcommand{\R}{{\mathcal R}}
\newcommand{\M}{{\mathcal M}}
\newcommand{\N}{{\mathcal N}}
\newcommand{\ZZ}{\mathbb Z}
\newcommand{\RR}{\mathbb R}
\newcommand{\CC}{\mathbb C}
\newcommand{\EE}{\mathbb E}
\newcommand{\ot}{{\, \otimes\, }}
\newcommand{\mc}[1]{\mathcal{#1}}
\newcommand{\bigO}{\mathcal O}
\newcommand{\littleO}{o}
\newcommand{\eps}{\varepsilon}
\renewcommand{\epsilon}{\varepsilon}
\renewcommand{\d}{\ensuremath{\mathrm{d}}}

\DeclareMathOperator{\tr}{tr}
\DeclareMathOperator{\id}{id}
\DeclareMathOperator{\ind}{ind}
\DeclareMathOperator{\round}{round}
\DeclareMathOperator{\sgn}{sgn}

\DeclareMathOperator{\diam}{diam}
\DeclareMathOperator{\vn}{vN}
\DeclarePairedDelimiter{\abs}{\lvert}{\rvert}
\DeclarePairedDelimiter{\norm}{\lVert}{\rVert}
\DeclarePairedDelimiter\ceil{\lceil}{\rceil}

\begin{document}

\title{A converse to Lieb-Robinson bounds in one dimension \texorpdfstring{\\}{} using index theory}
\date{}
\author[1,2]{Daniel Ranard}
\author[3,4,5]{Michael Walter}
\author[3]{Freek Witteveen}
\affil[1]{Stanford Institute for Theoretical Physics, Stanford University}
\affil[2]{Center for Theoretical Physics, Massachusetts Institute of Technology}
\affil[3]{Korteweg-de Vries Institute for Mathematics and QuSoft, University of Amsterdam}
\affil[4]{Institute for Theoretical Physics, Institute for Language, Logic, and Computation, University of Amsterdam}
\affil[5]{Faculty of Computer Science, Ruhr University Bochum}
\hypersetup{pdfauthor={Daniel Ranard, Michael Walter, Freek Witteveen}}
\maketitle
\begin{abstract}
Unitary dynamics with a strict causal cone (or ``light cone'') have been studied extensively, under the name of quantum cellular automata (QCAs).  In particular, QCAs in one dimension have been completely classified by an index theory. Physical systems often exhibit only approximate causal cones; Hamiltonian evolutions on the lattice satisfy Lieb-Robinson bounds rather than strict locality. This motivates us to study approximately locality preserving unitaries (ALPUs). We show that the index theory is robust and completely extends to one-dimensional ALPUs. As a consequence, we achieve a converse to the Lieb-Robinson bounds: any ALPU of index zero can be exactly generated by some time-dependent, quasi-local Hamiltonian in constant time.  For the special case of finite chains with open boundaries, any unitary satisfying the Lieb-Robinson bound may be generated by such a Hamiltonian. We also discuss some results on the stability of operator algebras which may be of independent interest.
\end{abstract}
\tableofcontents


\section{Introduction}
While quantum dynamics of closed systems are always unitary, systems of interest often possess an additional property: information propagates at finite speeds. In quantum field theories or local quantum circuits, information is strictly constrained to spread within a region called the light cone, or causal cone.  Systems with strict causal cones are called \emph{quantum cellular automata} (QCA)~\cite{margolus1986quantum,schumacher2004reversible}; or sometimes \emph{locality-preserving unitaries}.  However, the effective theories governing laboratory systems are only constrained by an \textit{approximate} causal cone, see \cref{fig:ALPU intro}.  For instance, nontrivial time evolution by a fixed local lattice Hamiltonian never satisfies a strict causal cone,\footnote{In finite-dimensional lattice systems, given two operators $A$ and $B$ on distant sites, if $[A,B(t)]$ is exactly zero in some interval $t \in [0,t^*]$, it must be zero always by analyticity. The system therefore violates any exact causal cone unless there is zero spread of information.} but it does exhibit an approximate causal cone, given by the Lieb-Robinson bounds~\cite{lieb1972}.

Evolutions with approximate causal cones constitute a wide class of natural systems. We can ask general questions about this class of dynamics, e.g.\ when can the evolution be generated by some local Hamiltonian, or when can one evolution be continuously deformed into another?  These fundamental questions also have application in the study of topological phases in many-body physics~\cite{po2016chiral, haah2018nontrivial}.

The rich theory of QCAs addresses these questions for the case of strict causal cones.
In~\cite{gross2012index} (GNVW) it was shown that in one dimension, any QCA is the composition of some local circuit and translation operators.
The total ``flux'' generated by the translation is discretized and measured by a numerical index, which completely classifies QCAs in terms of a computable invariant, yielding one of the most important tools in the study of QCAs.
In \cref{sec:gnvw index} we review the GNVW index.
More recently it has been shown~\cite{freedman2020classification, haah2019clifford} that similarly in two dimensions, any QCA can be written as a composition of a circuit and a generalized shift (i.e.\ a permutation of nearby lattice sites).
However, in higher dimensions strong evidence suggest the existence of QCAs not of this type~\cite{haah2018nontrivial}.
QCAs have found many and wide-ranging applications at the interface of quantum many-body physics and quantum information theory.
To mention a few recent examples, QCAs have been studied in the context of discretization of quantum field theories~\cite{arrighi2020quantum, bisio2018thirring}, quantum hydrodynamics and operator growth~\cite{alba2019operator, gopalakrishnan2018hydrodynamics}, subsystem symmetries and computational phases of matter~\cite{stephen2019subsystem}, tensor networks and matrix product operators~\cite{cirac2017matrix, csahinouglu2018matrix, gong2020classification, piroli2020quantum}, and topological phases of many-body localized dynamics~\cite{po2016chiral, zhang2020classification}.
See~\cite{farrelly2019review, arrighi2019overview} for recent reviews.

Nearly all rigorous results about QCAs rely heavily on their strict locality. In order to apply insights about QCAs to the ``real'' quantum lattice systems commonly encountered, one must first extend the theory of QCAs to the approximate case.  In this work we make an important first step by developing the theory of approximately locality-preserving unitaries (ALPUs), i.e.\ unitary evolutions with approximate causal cones, for one-dimensional systems.  In particular, we extend the topological index theory of QCAs in one dimension~\cite{gross2012index} to the case of ALPUs.  The extended index theory covers local Hamiltonian evolutions and perhaps other naturally occurring evolutions with approximate locality, as in \cref{sec:floquet}. We generally call a Hamiltonian local (or quasi-local, for emphasis) whenever its interactions decay sufficiently with distance, and our results refer to varying notions of decay.

Some interesting features of QCAs in one dimension are easily demonstrated by simple example.  Consider an infinite spin chain. The index theory in~\cite{gross2012index} shows that it is not possible to implement a translation operator by using a finite depth circuit. Is this still true if we allow time-dependent Hamiltonian evolutions? Is this perhaps possible if we allow Hamiltonian evolutions with polynomial tails? Naively, the classification of QCAs might have been expected to ``collapse'' under the introduction of ALPUs (especially when allowing polynomial tails), or alternatively the classification might have become more exotic. It turns out that neither is the case; we show that almost all properties of the classification and index theory of QCAs can be generalized to ALPUs with $\littleO(\frac{1}{r})$ tails, or even $\littleO(1)$ tails in many cases.

More generally, while local Hamiltonians satisfy Lieb-Robinson bounds and therefore generate ALPUs, we may conversely ask the following question:
\begin{center}
\begin{framed}
Given an automorphism $\alpha$ satisfying Lieb-Robinson bounds (i.e.\ an ALPU), can it be generated by some time-dependent local Hamiltonian?
If not, what are the obstructions?
\end{framed}
\end{center}
As foreshadowed by the GNVW index theory, it turns out the only obstruction to finding such a Hamiltonian in a one-dimensional system is when $\alpha$ has nonzero index. Thus, our classification offers a ``converse'' to the Lieb-Robinson bounds.  For instance, we show that an ALPU with exponentially decaying tails can be generated by a time-dependent Hamiltonian with exponentially decaying interactions precisely when the ALPU has index zero; we also find related statements for tails of slower decay, as in \cref{cor:LR-converse}.  If $\alpha$ has a nonzero index it can be constructed as the composition of a generalized shift and an ALPU of index zero, which can then be generated by some time-dependent Hamiltonian evolution.  Meanwhile, for a non-periodic chain of finite length, the index is always zero.  In that case, we conclude that the dynamics satisfy Lieb-Robinson bounds if and only if they are generated by a local Hamiltonian with sufficiently decaying tails.

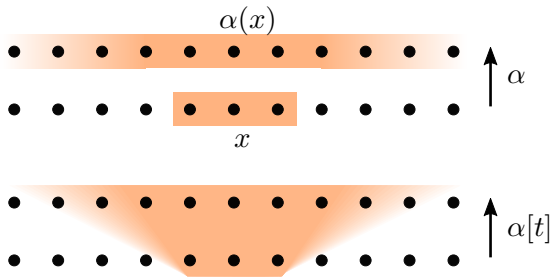
\begin{figure}
\centering
\begin{overpic}[width=0.4\textwidth,grid=false]{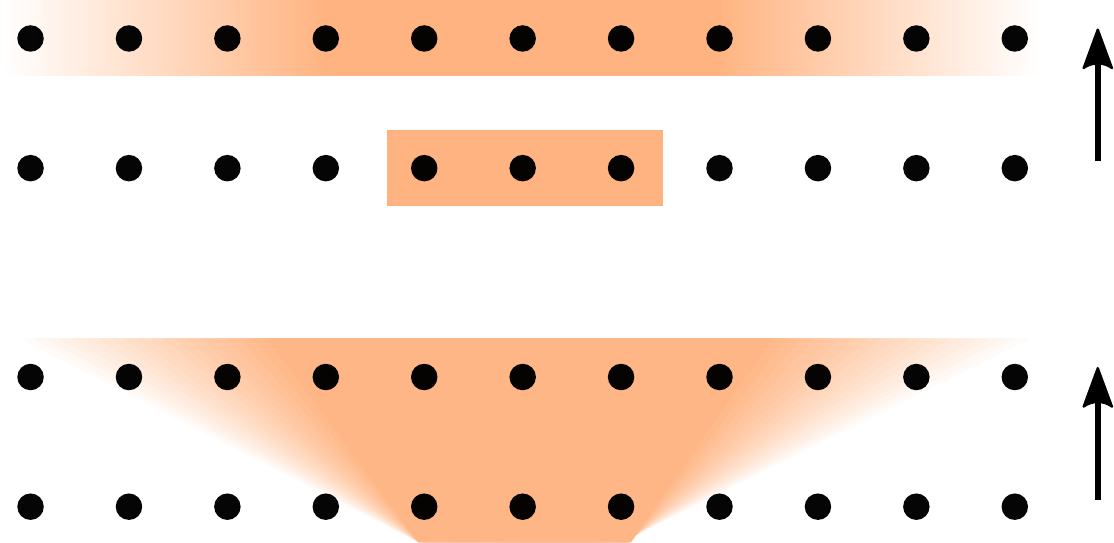}
\put(47,26){$x$}
\put(44,50){$\alpha(x)$}
\put(102,39){$\alpha$}
\put(102,8){$\alpha[t]$}
\end{overpic}
\caption{Illustration of an automorphism $\alpha$ with an approximate lightcone (an ALPU). Given such an $\alpha$, does there exist a continuous dynamics $\alpha[t]$ with $\alpha[0] = \id$ and $\alpha[1] = \alpha$ which remains approximately local at all times?}
\label{fig:ALPU intro}
\end{figure}

\subsection{Prior work}
To generalize the GNVW index to ALPUs, a natural concern is the sensitivity of the index to small perturbations of the QCA. However, the dependence of the sensitivity on the local Hilbert space dimension and the radius of the QCA is not immediately clear from the considerations in~\cite{gross2012index}, which yield a relatively weak continuity estimate.
Without stronger bounds, it appears possible that the homotopy classes established in~\cite{gross2012index} might collapse when considering ALPUs: two QCAs with different GNVW index might still be connected by a strongly continuous path through the space of ALPUs with some prescribed tails.
Another concern is whether the generalized index would take values in the same discrete set as the GNVW index. Relatedly, we ask whether every ALPU can be approximated to arbitrary accuracy $\delta$ by a QCA whose radius does not grow too fast with $\delta$.
In the special case of Hamiltonian evolutions, Lieb-Robinson-type estimates allow one to approximate the evolution by a strictly local quantum circuit~\cite{haah2018quantum, tran2019locality}, as in digital quantum simulation.  We might hope for a similar discretization or Trotterization procedure for arbitrary ALPUs, or at least for index-0 QCAs.

These questions are recognized in existing literature.
In fact, one of the main open questions in the original work~\cite{gross2012index} was how to extend the index theory to some class of automorphisms with only approximate locality.  Later work asked specifically whether ALPUs could be approximated by QCAs~\cite{hastings2013classifying}.  In the review~\cite{farrelly2019review} such questions were raised again, highlighting their relevance for the application of index theory to actual physical systems.
As an example, the GNVW index has been proposed to classify two-dimensional many-body localized Floquet phases, by computing the index of a certain dynamics that arises on the boundary of the system~\cite{po2016chiral}.
Such dynamics are typically not strictly local, and in \cref{sec:floquet} we comment on this specific application.  Other recent work~\cite{freedman2020classification} also suggested the extension to ALPUs as an avenue for research, proposing that one approach might involve Ulam stability results for operator algebras.  This is precisely the approach taken in this work.  The stability results we use~\cite{christensen1977perturbation, christensen1980near} and augment were developed throughout the 1970s and 80s for studying how operator algebras behave under perturbation.  Intriguingly, related questions about perturbations were tackled under a different guise in~\cite{chao2017overlapping} (cf. their Theorem 3.6), in the context of quantum device certification.

Regarding the converse to the Lieb-Robinson bounds, see~\cite{wilming2020lieb} for interesting work which develops a related converse with different assumptions.  They show that if you already know $\alpha$ is generated by a $k$-local Hamiltonian satisfying a Lieb-Robinson-like condition, then the evolution can also be generated by a \emph{geometrically} local $k$-local Hamiltonian.  Their condition can be checked at infintesimal times.  In contrast, we do not assume the ALPU is generated by any $k$-local Hamiltonian.

\subsection{Summary of results}
\begin{figure}
\centering
\begin{overpic}[width=0.8\textwidth,grid=false]{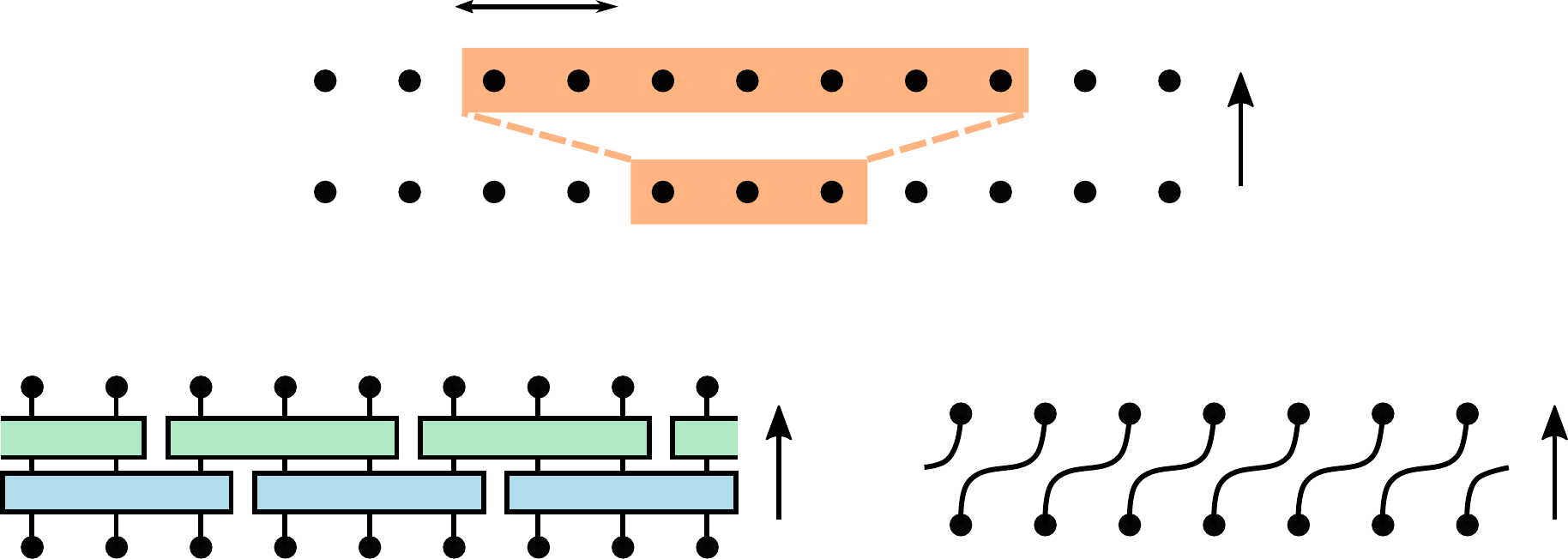}
\put(12,26){(a)} \put(-6,4.5){(b)}
\put(47,19){$x$}
\put(44,35){$\alpha(x)$}
\put(33,37){$R$}
\put(81,26){$\alpha$}
\put(51,4.5){$\beta$} \put(101,4.5){$\sigma$}
\end{overpic}
\caption{(a) Illustration of a QCA $\alpha$ with radius $R = 2$ mapping an operator $x$ supported on three sites to an operator $\alpha(x)$ supported on seven sites. (b) A local circuit QCA $\beta$ (left) and a translation QCA $\sigma$ (right).}
\label{fig:LPU examples}
\end{figure}
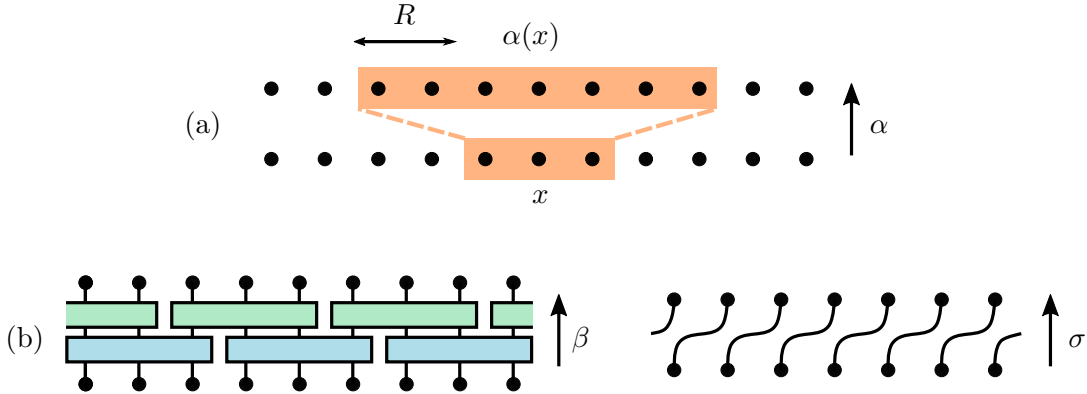

The theory is typically formulated in the Heisenberg picture, acting on operators rather than states.
This definition is more natural for infinite systems.
Unitary dynamics of the quantum system can then be described by an automorphism of an operator algebra.
A \emph{quantum cellular automaton}~(QCA) with \emph{radius} $R$ on some lattice of spin systems is an automorphism $\alpha$ of the algebra $\A$ of operators on the spin system, such that if an operator $x \in \A$ is supported on a set of sites $X$, then~$\alpha(x)$ is an operator supported on $B(X,R)$, the set of sites within distance $R$ of $X$.
For an \emph{approximately locality-preserving unitary} (ALPU) with \emph{tails} $f(r)$, we only ask that $\alpha(x)$ can be approximated by operators supported on $B(X,r)$ up to an error $f(r)$ for any $r$, as detailed in \cref{dfn:alpu}.  We also require that $\lim_{r \to \infty}f(r)=0$.
In other words, the map $\alpha$ satisfies a Lieb-Robinson bound governed by $f(r)$.
We will restrict to the situation where we have an infinite one-dimensional lattice with finite local dimension (i.e.\ a spin chain).

Two examples of QCAs are shown in \cref{fig:LPU examples},
\begin{enumerate}
\item\label{it:ex circuit} A local circuit, i.e. a composition of multiple layers of application of strictly local unitaries.
\item\label{it:ex translation} In the case where each local Hilbert space is identical, we have the translation automorphism, which simply shifts an operator by one site to the left.
Notice that in the (perhaps more intuitive) Schr\"{o}dinger picture this corresponds to a shift to the \emph{right} of the state.
\end{enumerate}
In~\cite{gross2012index} it was shown that these two examples generate \emph{all} examples, in the sense that any QCA can be written as a composition of (tensor products of) translations and circuits.
It is quite intuitive that in example~\ref{it:ex translation} there is a `flux' of information to the right, whereas in example~\ref{it:ex circuit} the net `flux' is zero.
This suggests that computing some sort of flux allows you to extract from a QCA how many translations you need to implement it.
This intuition was made precise in~\cite{gross2012index} by defining an \emph{index} (the GNVW index) which measures the flow of quantum information, based on ideas in \cite{kitaev2006anyons}.
In this work we first observe that one can re-formulate the definition of the index as follows:
one divides the chain into a left half $L$ and right half $R$, and one considers the Choi state~$\phi_{LR,L'R'}$ of the automorphism~$\alpha$.
Then the mutual information difference
\begin{align}\label{eq:mi index intro}
  \ind(\alpha) =\frac{ I(L':R)_{\phi} - I(L:R')_{\phi}}{2},
\end{align}
is precisely the index of~\cite{gross2012index}, but also well-defined for ALPUs with appropriately decaying tails!
In addition, the mutual information enjoys much better continuity than the related expression for the index in Eq.~(45) of~\cite{gross2012index}, which (in hindsight) can be understood as a difference of Renyi-2 entropies.
In the context of two-dimensional Floquet phases a similar expression has been derived in~\cite{duschatko2018tracking}.
We show that the expression in \cref{eq:mi index intro} generalizes to the approximately local setting.

\begin{figure}
\centering
\begin{overpic}[width=0.8\textwidth,grid=false]{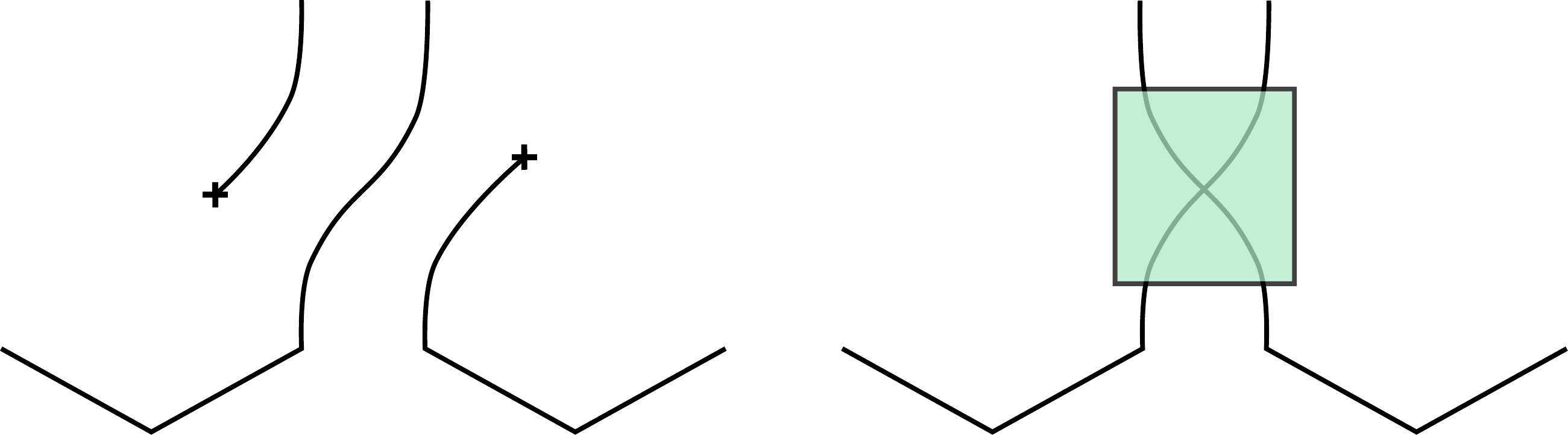}
\put(-2,20){(a)} \put(52,20){(b)}
\put(18,30){$L$} \put(26,30){$R$}
\put(0,7){$L'$}
\put(44,7){$R'$}
\put(98,7){$R'$} \put(80,30){$R$}
\put(54,7){$L'$} \put(72,30){$L$}
\end{overpic}
\caption{Illustration of~\eqref{eq:mi index intro}. For the translation on a qudit with local dimension $d$ in (a) we have $I(L':R)_{\phi} = 2\log(d)$ and $I(L:R')_{\phi} = 0$, so the index equals $d$. For a circuit, one can show that applying a local unitary as in (b) gives $I(L':R)_{\phi} = I(L:R')_{\phi}$, so the index is zero.}
\label{fig:index intro}
\end{figure}
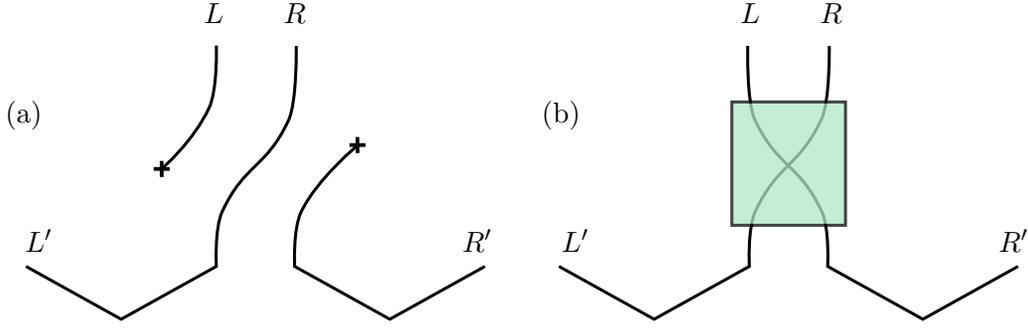

Our first main result consists of \cref{thm:qca approx,thm:index alpu after reorg}, summarized as
\begin{approxthm*}[informal]
Suppose that $\alpha$ is an ALPU in one dimension.
Then there exists a sequence of QCAs $\alpha_j$ of increasing radius such that $\alpha_j(x)$ converges to $\alpha(x)$ for any local operator~$x$, and that $\ind(\alpha_j)$ stabilizes for large $j$.
We define $\ind(\alpha) = \lim_{j \rightarrow \infty} \ind(\alpha_j)$.
If~$\alpha$ has $\bigO(\frac{1}{r^{1 +\delta}})$-tails for $\delta>0$, the index defined in~\eqref{eq:mi index intro} is finite and equal to $\ind(\alpha)$.
The exact index may also be computed locally through a rounding procedure.
\end{approxthm*}

To be precise, the error bounds are such that if $\alpha$ has tails like $\bigO(f(r))$, and $\alpha_j$ has radius $j$, then for an operator $x$ supported on an interval of $n$ sites, $\norm{\alpha_j(x) - \alpha(x)} = \bigO(f(j)(\frac{n}{j}+1)\norm{x})$.
The key technical ingredient we use is a stability result for inclusions of possibly infinite algebras which we state as \cref{thm:near inclusion}, an extension of results from~\cite{christensen1977perturbation, christensen1980near}.
This result deals with the situation where $\A$ and $\B$ are algebras of observables and $\A$ is ``nearly'' included in~$\B$, meaning that for each $x \in \A$ there is a $y \in \B$ such that $\norm{x-y} \leq \eps \norm{x}$ for some small $\eps$.
Then (under some technical but very general assumptions on the algebras), there exists a unitary $u \in B(\mc H)$ close to the identity, $\norm{u - I} = \bigO(\eps)$ with error independent of $\dim(\mc H)$, such that $u\A u^*$ is strictly contained in~$\B$.
Loosely speaking, we construct the QCAs $\alpha_j$ by ``localizing'' the images $\alpha(\A_n)$ of the algebra~$\A_n$ at each site $n$, by rotating $\alpha(\A_n)$  into an algebra supported within some radius of site $n$.  The main technical effort in the construction is to ensure the rotations are compatible and the errors do not accumulate.

The index defines equivalence classes of ALPUs.  These are characterized in our second main result, \cref{thm:properties index alpu}, sketched below:
\begin{classthm*}[informal]
Suppose $\alpha$ and $\beta$ are ALPUs with $f(r)$-tails in one dimension.
Then the following are equivalent conditions:
\begin{enumerate}
\item  $\ind(\alpha) = \ind(\beta)$.
\item $\alpha = \beta\gamma$ where $\ind(\gamma) = 0$.
\item There exists a ``blended'' ALPU which (up to small error) matches $\alpha$ on the left of the chain and matches $\beta$ on the right.
\item\label{it:informal lr converse} There exists a strongly continuous path from $\alpha$ to $\beta$ through the space of ALPUs with $g(r)$-tails for some $g(r)=\littleO(1)$.
If such a path exists, then it can be generated by evolving a time-dependent quasi-local Hamiltonian for unit time.
\end{enumerate}
\end{classthm*}

\noindent
In particular, \ref{it:informal lr converse} provides a converse to the Lieb-Robinson bounds in one dimension: an automorphism can be generated by evolution along a time-dependent Hamiltonian with certain locality bounds if and only if it has index zero.
Thus we see that the index theory of~\cite{gross2012index} completely generalizes to ALPUs and does not ``collapse,'' with as only essential difference that the role of quantum circuits is replaced by time evolutions along time-dependent geometrically local Hamiltonians.

As an application, it follows immediately that the translation operator cannot be implemented by a finite time evolution of any (time-dependent) Hamiltonian satisfying Lieb-Robinson bounds.
Moreover, there cannot exist a quasi-local ``momentum density'' that generates a lattice translation and also satisfies Lieb-Robinson bounds with $\littleO(1)$-tails at all times.
To show that it is necessary to impose some bound on the decay of the ALPU tails in our constructions, we give an example of a strongly continuous path of automorphisms generated by a Hamiltonian with $\frac{1}{r}$-decaying interactions that connects the identity map to a translation on a chain of qubits, showing that at this point the index theory does indeed collapse.
As a second potential application we discuss the definition of the index for two-dimensional Floquet systems with many-body localization.

We cast our results directly in the setting of infinite one-dimensional lattices. Finite chains with non-periodic boundary conditions become a straightforward special case; see \cref{sec:finite}.  There the index is always zero, and we obtain a universal converse to the Lieb-Robinson bounds.\footnote{The case of finite chains with periodic boundary conditions appears more difficult.  While we expect the index theory there to match that of the infinite lattice, we cannot offer rigorous results.}

This works is organized as follows: in \cref{sec:op algs} we review some basic properties of operator algebras, then discuss perturbations of operator algebras, including some new tools developed for this work.  (We expect that results like \cref{lem:homomorphism local error}, \cref{lem:simultaneous near inclusions}, and \cref{lem:lr from single site}  may also find broad application, e.g.\ in the development of Lieb-Robinson bounds.)  In \cref{sec:spin systems} and we define ALPUs.  In \cref{sec:gnvw index} we review the GNVW index theory, prove~\eqref{eq:mi index intro}, and discuss robustness of the index.  \cref{sec:alpu approximation} is the technical heart of this work, where we show how to construct a sequence of approximating QCAs to any ALPU and from this result derive the index theory for ALPUs.  In \cref{sec:applications} we discuss the two applications: the impossibility of finding a Hamiltonian for the translation operator and the definition of an index for two-dimensional many-body localized Floquet systems.
Finally, in \cref{sec:near inclusion appendix} we provide a self-contained proof of \cref{thm:near inclusion}, the stability result for algebra inclusions, which may be of independent interest.

\section{Operator algebras}\label{sec:op algs}
For infinite dimensional quantum mechanical systems it is often more convenient to work with operator algebras (algebras of observables) rather than Hilbert spaces, and use the Heisenberg rather than Schr\"{o}dinger picture of quantum mechanics.
A standard reference for operator algebras and their relation to quantum physics is~\cite{bratteli2012operator}; see~\cite{naaijkens2013quantum} for an accessible introduction.  We review $C^*$-algebras and von Neumann algebras, focusing especially on facts used in subsequent proofs.  Then we turn to methods for ``perturbations'' (e.g.\ small rotations) of operator algebras in \cref{sec:near inclusion prelim}.

\subsection{\texorpdfstring{$C^*$}{C*}-algebras}\label{subsec:opalg}
The notion of an operator algebra is formalized by a \emph{$C^*$-algebra}, which is a complex algebra $\A$ with a norm $\norm{\cdot}$ and an anti-linear involution $x \mapsto x^*$, satisfying
\begin{itemize}
\item $\A$ is complete in $\norm{\cdot}$,
\item $\norm{xy} \leq \norm{x}\norm{y}$,
\item $\norm{x^*x} = \norm{x}^2$.
\end{itemize}
We will only use algebras with an identity element~$I$.
An important example is the $C^*$-algebra~$B(\mc H)$ of operators on some Hilbert space~$\mc H$, where we take the operator norm as the norm, and the adjoint as the $*$-operation.
In finite dimensions this reduces to the algebra~$\M_{d\times d}$ of complex $d\times d$ matrices with the spectral norm and Hermitian conjugate.
A $C^*$-algebra $\A$ is called \emph{approximately finite-dimensional} (AF) if it contains a directed collection of finite-dimensional subalgebras whose union is dense in $\A$.
If $\A \subseteq \B$ are $C^*$-algebras we define the \emph{commutant} of $\A$ in $\B$ as
$\A' = \{x \in \B \text{ such that } [x,\A] = 0\}$,
which is again a $C^*$-algebra.
We denote by $U(\A)$ the set of elements $u \in \A$ that are \emph{unitary}, meaning that $uu^* = u^*u = I$.

A \emph{$*$-homomorphism}~$\alpha \colon \A \rightarrow \B$ between $C^*$-algebras is an algebra homomorphism which also preserves the $*$-operation, $\alpha(x^*) = \alpha(x)^*$.
Such a $*$-homomorphism is automatically continuous and indeed contractive, i.e., $\norm{\alpha(x)} \leq \norm x$.
The latter can also be written as $\norm\alpha\leq1$, where we define the notation~$\norm\beta = \sup_{\norm x\leq1} \norm{\beta(x)}$ for any linear map $\beta$ between $C^*$-algebras.
A \emph{($*$-)automorphism} is a bijective $*$-homomorphism.
The inverse of an automorphism is again a $*$-homomorphism, and any automorphism is automatically unital and isometric.
We write $\id$ for the identity automorphism.
Finally, a \emph{state} on a $C^*$-algebra $\A$ is given by a linear functional $\omega \colon \A \rightarrow \CC$ which is positive (meaning that $\omega(x^*x) \geq 0$ for all $x \in \A$) and normalized (meaning that $\omega(I) = 1$).

It turns out that any $C^*$-algebra can be realized as a subalgebra of $B( \mc H)$, the algebra of bounded operators on some Hilbert space $\mc H$.
This is proven by the following result known as the \emph{Gelfand-Naimark-Segal (GNS) construction} or \emph{representation}:

\begin{thm}[Gelfand-Naimark-Segal]\label{thm:gns}
Given a state $\omega$ on $\A$, there exists a Hilbert space~$\mc H$, a $*$-homomorphism $\pi \colon \A \rightarrow B(\mc H)$, and a cyclic vector $\phi$ (meaning $\pi(\A)\phi$ is dense in $\mc H$) such that
\begin{align*}
  \omega(x) = \langle \phi, \pi(x) \phi \rangle
\end{align*}
Moreover, if $(\mc H',\pi',\phi')$ is another triple as above then there exists a unique unitary $u \colon \mc H \rightarrow \mc H'$ such that $\phi' = u \phi$ and $\pi'(x) = u\pi(x)u^*$ for all $x \in \A$.
\end{thm}

\noindent
If $\omega$ is such that $\omega(x^*x) = 0$ implies $x = 0$, then the GNS representation is faithful (meaning that~$\pi_\omega$~is injective).
In that case, one way to construct the Hilbert space in \cref{thm:gns} is by letting $\langle x, y \rangle = \omega(x^*y)$ define an inner product on $\A$ and letting $\mc H$ be the completion of $\A$ with respect to this inner product.
Then $\A$ acts on $\mc H$ by left multiplication, which defines the $*$-homomorphism $\pi\colon \A \rightarrow B(\mc H)$.
The identity $I \in \A$ gives rise to a cyclic vector $\phi \in \mc H$.

\subsection{Von Neumann algebras}\label{sec:vN prelims}

A special class of $C^*$-algebras are von Neumann algebras.
A $C^*$-algebra $\A \subseteq B(\mc H)$ is a \emph{von Neumann algebra} if it is equal to its double commutant in $B(\mc H)$,
\begin{align*}
  \A = \A''.
\end{align*}
In fact, for any subset $S \subseteq B(\mc H)$, the double commutant~$S''$ is always a von Neumann algebra, called the von Neumann algebra generated by $S$.
It is the smallest von Neumann algebra that contains $S$.
There are various relevant topologies on $B(\mc H)$.
The \emph{strong operator topology} is such that a net $x_i$ converges to some operator $x$ if and only if $x_i v \to x v$ for each vector $v \in \mc H$.
The \emph{weak operator topology} is such that a net $x_i$ converges to some operator $x$ if and only if $\langle w,x_i v \rangle \to \langle w,x v\rangle$ for each pair $v, w \in \mc H$.
The weak operator topology is weaker than the strong operator topology, and both are weaker than the topology induced by the norm.
Sometimes also the weak-$*$ topology is relevant, induced by interpreting $B(\mc H)$ as the dual space of the trace class operators on~$\mc H$.
The weak operator topology is weaker than the weak-$*$ topology, but the two coincide on norm-bounded subsets of~$B(\mc H)$.
On convex subsets the weak operator closure and strong operator closure coincide.
The von \emph{Neumann bicommutant theorem} states that for unital $*$-subalgebra $\A \subseteq B(\mc H),$ $\A''$ is the weak operator closure of $\A$, so $\A$ is a von Neumann algebra if and only if $\A$ is weak operator closed.
Any $*$-automorphism of a von Neumann algebra is continuous with respect to the weak-$*$ topology.
Moreover, norm balls are compact in the weak operator topology (Theorem~5.1.3 in \cite{kadison1997fundamentals}, a consequence of the Banach-Alaoglu theorem).

A useful fact in the study of von Neumann algebras is the \emph{Kaplansky density theorem}, which states that for any self-adjoint subalgebra~$\A \subseteq B(\mc H)$ the unit ball of the strong operator closure of~$\A$ equals the strong operator closure of the unit ball of $\A$.
We refer to \cite[\S2.4]{bratteli2012operator} for more details.
Another useful fact is that the norm is lower semi-continuous in the weak operator topology, i.e., if a net $x_i$ converges to $x$ in the weak operator topology then $\norm x \leq \liminf_i \norm{x_i}$.

In infinite dimensions, working with a von Neumann algebra often confers advantages over more general $C^*$-algebras.
For instance, the output of the Borel functional calculus (taking functions of operators) on a $C^*$-algebra $\A \subseteq B(\mc H)$ produces operators that sometimes lie outside $\A$, but they always lie in the weak operator closure $\A''$.
A von Neumann algebra therefore allows one to use spectral projections and other technical tools.

A von Neumann algebra~$\A \subseteq B(\mc H)$ is called a \emph{factor} if it has trivial center, $\A' \cap \A = \CC I$.
In particular, $\A = B(\mc H)$ is a factor (a so-called type~I factor).
Any finite dimensional factor is of this form (for a finite dimensional Hilbert space), but there also exist infinite dimensional factors \emph{not} of the form $B(\mc H)$ (so-called type~II and type~III factors).

A von Neumann algebra~$\mc A$ is called \emph{hyperfinite} (or approximately finite-dimensional) if it contains a directed collection of finite-dimensional subalgebras whose union is dense in the weak operator topology (equivalently, the weak-$*$ topology).
Equivalently, $\mc A$ is hyperfinite when there exists an AF $C^*$-subalgebra~$\A_0 \subseteq \A$ such that $\A_0'' = \A$.

Finally, if $\M \subseteq B(\mc \mc{H})$ and $\N \subseteq B(\mc{K})$ are von Neumann algebras, we use $\M \ot \N \subseteq B(\mc{H} \ot \mc{K})$ to denote the von Neumann algebra tensor product, given by the weak operator closure of the algebraic tensor product of $\M$ and $\N$ in $ B(\mc{H} \ot \mc{K})$.
%
%

\subsection{Near inclusions and stability properties}\label{sec:near inclusion prelim}
We now define our notion of near inclusions of algebras and discuss related stability properties.
The notion of a near inclusion follows e.g.~\cite{christensen1980near}.

\begin{dfn}[Near inclusion]\label{dfn:near inclusion}
For a $C^*$-algebra $\B \subseteq B(\mc{H})$ and an operator $a \in B(\mc{H})$, we write $a \overset{\eps}{\in} \B$ when there exists $b \in \B$ such that $\norm{a-b}\leq \epsilon \norm{a}$.
Likewise for two $C^*$-algebras~$\A, \B \subseteq B(\mc{H})$, we write $\A \overset{\eps}{\subseteq} \B$ and say there is a \emph{near inclusion} whenever $a \overset{\eps}{\in} \B$ for all~$a \in \A$.
\end{dfn}

\noindent
We note that if~$\B$ is a von Neumann algebra, we have $a \overset{\eps}{\in} \B$ if and only if
\begin{align*}
 \inf_{b \in \B} \, \norm{a - b} \leq \eps\norm{a}.
\end{align*}
That is, the infimum is attained by some $b \in \B$.
Indeed, let $b_i$ be a sequence in $\B$ such that $\lim_i \norm{a - b_i} \leq \eps\norm{a}$.
In particular, $\norm{b_i}$ is bounded.
Since (any multiple of) the closed unit ball in~$B(\mc H)$ is compact in the weak operator topology (e.g., Theorem~5.1.4 in~\cite{kadison1997fundamentals}), the sequence~$b_i$ has a limit point $b \in \B$ in the weak operator topology.
Because the norm is lower semi-continuous in the weak operator topology,
$\norm{a - b} \leq \lim_i \norm{a - b_i} \leq \eps\norm{a}$, which concludes the argument.

When $\B \subseteq B(\mc H)$ is a $C^*$-algebra and $x \in \B(\mc H)$ is an operator that is nearly contained in its commutant, say~$x \overset\eps\in \B'$, then it is easy to see that, for any $b \in \B$,
\begin{align}\label{eq:commutator from eps inclusion}
  \norm{[x,b]} \leq 2 \eps \norm x \norm b.
\end{align}
Indeed $x \overset\eps\in \B'$ means there exists $y\in \B'$ such that $\norm{x - y} \leq \eps\norm x$.
Then we have for any $b\in\B$ that
\begin{align*}
  \norm{[x,b]}
= \norm{[x-y,b]}
\leq 2 \norm{x-y} \norm b
\leq 2 \eps \norm x \norm b.
\end{align*}
We will be interested in the converse of this statement, which is rather less clear.

To gain some intuition, we consider the finite-dimensional setting.
Suppose that $\mc H = \mc H_A \ot \mc H_B$ for finite-dimensional Hilbert spaces, and let $\B = I \ot B(\mc H_B) \subseteq B(\mc H)$ be the algebra of operators supported on the second tensor factor.
Then we can define a projection onto the commutant of $\B$ by twirling using Haar probability measure on the group~$U(\B)$ of unitaries on~$\B$:
\begin{align}\label{eq:twirling}
  \EE_{\B'} \colon B(\mc H) \to \B', \quad \EE_{\B'}(x) = \int_{U(\B)} u x u^* \, \d u.
\end{align}
In fact, the commutant is simply $\A = \B' = B(\mc H_A) \ot I_B$, and the projection can equivalently be written in terms of the normalized partial trace, $\EE_{\B'}(x) = \frac1{d_B} \tr_{\B}(x)$, where $d_B = \dim \mc H_B$.
The projection exhibits the desirable property that if
\begin{align}\label{eq:commutator bound}
  \norm{[x,b]} \leq \eps \norm x \norm b
\end{align}
for all $b\in\B$, then
\begin{align}\label{eq:cond exp bounding}
  \norm{x - \EE_{\B'}(x)}
\leq \int_{U(\B)} \norm{x - u x u^*} \, \d u
= \int_{U(\B)} \norm{[x,u]} \, \d u
\leq \eps \norm x.
\end{align}
This shows that, in the finite-dimensional setting, the commutator bound~\eqref{eq:commutator bound} implies that $x \overset\eps\in \B'$.

In infinite dimensions, where no Haar integral is available, we need a different way to define the projection.
One way to do so is using the so-called ``property~P.''
If $\B \subseteq B(\mc H)$ is a von Neumann algebra, it has \emph{property~P} if for any~$x \in B(\mc H)$, there exists some $y \in \B'$ such that $y$ is also in the weak operator closure (equivalently, the weak-$*$ closure) of the convex hull of $\{ u x u^* : u \in U(\B)\}$.
Note that in the finite-dimensional setting this is immediate from the definition in terms of the Haar integral.
We can also generalize the notion of twirling by a (non-commutative) \emph{conditional expectation} $\EE_{\A}$ onto a von Neumann algebra $\A \subseteq B(\mc H)$.
This is defined to be contractive completely positive linear map $\EE_{\A}\colon B(\mc H) \rightarrow \A \subseteq  B(\mc H)$ which is such that for $x \in B(\mc H)$ and $a,a' \in \A$ we have $\EE_{\A}(a) = a$ and $\EE_{\A}(axa') = a\EE_{\A}(x) a'$.
A von Neumann algebra $\A \subseteq B(\mc H)$ is called \emph{injective} if there exists such a conditional expectation~\cite[IV.2.1.4]{blackadar2006operator}.

For von Neumann algebras acting on separable Hilbert spaces, these properties are equivalent to each other and to hyperfiniteness as defined earlier:

\begin{thm}\label{thm:hyperfinite}
Let $\B \subseteq B(\mc H)$ be a von Neumann algebra with $\mc H$ separable. Then the following are equivalent:
\begin{enumerate}
\item\label{item:b hyperfinite} $\B$ is hyperfinite.
\item\label{item:b' hyperfinite} $\B'$ is hyperfinite.
\item\label{item:property p} $\B$ has property~P.
\item\label{item:b' injective} $\B'$ is injective.
\end{enumerate}
If $\mc H$ is not assumed to be separable, it is still true that~\ref{item:b hyperfinite} implies~\ref{item:property p} and \ref{item:property p} implies~\ref{item:b' injective}.
Moreover, $\B$ has property~P if and only if $\B'$ has property~P, and the same is true for injectivity.
\end{thm}

\noindent
For a comprehensive account of the theory and classification of von Neumann algebras see~\cite{takesaki2003theory,blackadar2006operator}.
\Cref{thm:hyperfinite} is proved in Proposition~4.1 of~\cite{nachtergaele2013local} in the case that~$\B$ is a factor.
The general case follows similarly by combination of well-known results, as we sketch for convenience.

\begin{proof}
The implications~\ref{item:b hyperfinite}$\Rightarrow$\ref{item:property p} and~\ref{item:property p}$\Rightarrow$\ref{item:b' injective} are explained in~\cite[IV.2.2.20]{blackadar2006operator}, as is the fact that $\B$ has property~P if and only if $\B'$ has property~P.
Moreover, \cite[IV.2.2.7]{blackadar2006operator} asserts that $\B$ is injective if and only if $\B'$ is injective.
Now assume that $\mc H$ is separable or, equivalently, $\B$ has a separable predual~\cite{martinstack}.
In this case, injectivity implies hyperfiniteness~\cite[IV.2.6.1]{blackadar2006operator}, so it follows that \ref{item:b' injective}$\Rightarrow$\ref{item:b' hyperfinite}.
Since $\B'' = \B$, \ref{item:b hyperfinite}$\Rightarrow$\ref{item:b' hyperfinite} also yields~\ref{item:b' hyperfinite}$\Rightarrow$\ref{item:b hyperfinite}, so that \ref{item:b hyperfinite}--\ref{item:b' injective} are all equivalent.
\end{proof}

When $\B$ is hyperfinite and $x\in B(\mc H)$ is such that $\norm{[x,b]} \leq \eps \norm x \norm b$ for all $b \in \B$, then~$x \overset\eps\in \B'$, providing a converse to the discussion above \cref{eq:commutator from eps inclusion}.
Indeed, since $\B$ has property~P by \ref{item:b hyperfinite}$\Rightarrow$\ref{item:property p}, there exists some~$y \in \B'$ in the weak operator closure of the convex hull of $\{uxu^* : u \in U(\B)\}$.
Using lower semicontinuity of the norm with respect to the weak operator topology, we find
\begin{align*}
  \norm{x - y} \leq \sup_{u \in U(\B)} \norm{x - uxu^*} = \sup_{u \in U(\B)} \norm{[x,u]} \leq \eps\norm{x},
\end{align*}
which shows that $x \overset\eps\in \B'$.
Moreover, if $x \in \M$, $\B \subseteq \M$ for a von Neumann algebra $\M \subseteq B(\mc H)$, then we have that $x \overset\eps\in \B' \cap \M$.
Indeed, in this case $\{uxu^* : u \in U(\B)\}$ is contained in~$\M$, and the same is true for the weak operator closure of its convex hull.
Since~$y$ was constructed as an element of the latter, it follows that $y \in \B' \cap \M$ and hence~$x \overset\eps\in \B' \cap \M$.

In turn, the above implies that any conditional expectation~$\EE_{\B'}\colon B(\mc H)\to\B'$ (and such conditional expectations exist due to \ref{item:property p}$\Rightarrow$~\ref{item:b' injective}) satisfies
\begin{align*}
  \norm{\EE_{\B'}(x) - x}
\leq \norm{\EE_{\B'}(x) - \EE_{\B'}(y)} + \norm{y - x}
\leq 2\eps\norm{x},
\end{align*}
using that $\EE_{\B'}(y) = y$ and that conditional expectations are contractions.
When $\B$ is a factor, a different proof strategy shows that the constant~2 can be omitted; see Proposition~4.1 in~\cite{nachtergaele2013local}.

As an easy consequence we obtain:

\begin{lem}[{Near inclusions and commutators \cite[Theorem~2.3]{christensen1977perturbations}}]\label{lem:near inclusion commutator 0}
Let $\A, \B \subseteq B(\mc H)$ be two $C^*$-algebras.
If $\A \overset{\eps}{\subseteq} \B'$ is a near inclusion, then
\begin{align*}
  \norm{[a,b]} \leq 2\eps\norm{a}\norm{b}.
\end{align*}
holds for all $a \in \A$ and $b \in \B$.

Conversely, if $\B$ is a hyperfinite von Neumann algebra and
\begin{align*}
  \norm{[a,b]} \leq \eps\norm{a}\norm{b}
\end{align*}
holds for all $a \in \A$ and $b \in \B$, then we have a near inclusion $\A \overset{\eps}{\subseteq} \B'$.
If moreover $\A,\B \subseteq \M$ for some von Neumann algebra $\M \ \subseteq B(\mc H)$, then $\A \overset{\eps}{\subseteq} \B' \cap \M$.
\end{lem}
\begin{proof}

The first claim follows from \cref{eq:commutator from eps inclusion}, since $\A \overset\eps\subseteq \B'$ means that $a \overset\eps\in \B'$ for every $a\in\A$.
For the converse claim, the discussion above the lemma shows that for every $a\in \A$ we have $a \overset\eps\in \B'$, hence $\A \overset\eps\subseteq \B'$;
moreover, if $a \in \M$, $\B \subseteq \M$ then $a \overset\eps\in \B' \cap \M$, and hence $\A \overset\eps\subseteq \B' \cap \M$.
\end{proof}

As a straightforward consequence of \cref{lem:near inclusion commutator 0}, we in turn obtain the following:

\begin{lem}[Near inclusion of commutants]\label{lem:near inclusion commutant}
Let $\C, \D \subseteq B(\mc H)$ be von Neumann algebras with $\C$ hyperfinite.
If $\C \overset{\eps}{\subseteq} \D$, then $\D' \overset{2\eps}{\subseteq} \C'$.
If moreover $\C \subseteq \M$ for a von Neumann algebra~$\M \subseteq B(\mc H)$, then $\D' \cap \M \overset{2\eps}{\subseteq} \C' \cap \M$.
\end{lem}
\begin{proof}
It suffices to prove the second statement, since it reduces to the first if we choose $\M=B(\mc H)$.
Since $\C \overset{\eps}{\subseteq} \D = (\D')'$, the first claim in \cref{lem:near inclusion commutator 0} (with $\A=\C$ and $\B=\D'$) shows that
\begin{align*}
  \norm{[a,b]} \leq 2 \eps \norm a \norm b = 2 \eps \norm a \norm b
\end{align*}
for all $a\in\C$ and $b\in\D'$.
Since $\mathcal C$ is hyperfinite, we can now use the converse in \cref{lem:near inclusion commutator 0} (with $\A=\D' \cap \M$ and $\B=\C \subseteq \M$) to conclude that~$\mathcal D' \cap \M \overset{2\eps}{\subseteq} \mathcal C' \cap \M$.
\end{proof}

We now come to a central and nontrivial result.
For hyperfinite von Neumann algebras, if $\A \overset{\eps}{\subseteq} \B$ for sufficiently small $\eps$, there exists a unitary close to the identity that rotates $\A$ into $\B$.

\begin{restatable}[Near inclusions of subalgebras]{thm}{nearinclusion}\label{thm:near inclusion}
For hyperfinite von Neumann algebras $\A, \B \subseteq B(\mc{H})$ with~$\A \overset{\epsilon}{\subseteq} \B$ for~$\epsilon \leq \frac1{64}$, there exists a unitary $u \in (\A \cup \B)''$ such that $u^* \A u \subseteq \B$ and we have:
\begin{enumerate}
\item  $\norm{I-u} \leq 12 \epsilon$.
\item\label{it:near inclusion small commutator}If~$z \in B(\mc{H})$ satisfies  $\norm{[z,c]} \leq \delta \norm{z} \norm{c}$ for all $c \in \A \cup \B$,  then $\norm{u^*zu-z} \leq  10 \delta \norm{z}$.
\end{enumerate}
Moreover, if $\A_0 \subseteq \A$ is an AF $C^*$-algebra such that $\A_0'' = \A$, then $u$ can be chosen such that also:
\begin{enumerate}[resume]
\item\label{it:near inclusion already close} If~$z \in B(\mc{H})$ satisfies~$z \overset{\delta}{\in} \A_0$ and~$z \overset{\delta}{\in} \B $, then $\norm{u^*zu-z} \leq  16\delta \norm{z}$.
\end{enumerate}

\end{restatable}

This theorem extends Theorem~4.1 of Christensen~\cite{christensen1980near}. 
The first item re-states his result, and we develop the remaining claims.
A self-contained proof appears in \cref{sec:near inclusion appendix}.
Similar stability theorems exist for various other classes of $C^*$-algebras~\cite{christensen1980near, christensen2012perturbations}.
The stability of subalgebra inclusions is closely related to what is often (especially in the context of groups) referred to as \emph{Ulam stability}~\cite{ulam1960collection,burger2013ulam}.
There, one is given a map that ``almost'' satisfies the homomorphism properties, and one asks whether the map can be slightly deformed into a true homomorphism.
See for instance~\cite{johnson1988approximately,park2004approximate} for Ulam stability results on $C^*$-algebras.
The proof of \cref{thm:near inclusion} implicitly involves one such Ulam stability property:
a completely positive map on a hyperfinite von Neumann algebra that is almost a homomorphism is then deformed to a true homomorphism; see e.g.~\cite{johnson1988approximately} more generally.

Using related methods, we also obtain the following useful lemma.
Here, we control the global error between two homomorphisms using the sum of errors on their local restrictions.

\begin{restatable}{lem}{homomorphismerror}\label{lem:homomorphism local error}
Consider two injective weak-$*$ continuous unital $*$-homomorphisms  $\alpha_1, \alpha_2 \colon \A \to \B$ between von Neumann algebras $\A \ \subseteq B(\mc H)$ and $\B \ \subseteq B(\mc K)$, with hyperfinite von Neumann subalgebras $\A_1,\dots,\A_n \subseteq \A$ that pairwise commute, i.e., $[\A_i,\A_j]=0$ for $i \neq j$, and generate $\A$ in the sense that $(\cup_{i=1}^n \A_i)'' = \A$.
Define
\begin{align*}
\epsilon &= \sum_{i=1}^n \norm{(\alpha_1 - \alpha_2)|_{\A_i}}.
\end{align*}
Then if $\epsilon < 1$,
\begin{align*}
\norm{\alpha_1 - \alpha_2}
&\leq 2\sqrt{2} \eps \left( 1 + (1 - \eps^2)^{\frac12} \right)^{-\frac12}
\leq 2 \sqrt 2 \eps,
\end{align*}
where we note that the expression in the middle is $2\epsilon + \mc{O}(\epsilon^2)$.
\end{restatable}

The proof appears in \cref{sec:near inclusion appendix}.
We find the difference between the homomorphisms~$\alpha_1$ and~$\alpha_2$ is controlled by the sum of their local differences.  It appears possible that in general, a tighter bound $\norm{\alpha_1-\alpha_2} \leq \epsilon+ \mc{O}(\eps^2)$ may be correct.
An easy example demonstrates the bound is optimal to within a constant factor.  Let $\A = \A_1  \otimes \cdots \otimes A_n$ for matrix algebras $A_i$, and let $\alpha_1(x)=x$ and $\alpha_2(x)=u^*xu$ for $u=u_1 \otimes \cdots \otimes u_n$, choosing any  $u_i \in U(\A_i)$  with spectrum $\{1,e^{i\frac{\epsilon}{n}}\}$, so that $\norm{u_i - I}= \frac{\eps}{n} +\bigO(\eps^2)$ and hence $\norm{u-I} = \eps+\bigO(\eps^2)$.
Note that by e.g.\ Theorem 26 of \cite{johnston2009computing}, the map on operators $x \mapsto vxv^*-x$ for unitary $v$ has norm given by the diameter of the smallest closed disk containing the spectrum of $v$. For $u_i$ and $u$, that diameter is given by $\frac{\eps}{n}+\bigO(\eps^2)$ and $\eps + \bigO(\eps^2)$, respectively.
Then $\norm{(\alpha_1-\alpha_2)|_{\A_i}}=\frac{\eps}{n}+\bigO((\frac{\eps}{n})^2)$, while $\norm{\alpha_1-\alpha_2}=\epsilon+ \mc{O}(\eps^2)$.

\section{Dynamics on spin systems}\label{sec:spin systems}
We introduce in \cref{subsec:quasi-local} the quasi-local algebra, which is the appropriate $C^*$-algebra to describe a lattice of quantum spin systems.
Next we discuss the celebrated Lieb-Robinson bounds in \cref{sec:alpu def}, give a definition of approximately locality-preserving unitaries, and prove some of their basic properties.

\subsection{The quasi-local algebra}\label{subsec:quasi-local}
If we have a system of a finite number of spins $\CC^{d_1} \ot \cdots \ot \CC^{d_n}$, the corresponding operator algebra is simply the full matrix algebra $\M_{d_1\times d_1} \ot \cdots \ot \M_{d_n\times d_n}$.
However, for infinitely many spins the tensor product structure becomes ambiguous.
If the spins form a lattice the physically appropriate choice of $C^*$-algebra is the \emph{quasi-local algebra}.
Consider a lattice~$\Gamma$, and associate a finite-dimensional matrix algebra $\A_n = \M_{d_n\times d_n}$ to each element $n$ of the lattice.
We assume that there is a uniform upper bound on the dimensions~$d_n$.
For any \emph{finite} subset $X \subseteq \Gamma$ we can define the algebra $\A_{X} = \bigotimes_{n \in X} \A_n$.
These algebras naturally form a local net, meaning that for any two subsets~$X \subseteq X'$ we have a natural inclusion~$\A_X \subseteq \A_{X'}$ (by tensoring with the identity on $X' \setminus X$), and for any two disjoint subsets $X \cap X' = \emptyset$ we have that $[\A_X, \A_{X'}] = 0$ (we embed the two algebras into any~$\A_{X''}$ such that $X \cup X' \subseteq X''$).
This allows us to define the algebra of all \emph{strictly local} operators as
\begin{align*}
  \A_{\Gamma}^{\text{strict}} = \bigcup_{X \subseteq\,\Gamma \text{ finite}} \A_{X}.
\end{align*}
This is a $*$-algebra which inherits a norm from the~$\A_X$, but it is not complete for this norm.
We define the \emph{quasi-local algebra} $\A_{\Gamma}$ to be the norm completion of $\A_{\Gamma}^{\text{strict}}$.
Thus, $\A_{\Gamma}$ is a $C^*$-algebra.
For infinite subsets $X \subseteq \Gamma$, we define $\A_{X}$ correspondingly as a norm-complete $C^*$-subalgebra of~$\A_{\Gamma}$.
Then we have inclusions $\A_{X} \subseteq \A_{\Gamma}$ for any subset $X \subseteq \Gamma$.
If $x \in \A_{X}$, we say $x$ is \emph{supported} on~$X$.

The quasi-local algebra has a natural state~$\tau$, called the \emph{tracial state}, which can be thought of as the generalization of the maximally mixed state to an infinite lattice.
It is defined on $x \in \A_X$ for finite~$X \subseteq \Lambda$ by
\begin{align*}
  \tau(x) = \frac{1}{d_X} \tr(x),
\end{align*}
where $d_X = \prod_{n\in X} d_n$, and can be extended to the full algebra.

We consider the GNS representation $\pi \colon \A_{\Gamma} \rightarrow B(\mc H)$ from \cref{thm:gns} of the quasi-local algebra using the tracial state $\tau$, and we let
\begin{align} \label{eq:II1 factor}
  \A^{\vn}_\Gamma = \pi(\A_{\Gamma})'' \subseteq B(\mc H),
\end{align}
denote the von Neumann algebra generated by the GNS representation of the quasi-local algebra.
The right-hand side is also the weak operator closure of the image $\pi(\A_{\Gamma})$.
It turns out $\A^{\vn}_\Gamma$ is a proper subalgebra of $B(\mathcal{H})$, which remains true even for $\Gamma$ finite; in our case that $\Gamma$ is infinite, $\A^{\vn}_\Gamma$ is the (unique up to unique isomorphism) \emph{hyperfinite type II$_1$ factor}.
This algebra is extensively studied, but for our purpose we will only need to observe that this factor is hyperfinite (as follows directly from its construction).
If $X \subseteq \Gamma$ we denote $\A^{\vn}_X = \pi(\A_X)''$.
Each $\A^{\vn}_X$ is hyperfinite and has the property that $\smash{(\A^{\vn}_X)' \cap \A^{\vn}_{\Gamma} = \A^{\vn}_{\Gamma \setminus X}}$.
Since the tracial state is faithful, we may think of each~$\A_X$ as a subalgebra of~$\A^{\vn}_X$.
Moreover, it holds that $\A^{\vn}_X \cap \A_{\ZZ} = \A_X$ for any $X \subseteq \ZZ$.%
\footnote{Since $\A_{X}^{\vn} = (\A_{\ZZ\setminus X}^{\vn})' \cap \A^{\vn}_\ZZ$, it suffices to show that $(\A_{\ZZ\setminus X}^{\vn})' \cap \A_\ZZ = \A_X$.
We will argue that $(\A_{\ZZ\setminus X}^{\vn})' \cap \A_\ZZ \subseteq \A_X$, since the other inclusion is immediate.
To this end, let $x \in (\A_{\ZZ\setminus X}^{\vn})' \cap \A_\ZZ$.
Since~$x \in \A_\ZZ$, we can choose a sequence $x_i \in \A_{X_i}$ converging to~$x$ in norm and such that each~$X_i$ is a finite set.
On the other hand, $x \in (\A_{\ZZ\setminus X}^{\vn})'$ implies that, for any $y \in \A_{\ZZ\setminus X}$ it holds that
\begin{align}\label{eq:estimate commutator finite support}
  \norm{[y,x_i]} = \norm{[y,x_i - x]} \leq 2\norm{y}\norm{x - x_i}.
\end{align}
Now let $\tilde{x_i} = \int_{U(\A_{X_i \cap (\ZZ \setminus X)})} u x_i u^* \, \d u$, similarly as in \cref{eq:twirling}.
Then $\tilde{x_i} \in \A_{X_i \cap X} \subseteq \A_X$, since it is an element of $\A_{X_i}$ that commutes with $\A_{X_i \cap (\ZZ \setminus X)} = \A_{X_i \setminus X}$.
On the other hand, \cref{eq:estimate commutator finite support} implies (cf.\ \cref{eq:commutator bound,eq:cond exp bounding}) that
\begin{align*}
  \norm{\tilde{x_i} - x_i}
\leq \int_{U(\A_{X_i \cap (\ZZ \setminus X)})} \norm{u x_i u^* - x_i} \, \d u
\leq \int_{U(\A_{X_i \cap (\ZZ \setminus X)})} \norm{[u, x_i]} \, \d u
\leq 2 \norm{x - x_i}.
\end{align*}
Therefore $\norm{\tilde{x_i} - x} \leq 3 \norm{x - x_i} \to 0$ and since each $\tilde{x_i} \in \A_X$, we conclude that $x \in \A_X$.}

The reason for introducing $\A^{\vn}_\Gamma$ is purely to be able to use technical tools, especially \cref{thm:near inclusion}, from the study of von Neumann algebras.
Our main results are all formulated in terms of the quasi-local algebra.

We observe that an automorphism $\alpha$ of $\A_{\Gamma}$ extends naturally to the associated von Neumann algebra in~\eqref{eq:II1 factor}, as follows.
If $\tau$ is the tracial state on the quasi-local algebra $\A_{\Gamma}$, then for any automorphism of $\A_{\Gamma}$ this state is left invariant, i.e., $\tau \circ \alpha = \tau$.
(One way to see this is by using that $\tau$ is the unique state for which $\tau(xy) = \tau(yx)$ for all $x,y \in \A_{\Gamma}$.)
By the uniqueness of the GNS construction this implies that $\alpha$ can be implemented by a unitary $u$ on $\mc H$, in the sense that $\pi(\alpha(x)) = u\pi(x)u^*$.
Therefore, $\alpha$ extends to an automorphism of the hyperfinite von Neumann algebra~$\A^{\vn}_\Gamma$, which we denote by the same symbol $\alpha$ if there is no danger of confusion.
Note that this extension is necessarily unique.

From \cref{sec:gnvw index} onwards we will only consider the situation where $\Gamma = \ZZ$ is the discrete line.
If $X = \{m \in \ZZ \text{ such that } m \leq n\}$ we will write $\A_{\leq n} := \A_X$ and similarly if $X = \{m \in \ZZ \text{ such that } m \geq n\}$ we will write $\A_{\geq n} := \A_X$.
We use the same notation to describe subalgebras $\A^{\vn}_{\leq n}, \A^{\vn}_{\geq n}$ of $\A^{\vn}_{\ZZ}$.

\subsection{QCAs and approximately locality-preserving unitaries}\label{sec:alpu def}
Consider a spin system on a lattice $\Gamma$ with some metric $d$ and the associated quasi-local $C^*$-algebra~$\A_{\Gamma}$.
If $X \subseteq \Gamma$ we will denote
\begin{align*}
  B(X,r) = \{n \in \Gamma \text{ such that } d(n, X) \leq r \}.
\end{align*}

\begin{dfn}[QCA]\label{dfn:lpu}
A \emph{quantum cellular automaton} (QCA) \emph{with radius~$R$} is an automorphism~$\alpha \colon \A_{\Gamma} \rightarrow \A_{\Gamma}$ such that if $x$ is an operator supported on a finite subset $X \subseteq \Gamma$, then $\alpha(x)$ is supported on $B(X,R)$.
We call $R$ a \emph{radius} of the QCA.
\end{dfn}

One of the reasons to study QCAs is that many physical quantum dynamics preserve locality in some form.
However, the locality in \cref{dfn:lpu} is very stringent, and one the most important classes of automorphisms violates strict locality, while preserving a form of approximate locality: evolution by a geometrically local Hamiltonian.
The locality of these evolutions is expressed by so-called \emph{Lieb-Robinson bounds}~\cite{lieb1972}.

We will state a fairly general form of the Lieb-Robinson bounds which also holds for Hamiltonians which are not strictly local, but have a sufficiently fast decay, following e.g.~\cite{nachtergaele2019quasi} or~\cite{hastings2010locality}.
Suppose that $\Gamma$ is a lattice with a metric $d$.
Then a monotonically decreasing function $F\colon \RR_{\geq0} \rightarrow \RR_{\geq 0}$ is called \emph{reproducing} (implying fast decay) if there exists a constant $C>0$ such that for all $n,m \in \Gamma$,
\begin{align*}
  \sum_l F(d(n,l))F(d(l,m)) & \leq CF(d(n,m)),\\
  \sup_y \sum_x F(d(x,y)) & < \infty.
\end{align*}
These conditions are related to a convolution and integral, respectively.
For $\Gamma = \ZZ^D$ with the Euclidean distance, the function $F(r) = (1 + r)^{-(D + \eps)}$ is reproducing for any $\eps > 0$.
Note the reproducing property is not strictly a measure of fast decay: an exponential decay alone is not reproducing, despite having faster decay than the previous power law, because it fails the first inequality.
Meanwhile, $F(r) = (1 + r)^{-(D + \eps)}e^{-ar}$ for any $a>0$ is again reproducing (\cite{nachtergaele2019quasi}, Appendix~8.2).

Now we consider the automorphism $\alpha$ on the quasi-local algebra $\A_{\Gamma}$ which is generated by time evolution for some fixed time $T$ by a Hamiltonian
\begin{align*}
   H = \sum_{n \in \Gamma} H_n + \sum_{X \subseteq \Gamma} H_X.
\end{align*}
The terms $H_n$ act only on site $n$, and the terms $H_X$ act on the sites in $X$.
Then, if the interaction terms $H_X$ have sufficient decay, we have the following bounds on $\alpha(x) = e^{iHt} x e^{-iHt}$.%
\footnote{In fact, one generally needs these bounds to prove that the time evolution defines a dynamics on the quasi-local algebra, i.e.\ that time-evolved quasi-local operators are still quasi-local \cite{nachtergaele2019quasi}.}
We state them without the dependence on the time $t$, which only affects the constant $C$ below, and which is irrelevant for our purposes:
\begin{thm}[Lieb-Robinson~\cite{nachtergaele2019quasi}]\label{thm:lr}
For $\alpha(x) = e^{iHt} x e^{-iHt}$ as above, if $F$ is reproducing and
\begin{align}\label{eq:hamiltonian bound}
  \sup_{n,m \in \Gamma} \sum_{\substack{X\subseteq  \Gamma \\ \text{s.t.\ } n,m \in X}} \frac{\norm{H_X}}{F(d(n,m))} \leq \infty
\end{align}
then there exists a constant $C>0$ such that for all $X, Y \subseteq \Gamma$ and for all $x \in \A_X$, $y \in \A_Y$ we have
\begin{align}\label{eq:lr}
  \norm{[\alpha(x), y]} \leq C\norm{x}\norm{y} \sum_{n \in X} \sum_{m \in Y} F(d(n,m)).
\end{align}
\end{thm}
\noindent
Here, the Hamiltonian is also allowed to be time-dependent, as long as~\eqref{eq:hamiltonian bound} holds uniformly.
See~\cite{nachtergaele2019quasi} for a proof and extensive discussion.

We are particularly interested in the one-dimensional case, where $\Gamma = \ZZ$ and $d(x,y) = \abs{x-y}$ for~$x,y\in\ZZ$.
In that setting we consider the case where $X$ is an \emph{interval} (a finite or infinite sequence of consecutive sites) and $Y$ has bounded distance away from $X$.
A consequence of the Lieb-Robinson bounds in~\eqref{eq:lr} is that certain algebras form near inclusions.

\begin{lem}\label{lem:LR in 1d}
Suppose $\alpha$ is an automorphism of $\A_{\ZZ}$
and suppose there exists a monotonically decreasing function $F\colon\ZZ_{\geq0} \to \RR_{\geq0}$
such that for all $X,Y \subseteq \Gamma$,
\begin{align}\label{eq:LR assumption}
  \norm{[\alpha(x), y]} \leq \norm{x}\norm{y} \sum_{n \in X} \sum_{m \in Y} F(\abs{n-m}).
\end{align}
and suppose $\sum_{n=1}^{\infty} \sum_{m=1}^{\infty} F(n + m) < \infty$.
Then for any (finite or infinite) interval $X \subseteq \ZZ$, we have
\begin{align*}
  \alpha(\A_X) \overset{f(r)}{\subseteq} \A_{B(X,r)},
\end{align*}
where
\begin{align*}
  f(r) = 4\sum_{n, m = 0}^{\infty} F(n + m + r + 1).
\end{align*}
\end{lem}

\begin{proof}
We first prove the near inclusion for finite~$X$.
By \cref{lem:near inclusion commutator 0}, it suffices to show that for any $x \in \A_X$ and any $y \in \A_{B(X,r)^c}$ we have
\begin{align*}
  \norm{[\alpha(x), y]} \leq f(r)\norm{x}\norm{y}
\end{align*}
as in that case $\A^{\vn}_{B(X,r)} = \A_{B(X,r)}$.
By \cref{eq:LR assumption}, we know that
\begin{align*}
  \norm{[\alpha(x), y]} \leq \norm{x}\norm{y} \sum_{n \in X} \sum_{m \in B(X,r)^c} F(\abs{n-m})
\end{align*}
Let
\begin{align*}
  X_k &= \{n \in X \text{  such that } d(n, X^c) = k \}, \\
  Y_l &= \{m \in X^c \text{  such that } d(m, X) = r + l \},
\end{align*}
using the notation $d(n,X)=\min_{x \in X} \abs{n-x}$.
Since $X$ is an interval the size of each of these sets is upper bounded by~2.
We can therefore estimate
\begin{align*}
  \sum_{n \in X} \sum_{m \in B(X,r)^c} F(\abs{n-m}) &\ \leq\  \sum_{k \geq 1} \sum_{l \geq 1} \sum_{n \in X_k} \sum_{m \in Y_l} F(k + l + r - 1)\\
  &\leq 4\sum_{ \geq 1} \sum_{ l \geq 1} F( k +  l + r - 1)\\
  &= f(r).
\end{align*}
We conclude that $\alpha(\A_X) \overset{f(r)}{\subseteq} \A_{B(X,r)}$ for any finite interval~$X$.

If $X$ is infinite and $x \in \A_X$, we can take a sequence $x_i$ such that $\lim_i x_i = x$ in norm and each~$x_i$ is supported on a finite interval inside~$X$.
By what we showed above, for each $i$ there exists some~$y_i \in \A_{B(X,r)}$ such that $\norm{\alpha(x_i) - y_i} \leq f(r)\norm{x}$.
Then
\begin{align*}
  \inf_{y \in \A_{B(X,r)}} \norm{\alpha(x) - y} &\leq \liminf_i \norm{\alpha(x) - y_i}\\
   &\leq \liminf_i\left( \norm{\alpha(x) - \alpha(x_i)} + \norm{\alpha(x_i) - y_i}\right)\\
   &\leq f(r)\norm{x}.
   \qedhere
\end{align*}
\end{proof}

For instance, if $F(r) = \frac{1}{r^{4}}$, then $f(r)=\bigO(\frac{1}{r^2})$; if $F(r)=e^{-ar}\frac{1}{r^2}$ for $a>0$, then $f(r)=\bigO(e^{-ar})$.
As a side note, we observe that one can use \cref{lem:simultaneous near inclusions} on simultaneous near inclusions to show that (in any dimension) Lieb-Robinson type bounds for single-site operators imply bounds for operators on arbitrary sets (which has already been remarked upon in a more restricted setting in~\cite{wilming2020lieb}):
\begin{lem}\label{lem:lr from single site}
Suppose $\alpha$ is an automorphism of the quasi-local algebra $\A_{\Gamma}$ and suppose there exists a function $G \colon \Gamma \times \Gamma \rightarrow \RR_{\geq 0}$ such that for any $n,m \in \Gamma$, $x \in \A_n$ and $y \in \A_m$
\begin{align*}
  \norm{[\alpha(x),y]} \leq \norm{x}\norm{y}G(n,m).
\end{align*}
Then for any finite sets $X, Y \subseteq \Gamma$ and $x \in \A_{X}$, $y \in \A_Y$,
\begin{align*}
  \norm{[\alpha(x), y]} \leq 128\norm{x}\norm{y} \sum_{n \in X} \sum_{m \in Y} G(n,m).
\end{align*}
\end{lem}

\begin{proof}
By assumption and \cref{lem:near inclusion commutator 0} with $\M = \A^{\vn}_\Gamma$ we have
\begin{align*}
  \alpha(\A_n)
\overset{G(n,m)}{\subseteq}
  \A_m' \cap \A^{\vn}_\Gamma = \ 
  \A^{\vn}_{\Gamma \setminus \{m\}}
\end{align*}
for all $m,n \in \Gamma$.
Applying \cref{lem:simultaneous near inclusions} we find that
\begin{align*}
  \alpha(\A_X) \overset{4\sum_{n \in X} G(n,m)}{\subseteq} \A^{\vn}_{\Gamma \setminus \{m\}}
\end{align*}
for all $m \in \Gamma$.
\cref{lem:near inclusion commutant} with $\M = \A^{\vn}_\Gamma$ shows that
\begin{align*}
  \A_m
\overset{8\sum_{n \in X} G(n,m)}{\subseteq}
  \alpha(\mathcal A_X)' \cap \A^{\vn}_\Gamma = \alpha(\mathcal A_X' \cap \A^{\vn}_\Gamma) =\ 
  \alpha(\A^{\vn}_{\Gamma \setminus X}).
\end{align*}
Again applying \cref{lem:simultaneous near inclusions} and \cref{lem:near inclusion commutant} as above yields
\begin{align*}
  \alpha(\A_X) \overset{64\sum_{n \in X, m \in Y} G(n,m)}{\subseteq} \A^{\vn}_{\Gamma \setminus Y}
\end{align*}
which implies the desired commutator bound by \cref{lem:near inclusion commutator 0}.
\end{proof}

Following~\cite{hastings2013classifying} we would like to generalize the notion of a QCA to the case where the automorphism does not preserve strict locality, but only approximate locality.
Such an automorphism is often called quasi-local.
There are various choices of definition that require different decays or dependence on support size; see for instance~\cite{nachtergaele2019quasi}.
For our purpose, the definition should at least include Hamiltonian evolutions satisfying Lieb-Robinson bounds.
We will restrict to the one-dimensional case, where \cref{thm:lr} and \cref{lem:LR in 1d} inspire the following definition:

\begin{dfn}[ALPU in one dimension]\label{dfn:alpu}
An automorphism $\alpha$ of the quasi-local algebra $\A_\ZZ$ is called an \emph{approximately locality-preserving unitary} (ALPU) if for all (possibly infinite) intervals~$X \subseteq \ZZ$ and for all $r\geq0$ we have
\begin{align*}
  \alpha(\A_X) \overset{f(r)}{\subseteq} \A_{B(X,r)}
\end{align*}
for some positive function $f(r)$ with $\lim_{r \to \infty} f(r)=0.$
Here we use the notation in \cref{dfn:near inclusion}.

We say $\alpha$ has $f(r)$-tails when it satisfies the above, or $\bigO(g(r))$-tails if $f(r) = \bigO(g(r))$.
We will always assume, without loss of generality, that $f(r)$ is non-increasing.
\end{dfn}

Note that by definition, if $\alpha$ has $f(r)$-tails, it also has $h(r)$-tails for any function~$h(r)$ with~$h(r)\geq f(r)$ for all~$r$, i.e., $f(r)$ only serves as an upper bound on the spread of $\alpha$.
Furthermore, note that any ALPU has $\littleO(1)$-tails, by definition.

It suffices to check the conditions in \cref{dfn:alpu} either for all finite or for all infinite intervals (see \cref{lem:finite,lem:half infinite} below).
If the above conditions on~$\alpha$ are satisfied for all intervals~$X$ of some fixed size (and arbitrary $r\geq0$), but $f(r)$ decays exponentially, then in fact $\alpha$ is an ALPU with $\bigO(f(r))$-tails by \cref{lem:lr from single site}. We note an equivalent definition of ALPUs when passing to von Neumann algebras in \cref{rmk:ALPU_for_vN}.

\begin{rmk}
In~\cite{hastings2013classifying} what we call an ALPU is simply called a locality-preserving unitary (LPU).
Moreover, there it is said that an automorphism is a \emph{locally generated unitary} (LGU) if it arises from time evolution by some time-dependent Hamiltonian.
We have chosen the more explicit term ALPU instead of LPU, since in the literature the latter has also been used as a synonym for QCA (e.g.\ \cite{csahinouglu2018matrix}).

We note that to call such automorphisms ``unitary'' is perhaps slightly misleading:
there need not be a unitary~$u \in \A_{\ZZ}$ such that $\alpha(x) = u^*xu$ (but there will be a unique unitary implementing $\alpha$ on the GNS Hilbert space with respect to the tracial state, as discussed in \cref{subsec:quasi-local}).
\end{rmk}

\begin{ex}
\cref{lem:LR in 1d} states that for the class of local Hamiltonians in \cref{thm:lr} (Lieb-Robinson), the automorphism $\alpha(x) = e^{iHt}xe^{-iHt}$ is an ALPU at fixed $t$.
It turns out that if the Hamiltonian has \emph{exponentially decaying tails} in the sense that $\norm{H_X} = \bigO(e^{-k\abs{X}})$ decays exponentially with the size of the support $X$, then for any $k' < k$ we may take $f(r)=\bigO(e^{-k'r})$ and $\alpha$ has $\bigO(e^{-k'r})$-tails~\cite{nachtergaele2019quasi,hastings2010locality}.
Such evolutions composed with translations are also ALPUs.
\end{ex}

To use \cref{thm:near inclusion}, we would like to work in the von Neumann algebra $\A^{\vn}_\ZZ$.
However, in the definition of an ALPU we consider an automorphism of $\A_\ZZ$, not  $\A^{\vn}_\ZZ$.
We therefore prove some results allowing us to translate between the tails for automorphisms of $\A_\ZZ$ versus $\A^{\vn}_\ZZ$.

\begin{lem}\label{lem:alpu to vn}
\begin{enumerate}
\item\label{item:alpu to vn} Suppose $\alpha$ is an automorphism of $\A^{\vn}_\ZZ$ and $\alpha(\A_X) \overset{\eps}{\subseteq} \A^{\vn}_{Y}$ for some~$X, Y\subseteq\ZZ$ and~$\eps\geq0$.
Then $\alpha(\A^{\vn}_X) \overset{\eps}{\subseteq} \A^{\vn}_{Y}$.
In particular, any ALPU with $f(r)$-tails extends to an automorphism~$\alpha$ of $\A^{\vn}_X$ such that $\alpha(\A^{\vn}_X) \overset{f(r)}{\subseteq} \A^{\vn}_{B(X,r)}$ for any interval $X$ and $r\geq0$.
\item\label{item:inverse vn auto} Suppose $\alpha$ is an automorphism of $\A_{\ZZ}^{\vn}$ such that $\alpha(\A^{\vn}_X) \overset{\eps}{\subseteq} \A^{\vn}_{B(X,r)}$ for any interval $X$ and some fixed $r,\eps\geq0$.
Then $\alpha^{-1}(\A^{\vn}_X) \overset{4\eps}{\subseteq} \A^{\vn}_{B(X,r)}$ for any interval~$X$.
\item\label{item:alpu from vn} Suppose $\alpha$ is an automorphism of $\A_{\ZZ}^{\vn}$ such that $\alpha(\A_X) \overset{f(r)}{\subseteq} \A^{\vn}_{B(X,r)}$ for any interval $X$ and any~$r\geq0$, where $f(r)$ is a function with $\lim_{r \to \infty} f(r) = 0$.
Then $\alpha$ restricts to an ALPU with $f(r)$-tails.
\item\label{item:inverse alpu} If $\alpha$ is an ALPU with $f(r)$-tails, then $\alpha^{-1}$ is an ALPU with $4f(r)$-tails.
\end{enumerate}
\end{lem}
\begin{proof}

\ref{item:alpu to vn}~Let $x \in \A_X^{\vn}$.
Using the Kaplansky density theorem, choose a net~$x_i \in \A_X$ with~$\norm{x_i} \leq \norm{x}$, converging to $x$ in the weak operator topology, hence also in the weak-$*$ topology (since these topologies are the same on bounded subsets).
Because $\alpha$ is weak-$*$ continuous, $\alpha(x_i)$ converges to~$\alpha(x)$ in the weak-$*$ and hence also in the weak operator topology.
Meanwhile, by our assumption that $\alpha(\A_X) \overset{\eps}{\subseteq} \A^{\vn}_{Y}$, there exist $y_i \in \A^{\vn}_{Y}$ such that $\norm{\alpha(x_i) - y_i} \leq \eps\norm{x}$.
In particular,
$\norm{y_i} \leq (1+\eps) \norm x$.
Hence the net $y_i$ is bounded in norm.
Since norm balls are compact in the weak operator topology,
this implies there must be a converging subnet, which we also denote by~$y_i$.
Denoting the limit of~$y_i$ by~$y$, then $y \in \A^{\vn}_{Y}$, and by lower semi-continuity of the norm in the weak operator topology,
\begin{align*}
  \norm{y - \alpha(x)} \leq \liminf_i \norm{y_i - \alpha(x_i)} \leq \eps\norm{x}.
\end{align*}
This shows that $\alpha(x) \overset{\eps}{\in} \A^{\vn}_{Y}$, proving~\ref{item:alpu to vn}.

\ref{item:inverse vn auto}~Note that $\ZZ \setminus B(X,r)$ is a disjoint union of at most two intervals $Y_1$ and $Y_2$, and $B(Y_i, r) \subseteq \ZZ \setminus X$ for $i=1,2$, so $\alpha(\A^{\vn}_{Y_i}) \overset{\eps}{\subseteq} \A^{\vn}_{\ZZ \setminus X}$.
Then applying~\eqref{eq:simultaneous near inclusion commutant with intersection} in \cref{lem:simultaneous near inclusions} to these two near inclusions with $\M = \A_{\ZZ}^{\vn}$,
\begin{align*}
  \A^{\vn}_X
\ = (\A^{\vn}_{\ZZ \setminus X})' \cap \A_{\ZZ}^{\vn}
&\overset{4\eps}{\subseteq}
  \left( \alpha(\A^{\vn}_{Y_1}) \cup \alpha(\A^{\vn}_{Y_2}) \right)' \cap \A_{\ZZ}^{\vn} =
  (\alpha(\A^{\vn}_{Y_1}))' \cap (\alpha(\A^{\vn}_{Y_2}))' \cap \A_{\ZZ}^{\vn}  \\
&\ \,=\, \ 
  \alpha((\A^{\vn}_{Y_1})' \cap \A_{\ZZ}^{\vn}) \cap \alpha((\A^{\vn}_{Y_2})' \cap \A_{\ZZ}^{\vn}) =
  \alpha((\A^{\vn}_{Y_1})' \cap (\A^{\vn}_{Y_2})' \cap \A_{\ZZ}^{\vn})  \\
&\ \,=\, \alpha(\A^{\vn}_{\ZZ \setminus Y_1} \cap \A^{\vn}_{\ZZ \setminus Y_2}) =
\alpha(\A^{\vn}_{B(X,r)})
\end{align*}
and the conclusion follows by applying $\alpha^{-1}$.

\ref{item:alpu from vn}~We need to show that if $x \in \A_{\ZZ}$, then $\alpha(x) \in \A_{\ZZ}$.
First consider $x$ strictly local, on some finite interval $X$.
Then by assumption there is a sequence $y_r \in \A^{\vn}_{B(X,r)} = \A_{B(X,r)}$ such that $\norm{\alpha(x) - y_r} \leq f(r)\norm{x}$.
Hence, $y_r$ is a sequence of strictly local operators converging in norm to $\alpha(x)$ and hence $\alpha(x) \in \A_{\ZZ}$.
If $x \in \A_{\ZZ}$ is not strictly local, let $x_i$ be a sequence of strictly local operators converging in norm to $x$.
Then $\alpha(x_i) \in \A_{\ZZ}$ and $\alpha(x_i)$ converges in norm to $\alpha(x)$.
Similarly, $\alpha^{-1}$ maps $\A_{\ZZ}$ into $\A_{\ZZ}$ (using that \ref{item:inverse vn auto} implies locality bounds for $\alpha^{-1}$) and hence we conclude that $\alpha$ restricts to an automorphism of $\A_{\ZZ}$.

This implies the desired result, as then
\begin{align*}
  \alpha(\A_X) \overset{f(r)}{\subseteq} \A^{\vn}_{B(X,r)} \cap \A_{\ZZ} = \A_{B(X,r)}
\end{align*}
using the general fact that $\A^{\vn}_{Y} \cap \A_{\ZZ} = \A_{Y}$ for any $Y \subseteq \ZZ$.

\ref{item:inverse alpu}~By \ref{item:alpu to vn}, $\alpha$ extends to an automorphism of $\A^{\vn}_{\ZZ}$ with~$f(r)$-tails, by \ref{item:inverse vn auto} the inverse of this extension has $4f(r)$-tails, and by \ref{item:alpu from vn} the restriction of the latter is an ALPU with $4f(r)$-tails.
\end{proof}

Recall that any automorphism~$\alpha$ of the quasi-local algebra~$\A_\ZZ$ extends uniquely to an automorphism of the von Neumann algebra~$\A_\ZZ^{\vn}$, which we denote by the same symbol~$\alpha$.
Then \cref{lem:alpu to vn}\ref{item:alpu to vn} and~\ref{item:alpu from vn} together allow an equivalent definition of ALPUs with $f(r)$-tails using the von Neumann algebra, rather than using the quasi-local algebra as in \cref{dfn:alpu}.  We summarize below.

\begin{rmk} \label{rmk:ALPU_for_vN}
  Any automorphism $\alpha : \A_\ZZ \to \A_\ZZ$ with $f(r)$-tails, i.e.\  that satisfies $\alpha(\A_X) \overset{f(r)}{\subseteq} \A_{B(X,r)}$ for all intervals~$X$ and~$r \geq 0$, uniquely extends to an automorphism $\alpha : \A^{\vn}_\ZZ \to \A^{\vn}_\ZZ$ that has the same tails, i.e.\ that satisfies $\alpha(\A^{\vn}_X) \overset{f(r)}{\subseteq} \A^{\vn}_{B(X,r)}$.  Conversely, any $\alpha : \A^{\vn}_\ZZ \to \A^{\vn}_\ZZ$ satisfying the latter (for all~$X$ and~$r$) restricts to an ALPU $\alpha : \A_\ZZ \to \A_\ZZ$ with $f(r)$-tails.  Hence we may identify an ALPU $\alpha$ with its extension to $\A_\ZZ^{\vn}$ and refer to the latter also as an ``ALPU with $f(r)$-tails.''
\end{rmk}

We can use \cref{lem:alpu to vn} to show that in \cref{dfn:alpu} we may in fact restrict to either only finite intervals or only half-infinite intervals, as shown by the following lemmas.

\begin{lem}\label{lem:finite}
Suppose $\alpha$ is an automorphism of $\A_\ZZ$ such that $\alpha(\A_X) \overset{f(r)}{\subseteq} \A_{B(X,r)}$ for any \emph{finite} interval $X \subseteq \ZZ$ and any $r\geq0$, where $f(r)$ is a positive function with $\lim_{r \to \infty} f(r) = 0$.
Then $\alpha$ is an ALPU with $f(r)$-tails.
\end{lem}
\begin{proof}
As explained earlier, we can extend~$\alpha$ to an automorphism of $\A_\ZZ^{\vn}$ (denoted again by $\alpha$).
Let $X\subseteq\ZZ$ be an infinite interval, $r>0$, and $0\neq x\in \A_X$.
We first show that for any~$\delta>0$,
\begin{align}\label{eq:half infinite bound}
  \alpha(x) \overset{(1+\delta)f(r)}{\in} \A_{B(X,r)}.
\end{align}
By definition of $\A_\ZZ$, we can approximate~$x$ with a sequence~$x_i \to x$ converging in norm, with~$x_i \in \A_{X_i}$, where each $X_i \subseteq X$ is a finite interval.
Then $\alpha(x_i) \to \alpha(x)$ converges in norm as well, and
\begin{align*}
  \inf_{y \in \A_{B(X,r)}} \norm{y-\alpha(x)}
&\leq \liminf_i \biggl( \inf_{y \in \A_{B(X,r)}}   \norm{y-\alpha(x_i)} + \norm{\alpha(x)-\alpha(x_i)} \biggr) \\
&\leq \liminf_i \biggl( f(r) \norm{x_i} + \norm{\alpha(x)-\alpha(x_i)} \biggr)
= f(r)\norm{x},
\end{align*}
so for any $\delta'>0$ there exists some $y \in \A_{B(X,r)}$ such that $\norm{y - \alpha(x)} \leq f(x) \norm x + \delta'$.
If we apply this with $\delta' = f(x) \delta \norm x$ then \cref{eq:half infinite bound} follows.
Next, we claim that
\begin{align}\label{eq:vN bound}
  \alpha(x) \overset{f(r)}{\in} \A_{B(X,r)}^{\vn}.
\end{align}
Indeed, by \cref{eq:half infinite bound} we can for any~$n>0$ take some~$y_n \in \A_{B(X,r)}$ such that $\norm{\alpha(x) - y_n} \leq (1+\frac1n) f(r)$.
In particular, $y_n$ is a bounded sequence in $\A_{B(X,r)}^{\vn}$.
Since norm balls are compact in the weak operator topology, there is a subsequence~$y_{n_i}$ converging to some $y \in \A_{B(X,r)}^{\vn}$, and
\begin{align*}
  \norm{\alpha(x) - y} \leq \liminf_i \norm{\alpha(x) - y_{n_i}} \leq \liminf_i \left( 1+\tfrac1{n_i} \right) f(r) = f(r).
\end{align*}
Thus we have proved \cref{eq:vN bound}.
As a consequence, we have for any interval $X\subseteq\ZZ$ and any $r\geq0$,
\begin{align*}
  \alpha(\A_X) \overset{f(r)}{\subseteq} \A_{B(X,r)}^{\vn}.
\end{align*}
Now the lemma follows from \cref{lem:alpu to vn}\ref{item:alpu from vn}.
\end{proof}

\begin{lem}\label{lem:half infinite}
Suppose $\alpha$ is an automorphism of $\A^{\vn}_{\ZZ}$ such that $\alpha(\A_{\leq n}) \overset{f(r)}{\subseteq} \A^{\vn}_{\leq n +r}$ and $\alpha(\A_{\geq n}) \overset{f(r)}{\subseteq} \A^{\vn}_{\geq n -r}$ for any $n \in \ZZ$ and $r\geq0$, where $f(r)$ is a positive function with $\lim_{r \to \infty} f(r) = 0$.
Then $\alpha$ restricts to an ALPU with $8f(r)$-tails.
\end{lem}
\begin{proof}
By \ref{item:alpu from vn} of \cref{lem:alpu to vn} we only need to show that for any finite interval $X = \{n,n+1,\ldots,n+m\}$ it holds that
\begin{align*}
	\alpha(\A_{X}) \overset{8f(r)}{\subseteq} \A^{\vn}_{B(X,r)}.
\end{align*}
Now, by \ref{item:alpu to vn} of \cref{lem:alpu to vn} we have
\begin{align*}
	\alpha(\A^{\vn}_{\leq n + m}) &\overset{f(r)}{\subseteq} \A^{\vn}_{\leq n +m + r}, \\
	\alpha(\A^{\vn}_{\geq n}) &\overset{f(r)}{\subseteq} \A^{\vn}_{\geq n -r},
\end{align*}
hence we obtain by taking commutants and applying \cref{lem:near inclusion commutant} with $\M = \A^{\vn}_\ZZ$ that
\begin{align*}
	\A^{\vn}_{\geq n + m + r + 1} &\overset{2f(r)}{\subseteq} \alpha(\A^{\vn}_{\leq n + m})' \cap \A^{\vn}_{\ZZ} = \alpha(\A^{\vn}_{\geq n +m + 1}) \subseteq \alpha(\A^{\vn}_{\ZZ \setminus X}),\\
	\A^{\vn}_{\leq n - r - 1} &\overset{2f(r)}{\subseteq} \alpha(\A^{\vn}_{\geq n})' \cap \A^{\vn}_{\ZZ} = \alpha(\A^{\vn}_{\leq n - 1}) \subseteq \alpha(\A^{\vn}_{\ZZ \setminus X}).
\end{align*}
By \cref{eq:simultaneous near inclusion commutant with intersection} of \cref{lem:simultaneous near inclusions} it follows that
\begin{align*}
  \alpha(\A_X) = \alpha(\A^{\vn}_X) = \alpha(\A^{\vn}_{\ZZ \setminus X})' \cap \A^{\vn}_{\ZZ}
\overset{8f(r)}{\subseteq} \left( \A^{\vn}_{\geq n + m + r + 1} \cup \A^{\vn}_{\leq n - r - 1} \right)' \cap \A^{\vn}_{\ZZ}
= \A_{B(X,r)}^{\vn}. 
\end{align*}
\end{proof}

If we consider an ALPU, we may coarse-grain the lattice by grouping together (or `blocking') sites.
This yields again an ALPU, but with faster decaying tails.
In particular, for any fixed $\eps > 0$, we can always coarse-grain by sufficiently large blocks of sites so that on the coarse-grained lattice, $\alpha(\A_X) \overset{\eps}{\subseteq} \A_{B(X,1)}$ for any interval $X$.
This motivates the following definition:

\begin{dfn}[$\eps$-nearest neighbor automorphism in one dimension]\label{dfn:nn}
An automorphism $\alpha$ of~$\A_\ZZ$ is called \emph{$\eps$-nearest neighbor} for some $\eps \geq 0$ if for any (finite or infinite) interval~$X \subseteq \ZZ$ we have
\begin{align}\label{eq:nn cs}
  \alpha(\A_X) \overset{\eps}{\subseteq} \A_{B(X,1)}.
\end{align}
If $\alpha$ is an automorphism of~$\A_\ZZ^{\vn}$ we instead require the weaker condition that
\begin{align}\label{eq:nn vn}
  \alpha(\A_X) \overset{\eps}{\subseteq} \A_{B(X,1)}^{\vn}
\end{align}
for all intervals $X \subseteq \ZZ$.
Note that \cref{eq:nn vn} is equivalent to $\alpha(\A_X^{\vn}) \overset{\eps}{\subseteq} \A_{B(X,1)}^{\vn}$ by \cref{lem:alpu to vn}\ref{item:alpu to vn}.
\end{dfn}

If an automorphism of~$\A_\ZZ^{\vn}$ extends an automorphism of~$\A_\ZZ$, as will usually be the case for us, then \cref{eq:nn cs,eq:nn vn} are equivalent, since $\A_{B(X,1)}^{\vn} \cap \A_{\ZZ} = \A_{B(X,1)}$ for any~$X\subseteq\ZZ$. As such, \cref{dfn:nn} is unambiguous.

\section{Index theory of one-dimensional QCAs revisited}\label{sec:gnvw index}
In this section we discuss the index theory of QCAs in one dimension.
First, in \cref{sec:review gnvw} we review the definition and some of the most important properties of the GNVW index as proven in~\cite{gross2012index}. In \cref{sec:entropic index qcas} we provide an alternative formula for the index in terms of a difference of mutual informations.
In \cref{sec:robust index} we prove some results about QCAs in one dimension which are locally close to each other. These results are interesting in their own right, but will also be crucial when extending the index to ALPUs.

\subsection{The structure of one-dimensional QCAs and the GNVW index}\label{sec:review gnvw}
One-dimensional QCAs have a beautifully simple structure theory, which we will now review.
The material in this section is based on~\cite{gross2012index}, which we recommend for a more extensive discussion.
The same material is also covered in the review~\cite{farrelly2019review}. The discussion below refers only to QCAs, serving as a warm-up for the case of ALPUs.

Suppose that $\alpha$ is a nearest-neighbor QCA, which we may assume without loss of generality after blocking sites.
Let
\begin{align}\label{eq:B and C regions}
\begin{split}
\B_n &= \A_{\{2n,2n+1\}} \\
\C_n &= \A_{\{2n - 1,2n\}}
\end{split}
\end{align}
be algebras on pairs of adjacent sites; with $\B_n$ and $\C_n$ corresponding to pairs staggered by one.
In particular, $\alpha(\B_n) \subseteq \C_n \ot \C_{n+1}$.
Define
\begin{align}\label{eq:support algebras}
\begin{split}
\L_n & = \alpha(\B_n) \cap \C_n  \\
  \R_n &  = \alpha(\B_n) \cap \C_{n+1}.
\end{split}
\end{align}
See \cref{fig:gnvw} as a mnemonic.
These are manifestly algebras, but naively they might be trivial.
Instead, it turns out that they provide factorizations of $\C_n$ and $\B_n$.
Using the notation $\M \ot \N := (\M \cup \N)''$ for finite-dimensional mutually commuting subalgebras $\M, \N \subset \A_\ZZ$, one has the following result.

\begin{thm}[Factorization~\cite{gross2012index}]\label{thm:gnvw}
\begin{align}
\label{eq:C_n factorization}
  \C_n &:= \A_{\{2n - 1,2n\}} = \L_n \ot \R_{n-1} \\
\label{eq:B_n factorization}
\B_n & := \A_{\{2n,2n+1\}} = \alpha^{-1}(\L_n) \ot \alpha^{-1}(\R_n).
\end{align}
\end{thm}
\noindent
Thus $\alpha^{-1}(\L_n)$ is the part of $\B_n$ that $\alpha$ sends to the left, and  $\alpha^{-1}(\R_n)$ is the part of $\B_n$ that $\alpha$ sends to the right.
Likewise, $\C_n$ is composed of a part $\L_n$ that was sent leftward from $\B_n$, and a part~$\R_{n-1}$ that was sent rightward from $\B_{n-1}$.
\begin{proof}
Recall from \cref{eq:twirling} that in general, for a finite-dimensional subalgebra $\M \subset \A_\ZZ$, we have the conditional expectation $\EE_{\M'}(x) = \int_{U(\M)} uxu^* \, \d u.$
We first show
\begin{align}
  \L_n & := \alpha(\B_n) \cap \C_n = \EE_{\C_{n+1}'}(\alpha(\B_n)) \label{eq:factorization 1} \\
  \R_{n-1} & := \alpha(\B_{n-1}) \cap \C_n = \EE_{\C_{n-1}'}(\alpha(\B_{n-1})) \label{eq:factorization 2} .
\end{align}
Clearly, $\L_n \subseteq \EE_{\C_{n+1}'}(\alpha(\B_n))$.
To show the reverse inclusion, let $y = \EE_{\C_{n+1}'}(\alpha(x))$ for some~$x \in \B_n$, i.e.
\begin{align*}
  y = \EE_{\C_{n+1}'}(\alpha(x)) = \int_{U(\C_{n+1})} u \alpha(x) u^* \, \d u.
\end{align*}
From this expression, we see $[y,\alpha(\B_{n-1})] = 0$ because $[\alpha(x),\alpha(\B_{n-1})]=0$ and $[\C_{n+1},\alpha(\B_{n-1})]=0$ (the latter because $\alpha(\B_{n-1}) \subseteq \C_{n-1} \ot \C_n$).
On the other hand, it follows from $\alpha(\B_n) \subseteq \C_n \ot \C_{n+1}$ that $y \in \C_n$.
Moreover, $\alpha^{-1}$ is again a nearest neighbor QCA, so we have $\alpha^{-1}(\C_n) \subset \B_{n-1} \ot \B_n$, so we find that $y \in \alpha(\B_{n-1} \ot \B_n)$.
Then $[y,\alpha(\B_{n-1})] = 0$ implies $y \in \alpha(\B_n)$.  We conclude \cref{eq:factorization 1} holds; a similar argument shows \cref{eq:factorization 2}.

Finally we demonstrate $\C_n \subseteq \L_n \ot \R_{n-1}$, which then becomes an equality. For any $c \in \C_n$, we can express $\alpha^{-1}(c)\in \B_{n-1} \ot \B_n$ as $\alpha^{-1}(c)=\sum_i a_i b_i$ for some elements $a_i \in \B_{n-1}, b_i \in \B_n$.  Then
\begin{align*}
c =\EE_{\C_{n-1}'} \EE_{\C_{n+1}'} (c) = \sum_i \EE_{\C_{n-1}'} \EE_{\C_{n+1}'}(\alpha(a_i)\alpha(b_i)) = \sum_i  \EE_{\C_{n-1}'} (\alpha(a_i))  \EE_{\C_{n+1}'}(\alpha(b_i)) \in \L_n \ot \R_{n-1},
\end{align*}
as desired.
The final equality follows from $\alpha(a_i) \in \C_{n-1} \ot \C_n$ and $\alpha(b_i) \in \C_n \ot \C_{n+1}$, and the final inclusion is manifest from \cref{eq:factorization 1,eq:factorization 2}.
Thus we have proved \cref{eq:C_n factorization}.

Noting again that $\alpha^{-1}$ is a nearest neighbor QCA, similar logic applied to $\alpha^{-1}$ yields \cref{eq:B_n factorization}.
Specifically, \cref{eq:factorization 1,eq:factorization 2} are replaced by
\begin{align*}
  \alpha^{-1}(\L_n) &= \B_n \cap \alpha^{-1}(\C_n) = \EE_{\B'_{n-1}}(\alpha^{-1}(\C_n))), \\
  \alpha^{-1}(\R_n) &= \B_n \cap \alpha^{-1}(\C_{n+1}) = \EE_{\B'_{n+1}}(\alpha^{-1}(\C_{n+1}))),
\end{align*}
which follow using $\alpha^{-1}(\C_n) \subseteq \B_{n-1} \ot \B_n$, $\alpha^{-1}(\C_{n+1}) \subseteq \B_n \ot \B_{n+1}$, and $\alpha(\B_n) \subseteq \C_n \ot \C_{n+1}$, and one uses this to prove the nontrivial inclusion $\B_n \subseteq \alpha^{-1}(\L_n)) \ot \alpha^{-1}(\R_n)$.
\end{proof}

For later use in \cref{sec:alpu approximation}, below we note \cref{thm:gnvw} also holds for weaker assumptions, by an identical argument.

\begin{rmk}\label{rem:weaker locality properties}
Although  in \cref{thm:gnvw} we assumed the automorphism $\alpha$ was a QCA, the only locality properties of $\alpha$ required to achieve $\C_n = \L_n \ot \R_{n-1}$ were $\alpha(\B_{n-1}) \subseteq \C_{n-1} \ot \C_{n}$,  $\alpha(\B_{n}) \subseteq \C_{n} \ot \C_{n+1}$, and $\alpha^{-1}(\C_n) \subseteq \B_{n-1} \ot \B_n$.
Similarly, to achieve $\B_n = \alpha^{-1}(\L_n) \ot \alpha^{-1}(\R_n)$ we need only $\alpha^{-1}(\C_n) \subseteq \B_{n-1} \ot \B_n$, $\alpha^{-1}(\C_{n+1}) \subseteq \B_{n} \ot \B_{n + 1}$, and $\alpha(\B_n) \subseteq \C_n \ot \C_{n+1}$.
\end{rmk}

\begin{figure}
\centering
\raisebox{1cm}{
\begin{overpic}[width=0.5\textwidth,grid=false]{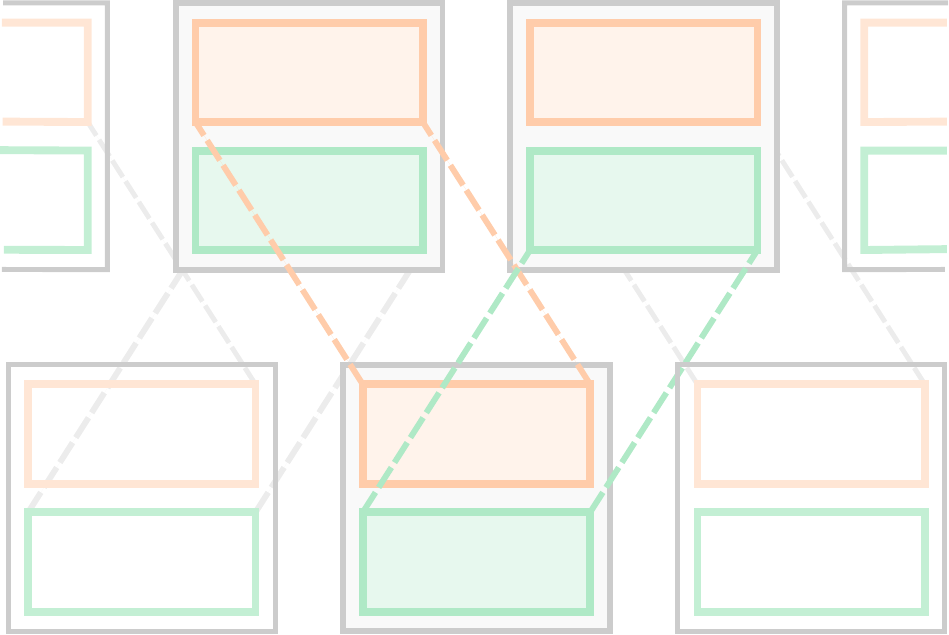}
\put(30,58){$\L_n$} \put(65,44){$\R_n$}
\put(28,44){\color{gray}{$\R_{n-1}$}} \put(64,58){\color{gray}{$\L_{n+1}$}}
\put(30,69){$\C_n$} \put(64,69){$\C_{n+1}$}
\put(42,20){$\alpha^{-1}(\L_n)$} \put(42,6){$\alpha^{-1}(\R_n)$}
\put(47,30){$\B_n$}
\put(42.5,-5){\scriptsize{$2n$}} \put(51.5,-5){\scriptsize{$2n+1$}}
\put(3,-5){\scriptsize{$2n-2$}} \put(17,-5){\scriptsize{$2n-1$}}
\put(74,-5){\scriptsize{$2n+2$}} \put(87,-5){\scriptsize{$2n+3$}}
\end{overpic}}
\quad
\caption{The factorization \cref{thm:gnvw} decomposes every two-site algebra into a left-moving and right-moving part.}
\label{fig:gnvw}
\end{figure}
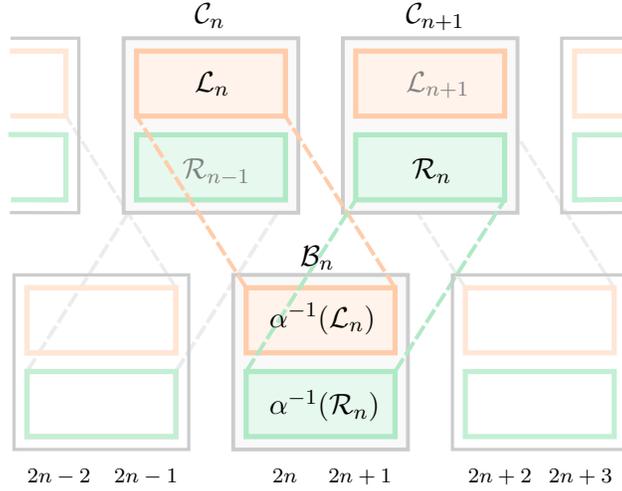

Based on \cref{thm:gnvw} we can show that the ratio of $\dim(\L_n)$ and $\dim(\A_{2n})$ is independent of~$n$, motivating the following definition:

\begin{dfn}[Index of QCA]
Suppose $\alpha$ is a one-dimensional nearest neighbor QCA.
Let $\L_n$ and $\R_n$ be defined as in \eqref{eq:support algebras}, then the \emph{index} of $\alpha$ is given by
\begin{align}\label{eq:LPU index}
  \ind(\alpha) &:= \frac12\bigl(\log(\dim(\L_n)) - \log(\dim(\A_{2n}))\bigr)\\
  &= \frac12\bigl(\log(\dim(\A_{2n + 1})) - \log(\dim(\R_n))\bigr). \nonumber
\end{align}
\end{dfn}

\noindent
The value of $\ind(\alpha)$ is independent of the choice of $n$ \cite{gross2012index}.
We choose to take the logarithm of the original definition.
The index of a QCA with radius~$R > 1$ may be defined by blocking sites such that the resulting QCA is nearest neighbor, and one can show that the index is independent of the choice of blocking.
This index can be thought of as a `flux', measuring the difference between how much quantum information is flowing to the right vs.\ left.
From the definition it is clear it cannot take arbitrary values, but is restricted to integer linear combinations $\ZZ[\{\log(p_i)\}]$ where the $p_i$ are all prime factors of local Hilbert space dimensions $d_n$.

The index can be used to characterize all one-dimensional QCAs.
In order to do so, we introduce two types of QCAs: circuits and shifts.
We will say a QCA $\alpha$ is a \emph{block partitioned unitary} if it can be written as
\begin{align*}
  \alpha(x) = \left(\prod_j u_j^* \right) x \left(\prod_j u_j \right)
\end{align*}
where the $u_j$ are a family of local unitaries, the $u_j$ having disjoint and finite support.
We will say $\alpha$ is a \emph{circuit} (in~\cite{gross2012index} a similar notion is called \emph{locally implementable}) if it can be written as a composition of block partitioned unitaries where each local unitary is supported on a uniformly bounded finite set.
In one dimension any circuit QCA of radius $R$ can be written as a composition of at most two block partitioned unitaries where each local unitary is supported on at most $2R$ contiguous sites.
We denote by $\sigma^k_d$ the \emph{translation} QCA which has local Hilbert space dimension $d$ and which translates any operator by $k$ sites, mapping $\sigma_d^k(\A_n) = \A_{n-k}$.
Here $k$ can be negative.
We will say a QCA is a \emph{shift} if it is a tensor product of translations of the form $\sigma^k_d$.

\begin{thm}[Properties of GNVW index~\cite{gross2012index}]\label{thm:properties index}
Let $\alpha$, $\beta$ be one-dimensional QCAs. Then:
\begin{enumerate}
\item $\ind(\alpha \ot \beta) = \ind(\alpha) + \ind(\beta)$
\item If $\alpha$ and $\beta$ are defined on the same quasi-local algebra (i.e., with the same local dimensions), $\ind(\alpha\beta) = \ind(\alpha) + \ind(\beta)$.
\item $\alpha$ is a circuit if and only if $\ind(\alpha) = 0$.
\item $\ind(\sigma^k_d) = k\log(d)$.
\item Every one-dimensional QCA is a composition of a shift and a circuit.%
\footnote{Strictly speaking this only makes sense if all the local dimensions are equal. We can always achieve this by taking a tensor product with the identity automorphism on a quasi-local algebra with appropriate local dimensions.}
\item If $\alpha$ and $\beta$ are defined on the same quasi-local algebra the following are equivalent: \label{item:GNVW-index-equivalence}
\begin{enumerate}
  \item $\ind(\alpha) = \ind(\beta)$.
  \item There exists a circuit $\gamma$ such that $\alpha = \beta\gamma$. \label{item:circuits}
  \item There exists a strongly continuous path from $\alpha$ to $\beta$ through the space of QCAs with a uniform bound on the radius.  \label{item:deformations}
  \item There exists a \emph{blending} of $\alpha$ and $\beta$, meaning a QCA $\gamma$ which is identical to $\alpha$ on a region extending to left infinity and equal to $\beta$ on a region extending to right infinity.   \label{item:blending}
\end{enumerate}
\end{enumerate}
\end{thm}

\noindent The ``classification'' of one-dimensional QCAs refers to the set of QCAs modulo an equivalence relation, given either by circuits~\ref{item:circuits}, continuous deformations~\ref{item:deformations}, or blending~\ref{item:blending}.  These equivalence classes are identical and characterized by the index, as expressed in~\ref{item:GNVW-index-equivalence}.  If $\alpha$ and $\beta$ are not defined on the same quasi-local algebra (i.e. have different local dimensions), analogous statements to~\ref{item:circuits},~\ref{item:deformations},~\ref{item:blending} hold after separately tensoring $\alpha$ and $\beta$ with appropriate identity automorphisms, i.e. adding extra tensor factors on which they act trivially, such that $\alpha$ and $\beta$ then have the same local dimensions.  The notion of equivalence between QCAs that further allows extra tensor factors is called ``stable equivalence,'' discussed in \cite{freedman2020classification}. We will prove generalizations of all these properties for ALPUs.

As observed in~\cite{gross2012index}, the tensor product property together with the normalization on shifts and circuits completely determines the index.
\begin{lem}\label{lem:unique index}
Suppose $I$ assigns a real number $I(\alpha)$ to any one-dimensional QCA $\alpha$ such that
\begin{enumerate}
\item\label{it:I multiplicative} $I(\alpha \ot \beta) = I(\alpha) + I(\beta)$ for all one-dimensional QCAs $\alpha$ and $\beta$.
\item\label{it:same on circ and trans} $I$ takes the same values as $\ind$ on circuits and on $\sigma^k_d$.
\end{enumerate}
Then $I(\alpha) = \ind(\alpha)$ for any one-dimensional QCA $\alpha$.
\end{lem}
\begin{proof}
Let $\alpha$ be any one-dimensional QCA and let $\beta$ be a shift with $I(\beta) = \ind(\beta) = -\ind(\alpha)$, using~\ref{it:same on circ and trans}.
Then $I(\alpha \ot \beta) = I(\alpha) + I(\beta) = I(\alpha) - \ind(\alpha)$ by~\ref{it:I multiplicative}.
On the other hand, $\ind(\alpha \ot \beta) = 0$ so it is a circuit.
Again by property~\ref{it:same on circ and trans} this implies that $I(\alpha \ot \beta) = 0$, showing that $I(\alpha) = \ind(\alpha)$.
\end{proof}

\subsection{An entropic definition of the GNVW index}\label{sec:entropic index qcas}
Here we provide a new formula for the index in terms of the mutual information, which can also be defined for infinite $C^*$-algebras. This reformulation is not strictly necessary to develop an index theory for ALPUs, but it does allow us to give a clean expression for the index of an ALPU.

We consider two copies of the quasi-local algebra $\A_{\ZZ}$.
Then the tensor product $\A_{\ZZ} \ot \A_{\ZZ}$ is uniquely defined as a $C^*$-algebra since $\A_{\ZZ}$ is nuclear (so there is no ambiguity in the norm completion of the tensor product).
We choose a transposition on each local algebra, which gives rise to a transposition $x \mapsto x^T$ on $\A_{\ZZ}$.
Let $\tau$ be the tracial state on $\A_{\ZZ}$.
Then we define the \emph{maximally entangled state} $\omega$ by
\begin{align}\label{eq:max ent state}
  \omega(x \ot y) = \tau(xy^T)
\end{align}
for $x \ot y \in \A_{\ZZ} \ot \A_{\ZZ}$.
It is not hard to see that if we restrict to a finite number of sites, $\omega$ indeed restricts to the usual maximally entangled state.
Then we define
\begin{align*}
  \phi = (\alpha^{\dagger} \ot \id)(\omega).
\end{align*}
where $\id$ is the identity channel, and $\alpha^{\dagger}$ is the adjoint channel.
In other words, $\phi$ is the \emph{Choi state} of $\alpha$.

Split the algebra $\A_{\ZZ}$ at any point $n$ in the chain, letting
\begin{equation}\label{eq:left right chain algebra}
\begin{aligned}
  \A_L &:= \A_{\leq n} \\
  \A_R &:= \A_{> n}.
\end{aligned}
\end{equation}
and similarly split the copy as $\A_{L'}$ and $\A_{R'}$.
For a QCA with radius $r$, we will also consider
\begin{equation}\label{eq:finite subalgebras index}
\begin{aligned}
  \A_{L_1} &= \A_{n - r + 1,\ldots,n},&
  \quad \A_{L_2} &= \A_{\leq n - r}, \\
  \A_{R_1} &= \A_{n+1, \ldots, n + r},&
  \quad \A_{R_2} &= \A_{\geq n + r + 1}.
\end{aligned}
\end{equation}

We will define the index in terms of a difference of mutual informations of the Choi state.
If $\phi$, $\psi$ are states on a $C^*$-algebra we may define the relative entropy $S(\phi,\psi)$~\cite{ohya2004quantum}.
The mutual information of a state $\phi$ on $\A_A \ot \A_B$ is then defined as $I(A : B)_\phi = S(\phi, \phi|_{\A_A} \ot \phi|_{\A_B})$.
On finite dimensional subsystems this definition coincides with the usual one.
The only property we need is that relative entropies, and hence mutual informations, on the full algebra can be computed as limits:

\begin{prop}[Proposition 5.23 in~\cite{ohya2004quantum}]\label{prop:mi vN}
Let $\A$ be a $C^*$-algebra and let $\{\A_i\}_i$ be an increasing net of $C^*$-subalgebras so that $\cup_i \A_i$ is dense in $\A$. Then for any two states $\phi$, $\psi$ on $\A$ the net $S(\phi_i, \psi_i)$ converges to $S(\phi,\psi)$ where $\phi_i = \phi|_{\A_i}$, $\psi_i = \psi|_{\A_i}$.
\end{prop}

\begin{figure}
\centering
\begin{overpic}[width=0.35\textwidth,grid=false]{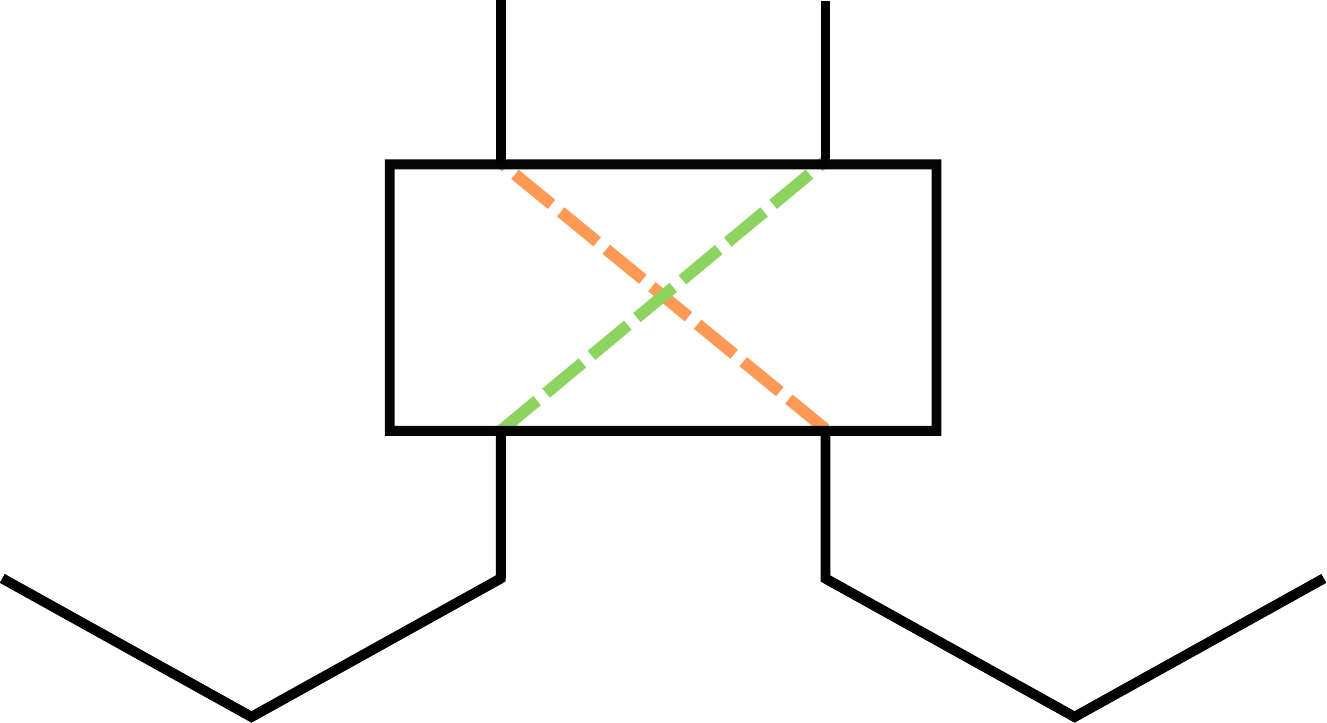}
\put(27,49){$L$} \put(67,49){$R$}
\put(-2,15){$L'$} \put(98,15){$R'$}
\put(72,30){$\alpha$}
\end{overpic}
\caption{Illustration of~\eqref{eq:mi index}.  The index measures the difference in information flows, left to right minus right to left, as $\ind(\alpha) = \frac12\left(I(L':R)_{\phi} - I(L:R')_{\phi}\right)$.}
\label{fig:mi index}
\end{figure}
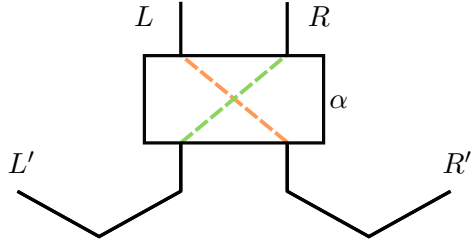

\begin{prop}\label{prop:mi formula index}
For any choice of $n$ in \cref{eq:left right chain algebra} the index of a one-dimensional QCA~$\alpha$ is~given~by
\begin{align}\label{eq:mi index}
  \ind(\alpha) = \frac12\left(I(L':R)_{\phi} - I(L:R')_{\phi}\right).
\end{align}
For a QCA with radius $r$, this can also be computed locally as
\begin{align}\label{eq:mi index local}
  \ind(\alpha) = \frac12\left(I(L_1':R_1)_{\phi} - I(L_1:R_1')_{\phi}\right).
\end{align}
Here, the mutual information terms are computed with respect to the corresponding subalgebras of~$\A_\ZZ \ot \A_\ZZ$ (with primed systems corresponding to subalgebras of the second factor).
\end{prop}

\begin{proof}
Denote by $I(\alpha)$ the expression in~\eqref{eq:mi index}.
First we will argue that $I(L':R)_{\phi} = I(L_1':R_1)$ and $I(L:R')_{\phi} = I(L_1:R_1')$.
One sees this by verifying that
\begin{align*}
  \phi_{L'R} &= \phi_{L_1' R_1} \ot \tau_{L_2' R_2}\\
  \phi_{LR'} &= \phi_{L_1 R_1} \ot \tau_{L_2 R_2'}
\end{align*}
where the $\tau$ denote tracial (i.e\ maximally mixed) states.
Next, to see that $I(\alpha) = \ind(\alpha)$ we will apply \cref{lem:unique index}.
From the definition it is clear that $I(\alpha \ot \beta) = I(\alpha) + I(\beta)$, so it suffices to compute $I(\alpha)$ for a circuit and a shift.
For a shift $\alpha = \sigma^k_d$ it is clear from the definition that for positive $k$
\begin{align*}
  I(L':R)_{\phi} &= 2k\log(d)\\
  I(L:R')_{\phi} &= 0
\end{align*}
and for negative $k$
\begin{align*}
  I(L':R)_{\phi} &= 0\\
  I(L:R')_{\phi} &= 2k\log(d).
\end{align*}
Finally, for a circuit $\alpha$, notice that we can ignore any unitaries that act only on $L$ or $R$ as they keep the mutual information invariant.
In this way, we may also reduce to the finite subsystem $L_1 R_1 L_1' R_1'$.
In order to see that $I(\alpha) = 0$ we thus only need to check that
\begin{align*}
  I(L_1':R_1)_{\phi} = I(L_1:R_1')_{\phi}
\end{align*}
where $\ket{\phi} = U \ot I \ket{\omega}$ for some unitary $U$ acting on $L_1R_1$ and where $\ket{\omega}$ is a maximally entangled state between $L_1R_1$ and $L_1' R_1'$.
In that case $\ket{\phi}$ is a maximally entangled state  between~$L_1R_1$ and~$L_2 R_2$ and
\begin{align*}
  S(L_1')_\phi &= S(L_1)_\phi \\
  S(R_1')_\phi &= S(R_1)_\phi \\
  S(L_1'R_1)_\phi &= S(L_1 R_1')_\phi.
\end{align*}
The first two equalities hold because $\phi$ is maximally entangled, and the third equality holds because~$\phi$ is pure.
Thus we see that
\begin{align*}
  I(L_1':R_1)_{\phi} &= S(L_1')_{\phi} + S(R_1)_{\phi} - S(L_1'R_1)_{\phi}\\
  &= S(L_1)_{\phi} + S(R_1')_{\phi} - S(L_1 R_1')_{\phi}\\
  &= I(L_1:R_1')_{\phi}. \qedhere
\end{align*}
\end{proof}

The expression of the index in~\eqref{eq:mi index} is intuitive: $I(L':R)_{\phi}$ and $I(L:R')_{\phi}$ measure the flow of information to the right and left respectively.
Notice that depending on the choice of cut $I(L':R)_{\phi}$ and $I(L:R')_{\phi}$ can vary individually, but the total \emph{flux} as defined by~\eqref{eq:mi index} is invariant.
One reason this expression for the index is useful is that, contrary to the original definition, it is plausibly well-defined for automorphisms which are not strictly local (or for channels which are not automorphisms).
In \cref{thm:index alpu after reorg} we will show that taking the limit of the finite subalgebras in~\eqref{eq:finite subalgebras index} with increasing radius gives a well-defined and finite limit for any ALPU with appropriately decaying tails, and hence using \cref{prop:mi vN} we conclude that both mutual information terms in~\eqref{eq:mi index} are finite and \eqref{eq:mi index} gives a finite, quantized answer also for an ALPU.

In~\cite{gross2012index}, a similar numerical expression for the index is provided in terms of overlaps of algebras (their Eq.~45).
In fact, their formula (or rather its logarithm) can be interpreted as~\eqref{eq:mi index} but with the entropies replaced by Renyi-2 entropies,
\begin{align*}
  \ind(\alpha) = \frac12\left(I_2(L':R)_{\phi} - I_2(L:R')_{\phi}\right),
\end{align*}
where $I_2(A:B)_{\rho} := S_2(A)_{\rho} + S_2(B)_{\rho} - S_2(AB)_{\rho}$.
While the values of the individual mutual information terms depend on the choice of Renyi-2 or von Neumann entropy, for QCAs, the difference of mutual informations used to define the index does \emph{not} depend on this choice, and in the proof of \cref{prop:mi formula index} one can simply replace the entropies~$S$ by Renyi entropies~$S_2$.
However, the mutual information has better continuity properties with respect to the dimension of the local Hilbert spaces compared to the Renyi-2 mutual information (compare the following with the continuity bound in Lemma 12 of~\cite{gross2012index}):

\begin{thm}[Continuity of mutual information \cite{alicki2004continuity,winter2016tight,wilde2013quantum}]\label{thm:continuity mi}
Suppose $\rho, \sigma$ are states on $\mc H_A \ot \mc H_B$, and $\frac12\norm{\rho_{AB} - \sigma_{AB}}_1 \leq \eps < 1$.
Then
\begin{align*}
  \abs[\big]{I(A:B)_{\rho} - I(A:B)_{\sigma}} \leq 3\eps \log(d_A) + 2(1 + \eps) h\left(\frac{\eps}{1+\eps}\right) \leq 3\eps \log(d_A) +  \eps \log \tfrac1\eps
\end{align*}
where $d_A = \dim(\mc H_A)$ and $h(x)=-x\log(x)-(1-x)\log(1-x)$ is the binary entropy.  
\end{thm}

This continuity is important for the extension to ALPUs, where we need to compute the approximation to the index on a sequence of increasing finite subalgebras.  In that case, the indices defined using the Renyi-2 and von Neumann entropies give different answers when restricted to the finite subalgebras.
A final remark is that~\eqref{eq:mi index local} can also be rewritten as an entropy difference
\begin{align}\label{eq:entropy diff}
  \ind(\alpha) &= \frac{1}{2} \left( I(L_1':R_1)_{\phi} - I(L_1:R_1')_{\phi} \right) \nonumber \\
  &= \frac{1}{2}\left( S(L_1 R'_1)_\phi - S(L'_1 R_1)_\phi \right).
\end{align}
However, the extension of this expression to infinite-dimensional setting is less clear, because both terms diverge.

\subsection{Robustness of the GNVW index}\label{sec:robust index}
Because the index can be computed locally, it appears that two QCAs with different index should be easy to distinguish locally.
We make this quantitative in \cref{prop:nearby-qcas}: two QCAs which look locally similar must have equal index.
We begin with a cruder but more general estimate, describing how the mutual information of the Choi state varies continuously with respect to the automorphism that defines it.
This estimate applies to general automorphisms which may not be QCAs, proving useful in the argument for \cref{thm:index alpu after reorg}.

Let $\alpha$ be an automorphism of $\A_\ZZ$.
Even when $\alpha$ is not a QCA, we can mimic the local definition of the index in~\eqref{eq:mi index local} using finite disjoint regions~$L,R$.
We denote this quantity~$\widetilde{\ind}_{L,R}(\alpha)$ to emphasize $\alpha$ may not be a QCA nor even an ALPU,
\begin{align}\label{eq:tilde index}
  \widetilde{\ind}_{L,R}(\alpha) = \frac12\left(I(L':R)_\phi - I(L:R')_\phi\right),
\end{align}
where the mutual information terms are computed with respect to the corresponding subalgebras of~$\A_\ZZ \ot \A_\ZZ$. (As above, primed systems refer to the second copy of~$\A_{\ZZ}$.)
Clearly, \cref{eq:tilde index} only depends on the restriction of the Choi state to $\A_X \ot \A_{X'}$, where $X = L \cup R$, i.e.\ on the state~$\widetilde\phi_{XX'} := \phi|_{\A_X \ot \A_{X'}}$, which is given by
\begin{align*}
  \widetilde\phi_{XX'}(x) = \omega\left((\alpha \ot \id)(x)\right)
\end{align*}
for all $x \in \A_X \ot \A_{X'}$.
Then we have the following continuity estimate.

\begin{lem}\label{lem:crude index continuity}
For two automorphisms $\alpha_1$ and $\alpha_2$ of $\A_\ZZ$ with maximum local dimension $d$, the quantity $\widetilde{\ind}_{L,R}$ in~\eqref{eq:tilde index} obeys
\begin{align*}
  \abs[\big]{\widetilde{\ind}_{L,R}(\alpha_1) - \widetilde{\ind}_{L,R}(\alpha_2)} = \bigO\bigl(\eps|X|\log(d) + \eps \log\tfrac1\eps\bigr),
\end{align*}
where $\eps = \norm{(\alpha_1-\alpha_2)|_{\A_{X}}}$.
The same continuity estimate with respect to $\alpha_1$ and $\alpha_2$ holds for the individual terms in \eqref{eq:tilde index}.
\end{lem}
\begin{proof}
First we compare the restricted Choi states $\widetilde\phi_{XX',1}$ and $\widetilde\phi_{XX',2}$ of $\alpha_1$ and $\alpha_2$, respectively.
For any $x \in \A_X \ot \A_{X'}$ with $\norm{x}=1$,
\begin{align*}
  \abs[\big]{\widetilde\phi_{XX',1}(x) - \widetilde\phi_{XX',2}(x)}_1
= \abs[\big]{\omega((\alpha_1 \ot \id - \alpha_2 \ot \id)(x))}
\leq \norm{(\alpha_1 \ot \id - \alpha_2 \ot \id)|_{\A_X \ot \A_{X'}}}.
\end{align*}
Thus the trace distance between the two Choi states is bounded by
\begin{align*}
\norm{\widetilde\phi_{XX',1} - \widetilde\phi_{XX',2}}_1
\leq \norm{(\alpha_1 \ot \id - \alpha_2 \ot \id)|_{\A_X \ot \A_{X'}}}
\leq 2\eps + \bigO(\eps^2)
\end{align*}
using \cref{lem:homomorphism local error} for the last inequality (with $\A_1 = \A_X \ot I$, $\A_2 = I \ot \A_{X'}$ and $\A = \A_X \ot \A_{X'}$ finite-dimensional and $\B = \A_\ZZ^{\vn}$).
The conclusion follows from the continuity of mutual information in \cref{thm:continuity mi} with respect to the state, noting the region $X$ has associated Hilbert space of dimension at most $d^{|X|}$.
\end{proof}

If $\alpha_1$ and $\alpha_2$ are one-dimensional QCAs of radius $r$, then because the index takes discrete values, there exists~$\eps_0$ such that if $\eps \leq \frac{\eps_0}{r\log(d)}$ then $\ind(\alpha_1) = \ind(\alpha_2)$.  However, we can do better and eliminate the dependence on the local dimension, as a simple application of \cref{thm:near inclusion}. By blocking sites, we may assume without loss of generality that the QCA is nearest neighbor.

\begin{prop}[Robustness of GNVW index for QCAs]\label{prop:nearby-qcas}
Suppose $\alpha_1$ and $\alpha_2$ are two nearest-neighbor QCAs defined on the same quasi-local algebra $\A_{\ZZ}$ such that $\norm{(\alpha_1 - \alpha_2)|_{\A_{\{2n,2n+1\}}}} \leq \eps$ for some $n$ with $\eps \leq \frac{1}{192}$.  Then $\ind(\alpha_1) = \ind(\alpha_2)$.

Moreover, the algebras $\L^{(1)}_n$ and $\L^{(2)}_n$ defined by~\eqref{eq:support algebras} using $\alpha_1$ and $\alpha_2$ respectively are isomorphic, with isomorphism implemented by a unitary~$u \in \A_{\{2n-1,2n\}}$ satisfying $\norm{u - I} \leq 36\eps$.
\end{prop}

\noindent
Note that when working with a coarse-grained QCA, where each site is composed of many smaller sites, the hypotheses like $\norm{(\alpha_1 - \alpha_2)|_{\A_{\{2n,2n+1\}}}} \leq  \eps$ constraining error on coarse-grained sites may always be replaced by hypotheses constraining the sum of errors on fine-grained sites, using \cref{lem:homomorphism local error}.
(In other words, upper bounds for errors on small regions control errors on larger regions.)

\begin{proof}
By the structure theory for QCAs in \cref{thm:gnvw} there exist algebras $\L^{(i)}_n$, $\R^{(i)}_{n-1}$ for $i = 1,2$ defined as in~\eqref{eq:support algebras} that satisfy
\begin{align*}
    \A_{\{2n-1,2n\}} = \L^{(i)}_n \ot \R^{(i)}_{n-1}
\end{align*}
To prove that $\ind(\alpha_1) = \ind(\alpha_2)$, by~\eqref{eq:LPU index} it suffices to show that $\L_n^{(1)}$ and $\L_n^{(2)}$ are isomorphic.
To see the isomorphism, take $x \in \L^{(1)}_n$ with $\norm{x} = 1$ and let $y = \alpha_2(\alpha_1^{-1}(x))$.
Then $\norm{x - y} = \norm{\alpha_1(\alpha_1^{-1}(x)) - \alpha_2(\alpha_1^{-1}(x))} \leq \eps$ using the assumption $\norm{(\alpha_1 - \alpha_2)|_{\A_{\{2n,2n+1\}}}} \leq \eps$ and noting $\alpha_1^{-1}(x) \in \A_{\{2n,2n+1\}}$ since $x \in \L^{(1)}_n$.
Using the conditional expectation from~\eqref{eq:twirling}, define
\begin{align*}
z = \EE_{\A_{\{2n+1,2n+2\}}'}(y) = \int_{U(\A_{\{2n+1,2n+2\}})} uyu^* \, \d u
\end{align*}
such that $z \in \L^{(2)}_n$ by the characterization of $\L_n$ in~\eqref{eq:factorization 1}.
Note $\norm{[a,y]} = \norm{[a,y-x]} \leq 2\eps\norm{a}$ for all~$a\in\A_{\{2n+1,2n+2\}}$, so by its definition $z$ satisfies $\norm{y-z} \leq 2\eps$, and $\norm{x - z} \leq \norm{x - y} + \norm{y - z} \leq 3\eps$.
We conclude $\L^{(1)}_n \overset{3\eps}{\subseteq} \L^{(2)}_n$, and by a symmetric argument we see $\L^{(2)}_n \overset{3\eps}{\subseteq} \L^{(1)}_n$.
By \cref{thm:near inclusion}, noting that $3\eps \leq \frac1{64}$, we obtain that $\L^{(1)}_n$ and $\L^{(2)}_n$ are isomorphic, and the isomorphism is implemented by a unitary $u \in \A_{\{2n-1,2n\}}$ with $\norm{u - I} \leq 36\eps$.
\end{proof}

For later use in \cref{sec:alpu approximation}, below we build on \cref{rem:weaker locality properties} to note that \cref{prop:nearby-qcas} also holds for weaker assumptions, by an identical argument.
\begin{rmk} \label{rmk:weaker locality properties 2}
Although in \cref{prop:nearby-qcas} we assumed the automorphisms $\alpha_1$ and $\alpha_2$ were QCAs, the only locality properties required to achieve the isomorphism between $\L^{(1)}_n$ and $\L^{(2)}_n$ are the properties listed in \cref{rem:weaker locality properties} as those required to achieve $\A_{\{2n-1,2n\}} = \L^{(i)}_n \ot \R^{(i)}_{n-1}$ for $i = 1,2$.
More explicitly, we only require that $\alpha_i(\A_{\{2n-2,2n-1\}}) \subseteq \A_{\{2n-3,\dots,2n\}},$  $\alpha_i(\A_{\{2n,2n+1\}}) \subseteq \A_{\{2n-1,\dots,2n+2\}}$, and $\alpha_i^{-1}(\A_{\{2n-1,2n\}}) \subseteq \A_{\{2n-2,\dots,2n+1\}}$ for $i=1,2$.
\end{rmk}

This also allows us to confirm the intuition that a one-dimensional QCA which is locally close to the identity can be implemented locally with unitaries close to the identity.

\begin{prop}\label{prop:lpu close to identity}
Suppose $\alpha$ is a one-dimensional QCA with radius $R$ and suppose that for~$\eps \leq \frac{1}{192}$ we have $\norm{\alpha(x) - x} \leq \eps\norm{x}$ for any~$x\in\A_{\ZZ}$ supported on at most $2R$ sites.
Then $\alpha$ can be implemented as a composition of two block partitioned unitaries $u = \prod_n u_n$ and $v = \prod_n v_n$, i.e.
\begin{align*}
  \alpha(x) = v^*u^*xuv
\end{align*}
with each of the unitaries $u_n, v_n$ acting on $2R$ adjacent sites and satisfying 
\begin{align*}
  \norm{u_n - I} = \bigO(\eps), \quad \norm{v_n - I} = \bigO(\eps).
\end{align*}
\end{prop}

\begin{proof}
By blocking sites in groups of $R$ sites, we may assume without loss of generality that~$\alpha$ is nearest neighbor.
Let $\alpha_1 = \id$ and $\alpha_2 = \alpha$ in \cref{prop:nearby-qcas}.
Clearly, $\L_n^{(1)} = \A_{2n}$ and~$\R_{n-1}^{(1)} = \A_{2n-1}$.
By \cref{prop:nearby-qcas} there exists some $v_n \in \A_{\{2n - 1,2n\}}$ such that $v_n \L_n^{(2)} v_n^* = \A_{2n}$ with $\norm{v_n - I} = \bigO(\eps)$.
It follows that $v_n \R_{n-1}^{(2)} v_n^* = \A_{2n - 1}$.
Let $v = \prod v_n$ and let $\tilde{\alpha} = v \alpha(x)v^*$.
Then $\tilde{\alpha}(\A_{\{2n,2n+1\}}) = \A_{\{2n,2n+1\}}$.
Moreover, for all $x\in\A_{\{2n,2n+1\}}$, we estimate
\begin{align*}
  \norm{\tilde{\alpha}(x) - x} &\leq \norm{v\alpha(x)v^* - \alpha(x)} + \norm{\alpha(x) - x}\\
  &\leq 2\norm{v_n \ot v_{n+1} - I}\norm{x} + \eps\norm{x}\\
  &\leq 2(\norm{v_n - I} + \norm{v_{n+1}- I})\norm{x} + \eps\norm{x}\\
  &= \bigO(\eps)\norm{x}.
\end{align*}
Then \cref{prp:inner automorphism} shows that $\tilde{\alpha}|_{\A_{\{2n,2n+1\}}}$ can be implemented by a unitary $u_n$ on $\A_{\{2n,2n+1\}}$ with $\norm{u_n - I} = \bigO(\eps)$.
\end{proof}

\section{Index theory of approximately locality-preserving unitaries in one dimension}\label{sec:alpu approximation}
In this section we develop the index theory of ALPUs in one dimension.
Just like in the rest of the paper, all ALPUs will be one-dimensional.

For a general ALPU~$\alpha$, we show in \cref{thm:qca approx,thm:index alpu after reorg} that there always exist an approximation of~$\alpha$ by a sequence of QCAs~$\beta_j$.
We can use the limit of the indices of the latter (which become stationary for large~$j$) as the definition of the index of~$\alpha$.
If~$\alpha$ has $\bigO(r^{-(1+\delta})$-tails for some~$\delta>0$, we further show that this index can be computed as a difference of mutual informations, 
\begin{align}\label{eq:mi index for alpus with enough decay}
  \ind(\alpha) = \frac12\left(I(L':R)_{\phi} - I(L:R')_{\phi}\right),
\end{align}
with both terms being finite, just like we saw in \cref{eq:mi index} for QCAs.
The local computation of the index in~\eqref{eq:mi index local} does not yield the exact index for ALPUs.
However, the exact index can still be computed locally;
we show that on sufficiently large regions, the local index computation gives the exact answer when rounded to the nearest value in the fixed set of discrete index values.

In the remainder of the section, we discuss the properties of this index.
We find that once circuits are replaced by evolutions by time-dependent Hamiltonians, the results of~\cite{gross2012index} stated in \cref{thm:properties index} generalize in a natural way.
Our results are summarized in \cref{thm:properties index alpu}.

\subsection{Approximating an ALPU by a QCA}
We sketch the general strategy for approximating an ALPU $\alpha$ by a QCA.
We first develop a method for deforming $\alpha$ into an ALPU $\alpha_n$ that behaves as a QCA with a strict causal cone in the proximity of the site $n$, exhibited by \cref{prop:single patch localized} and \cref{fig:single patch localized}.
In \cref{prp:approximation almost nn} we then we stitch the different~$\alpha_n$ together into a QCA~$\beta$ using the structure theory for one-dimensional QCAs, obtaining a QCA approximation to $\alpha$. If we apply this result to increasingly coarse-grained lattices, in \cref{thm:qca approx} we obtain a sequence of QCAs of increasing radius that approximate~$\alpha$ with increasing accuracy.

To achieve \cref{prop:single patch localized} localizing an ALPU $\alpha$ on a local patch, we compose $\alpha$ with a sequence of unitary rotations.
Some individual rotation steps are described by \cref{lem:single rotation 0} and \cref{lem:single rotation 1}, with proof illustrated in  \cref{fig:single rotation}.
Each step uses \cref{thm:near inclusion} to rotate nearby subalgebras, e.g.\ rotating an algebra $\alpha(\A_X)  \overset{\eps}{\subseteq}  \A_Y$ to obtain an exact inclusion.
We start with these two lemmas.
\cref{lem:single rotation 0}, \cref{lem:single rotation 1} and \cref{prop:single patch localized} are each divided into two parts, (i) and~(ii).
In each case, part~(i) is valid for $\eps$-nearest neighbour automorphisms (which need not be ALPUs), while part~(ii) gives a more refined statement when assuming an ALPU as input.
For the majority of the further development in this paper, in fact only the parts~(i) will be necessary, and so the first-time reader may wish to skip part~(ii) of these results, as well as the supporting \cref{lem:adjusted alpha is alpu}.
Those parts will only be necessary for later results about blending, following \cref{dfn:blending}.

\begin{lem}\label{lem:single rotation 0}
\begin{enumerate}
\item There exist universal constants $C_0, \eps_0 > 0$ such that if $\alpha$ is an $\eps$-nearest neighbor automorphism of $\A^{\vn}_\ZZ$ with $\eps \leq \eps_0$ and
\begin{align*}
  \alpha(\A^{\vn}_{\geq n}) \subseteq \A^{\vn}_{\geq n-1}
\end{align*}
for some site $n \in \ZZ$, then there exists an automorphism of $\A^{\vn}_\ZZ$ of the form $\tilde{\alpha}(x) = u^*\alpha(x)u$ for some unitary $u \in \A^{\vn}_{\geq n - 1}$ with $\norm{u - I} \leq C_0\eps$ and
\begin{align}
    \tilde{\alpha}(\A^{\vn}_{\leq n-1}) &\subseteq \A^{\vn}_{\leq n}, \label{eq:single rotation 0 a} \\
 \tilde{\alpha}(\A^{\vn}_{\geq n}) &\subseteq \A^{\vn}_{\geq n-1}. \label{eq:single rotation 0 b}
\end{align}
\item If additionally $\alpha$ is an ALPU with $f(r)$-tails, we can take $u$ such that $\tilde{\alpha}$ is an ALPU with $\bigO(f(r-1))$-tails and such that we have, for $r\to\infty$,
\begin{align}
\label{eq:asymptotically equal 0 left}
  \norm{(\alpha - \tilde{\alpha})|_{\A^{\vn}_{\leq n-r-1}}} &= \bigO(f(r)), \\
\label{eq:asymptotically equal 0 right}
  \norm{(\alpha - \tilde{\alpha})|_{\A^{\vn}_{\geq n+r}}} &= \bigO(f(r-1)) \\
\intertext{and, for all $x \in \A_{\geq n+r+1},$}
\label{eq:asymptotically equal 0 on image}
	\norm{u^*xu - x} &\ = \bigO(f(r)\norm{x}).
\end{align}
\end{enumerate}
\end{lem}
\begin{proof}
(i)~Note $\alpha^{-1}$ is $4\eps$-nearest neighbor by~\ref{item:inverse vn auto} of \cref{lem:alpu to vn}.  Thus $\alpha^{-1}(\A^{\vn}_{\geq n+1})  \overset{4\eps}{\subseteq} \A^{\vn}_{\geq n}$ and then $\A^{\vn}_{\geq n+1}  \overset{4\eps}{\subseteq} \alpha(\A^{\vn}_{\geq n})$.
By \cref{thm:near inclusion} with $\A_0 = \A_{\geq n+1}$, $\A = \A_0'' = \A^{\vn}_{\geq n+1}$ and $\B = \alpha(\A^{\vn}_{\geq n})$, provided that $\eps \leq \frac{1}{256}$, there exists a unitary $u \in ( \A^{\vn}_{\geq n+1} \cup  \alpha(\A^{\vn}_{\geq n}) )''$ such that
\begin{align*}
u \A^{\vn}_{\geq n+1} u^* \subseteq  \alpha(\A^{\vn}_{\geq n})
\end{align*}
and $\norm{u - I} \leq 48\eps.$   Because  $\alpha(\A^{\vn}_{\geq n}) \subseteq \A^{\vn}_{\geq n-1}$, we also have $u \in \A^{\vn}_{\geq n-1}$.
We define a new automorphism $\tilde{\alpha}(x) = u^* \alpha(x)u$ that satisfies
$\tilde{\alpha}(\A^{\vn}_{\geq n})  \supseteq \A^{\vn}_{\geq n+1}$,
and then satisfies~\eqref{eq:single rotation 0 a} by taking commutants.
Moreover,
\begin{align*}
\tilde{\alpha}(\A^{\vn}_{\geq n}) = u^* \alpha(\A^{\vn}_{\geq n})u \subseteq u^*\A^{\vn}_{\geq n-1}u = \A^{\vn}_{\geq n-1}
\end{align*}
using the assumption $\alpha(\A^{\vn}_{\geq n}) \subseteq \A^{\vn}_{\geq n-1}$ and  the fact $u \in \A^{\vn}_{\geq n-1}$. Then $\tilde{\alpha}$ also satisfies~\eqref{eq:single rotation 0 b}.

(ii)~Now we further assume $\alpha$ is an ALPU with $f(r)$-tails and show~\eqref{eq:asymptotically equal 0 left}.
By our use of \cref{thm:near inclusion} to construct $u \in ( \A^{\vn}_{\geq n+1} \cup  \alpha(\A^{\vn}_{\geq n}) )''$, we know that for $x \in \A^{\vn}_\ZZ$,
\begin{align*}
\norm{[x,y]} \leq \delta \norm{x}\norm{y} \;\; \forall y \in  \A^{\vn}_{\geq n+1}\cup  \alpha(\A^{\vn}_{\geq n})\; \implies \; \norm{u^*xu-x} = \mc{O}(\delta \norm{x}).
\end{align*}
For $r\geq0$ and $x \in \alpha(\A^{\vn}_{\leq n-r-1})$, the above condition is satisfied for $\delta=2 f(r+1)$ because $x \overset{f(r+1)}{\in} \A^{\vn}_{\leq n}$ (using \cref{lem:near inclusion commutator 0}) and $x \in \alpha(\A^{\vn}_{\geq n})'$, so $\norm{(\alpha - \tilde{\alpha})|_{\A^{\vn}_{\leq n-r-1}}} = \bigO(f(r+1))$ and \cref{eq:asymptotically equal 0 left} follows.

Our application of \cref{thm:near inclusion} also implies that for $x \in \A^{\vn}_\ZZ$,
\begin{align}\label{eq:part iii in five one}
x \overset{\delta}{\in}  \A_{\geq n+1}\; \textrm{ and } \; x \overset{\delta}{\in}   \alpha(\A^{\vn}_{\geq n}) \; \implies \; \norm{u^*xu - x} = \mc{O}(\delta \norm{x}).
\end{align}
For $r\geq1$ and $x \in \alpha(\A_{\geq n+r})$, those conditions are satisfied for $\delta=f(r-1)$ because $x \overset{f(r-1)}\in \A_{\geq n+1}$ and~$x \in \alpha(\A^{\vn}_{\geq n})$, so $\norm{(\alpha - \tilde{\alpha})|_{\A_{\geq n+r}}} = \bigO(f(r-1))$ and hence \cref{eq:asymptotically equal 0 right} follows.
%

Next, we prove \cref{eq:asymptotically equal 0 on image}.
Recall that by \cref{lem:alpu to vn}\ref{item:inverse alpu}, $\alpha^{-1}$ is also an ALPU with $\bigO(f(r))$-tails.
Therefore, for any $r\geq0$ and $x \in \A_{\geq n+r+1}$ we have $\alpha^{-1}(x) \overset{f(r+1)}\in \A_{\geq n}$, hence $x \overset{f(r+1)}\in \alpha(\A_{\geq n}^{\vn})$, and now \cref{eq:part iii in five one} shows that $\norm{u^*xu - x} = \bigO(f(r+1)\norm{x})$ and \cref{eq:asymptotically equal 0 on image} follows.

Finally we show that $\tilde\alpha$ is an ALPU with $\bigO(f(r-1))$-tails.
This follows from \cref{lem:adjusted alpha is alpu} below and the fact that
\begin{align*}
  \norm{uxu^* - x} = \bigO(f(r-1)\norm{x})
\end{align*}
holds for the following $x$ and all $r\geq0$: for $x \in \A_{\leq n - r - 2}$ since $u \in \A^{\vn}_{\geq n - 1}$, for $x \in \A_{\geq n + r + 1}$ by \cref{eq:asymptotically equal 0 on image}, for $x \in \alpha(\A_{\leq n - r - 1})$ by \cref{eq:asymptotically equal 0 left} and for $x \in \alpha(\A_{\geq n + r})$ by \cref{eq:asymptotically equal 0 right}.
\end{proof}

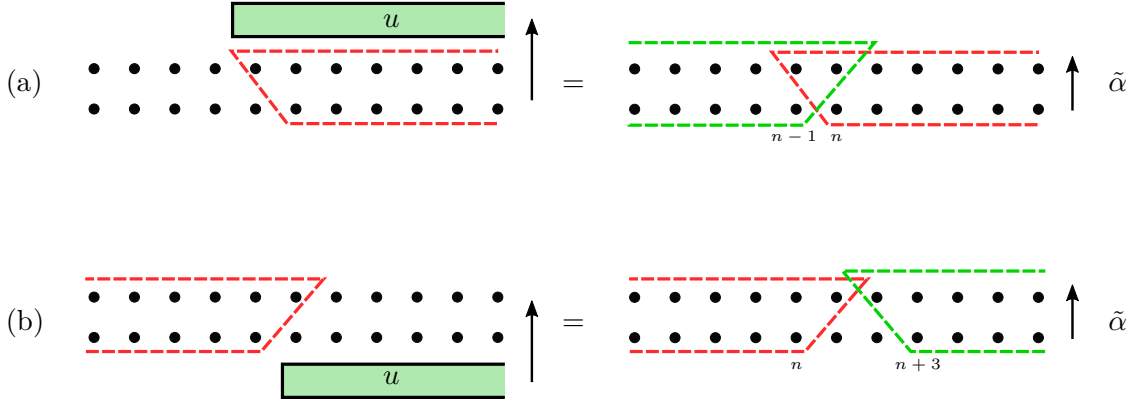
\begin{figure}
\centering
\begin{overpic}[width=0.8\textwidth,grid=false]{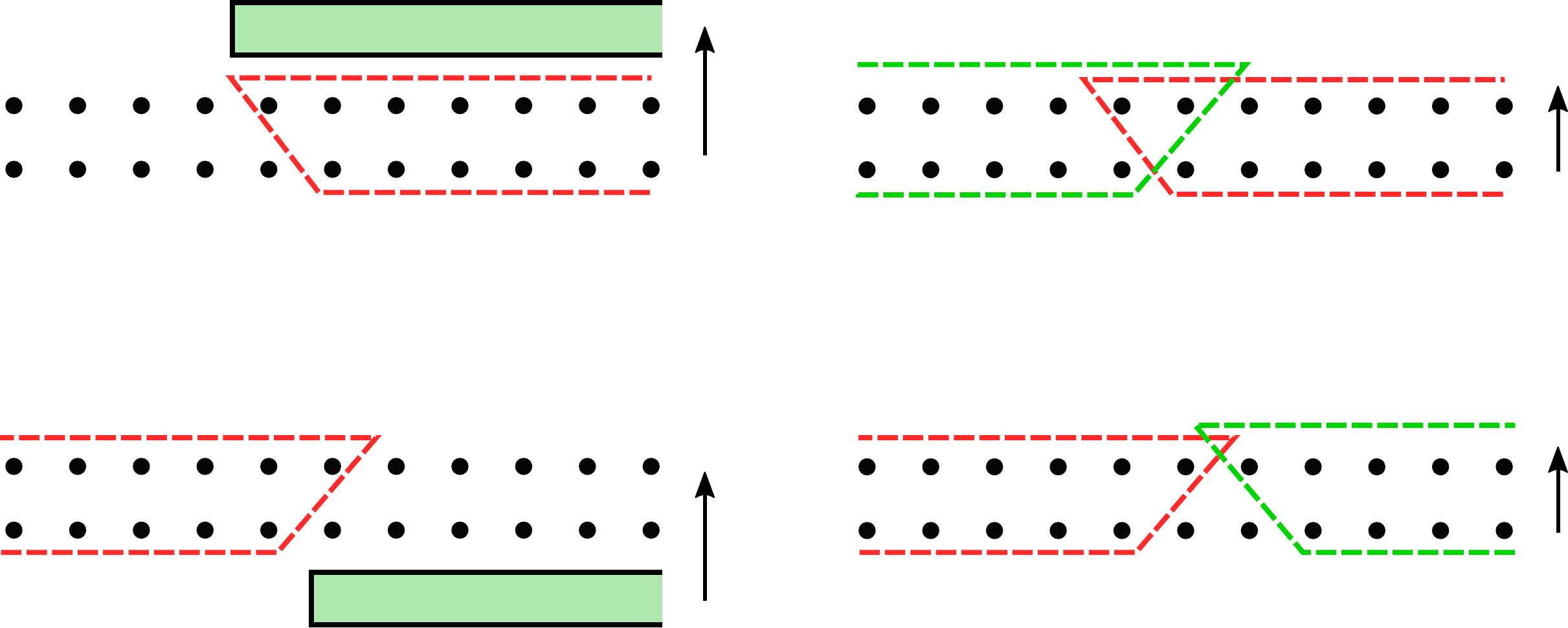}
\put(-8,31){(a)} \put(-8,7){(b)}
\put(30,37.5){$u$} \put(30,1.5){$u$}
\put(75,26){\tiny{$n$}} \put(69,26){\tiny{$n-1$}}
\put(71,3){\tiny{$n$}} \put(81.5,3){\tiny{$n+3$}}
\put(48,31){$=$} \put(48,7){$=$}
\put(103,31){$\tilde{\alpha}$} \put(103,7){$\tilde{\alpha}$}
\end{overpic}
\caption{(a) Illustration of the construction in \cref{lem:single rotation 0}. The dashed lines indicate causal cones. (b) Analogous illustration of \cref{lem:single rotation 1}.}
\label{fig:single rotation}
\end{figure}

The following lemma is used in the proof above (and in similar proofs below) that the construction gives rise to an ALPU when the input is an ALPU.

\begin{lem}\label{lem:adjusted alpha is alpu}
Suppose that $\alpha$ is an ALPU with $f(r)$-tails and $\tilde\alpha$ is an automorphism of $\A^{\vn}_{\ZZ}$ of the form $\tilde\alpha(x) = u^*\alpha(x)u$ for some $u \in \A^{\vn}_{\ZZ}$ and all $x \in \A^{\vn}_{\ZZ}$, which satisfies
\begin{equation}\label{eq:nec locality for alpu lemma}
\begin{aligned}
  \tilde{\alpha}(\A_{\leq n-1}) &\subseteq \A^{\vn}_{\leq n},  \\
  \tilde{\alpha}(\A_{\geq n}) &\subseteq \A^{\vn}_{\geq n-1}
\end{aligned}
\end{equation}
for some site $n\in\ZZ$.
If for any $r\geq0$ and $x \in \A_{\leq n - r - 2} \cup \A_{\geq n + r +1} \cup \alpha(\A_{\leq n- r - 1}) \cup \alpha(\A_{\geq n+r })$ we have
\begin{align*}
  \norm{u^*xu - x} \leq g(r)\norm{x},
\end{align*}
where $g(r)$ is non-increasing with $\lim_{r\to\infty} g(r)=0$, then $\tilde\alpha$ is an ALPU with $\bigO(f(r) + g(r))$-tails.
\end{lem}
\begin{proof}
We abbreviate $h(r) = f(r) + g(r)$.
By \cref{lem:half infinite}, it suffices to show
\begin{align*}
  \tilde\alpha(\A_{\leq m}) \overset{\bigO(h(r))}{\subseteq} \A^{\vn}_{\leq m+r}
\quad\text{and}\quad
  \tilde\alpha(\A_{\geq m}) \overset{\bigO(h(r))}{\subseteq} \A^{\vn}_{\geq m-r}
\end{align*}
for all $m\in\ZZ$ and $r\geq0$.
We only prove the former, since the proof of the latter proceeds analogously.
We distinguish two cases:
\begin{itemize}
\item $m<n$:
Then $m = n-k-1$ for some $k\geq0$.
By assumption, $\tilde\alpha(\A_{\leq n - k-1}) \subseteq \A^{\vn}_{\leq n}$, so it remains to show
\begin{align*}
  \tilde\alpha(\A_{\leq n-k-1}) \overset{\bigO(h(r))}{\subseteq} \A^{\vn}_{\leq n - k-1 + r}.
\end{align*}
for $0 \leq r \leq k$.
This holds since, by assumption,
$\norm{u^*xu - x} \leq g(k)\norm{x}$
for all~$x \in \alpha(\A_{\leq n- k - 1})$, and hence
$\norm{(\tilde\alpha - \alpha)|_{\A_{\leq n- k - 1}}} \leq g(k) \leq g(r)$
for any $0 \leq r \leq k$.

\item $m\geq n$:
Then $m = n+k$ for some $k\geq0$.
To prove that
\begin{align*}
  \tilde\alpha(\A_{\leq n+k}) \overset{\bigO(h(r))}{\subseteq} \A^{\vn}_{\leq n+k + r}
\end{align*}
by \cref{lem:near inclusion commutator 0} it suffices to show that for all $x \in \A_{\leq n+k}$ and for all $y \in \A^{\vn}_{\geq n+k + r+1}$,
\begin{align*}
  \norm{[\tilde\alpha(x),y]} = \bigO(h(r)) \norm x \norm y.
\end{align*}
By assumption $\norm{u y u^* - y} = \norm{u^* y u - y} \leq g(k+r) \norm y \leq g(r) \norm y$ for all~$y\in \A_{\geq n+k + r+1}$ and thus for all~$\A^{\vn}_{\geq n+k + r+1}$, and hence we indeed have that
\begin{align*}
  \norm{[\tilde\alpha(x),y]}
= \norm{[\alpha(x),uyu^*]}
\leq \norm{[\alpha(x),y]} + \norm{[\alpha(x),uyu^*-y]}
= O(h(r)) \norm x \norm y,
\end{align*}
using that $\norm{[\alpha(x),y]} \leq 2 f(r) \norm x \norm y$ by \cref{lem:near inclusion commutator 0}, since $\alpha$ is an ALPU with $f(r)$-tails.
\qedhere
\end{itemize}
\end{proof}

\begin{lem}\label{lem:single rotation 1}
\begin{enumerate}
\item There exist universal constants $C_0', \eps'_0 > 0$ such that if $\alpha$ is an $\eps$-nearest neighbor automorphism of $\A^{\vn}_\ZZ$ with $\eps \leq \eps'_0$ and
\begin{align*}
  \alpha(\A^{\vn}_{\leq n}) \subseteq \A^{\vn}_{\leq n + 1}
\end{align*}
for some site $n \in \ZZ$, then there exists an automorphism of~$\A^{\vn}_\ZZ$ of the form~$\tilde{\alpha}(x) = \alpha(uxu^*)$ for some unitary $u \in \A^{\vn}_{\geq n + 1}$ with $\norm{u - I} \leq C_0'\eps$ and
\begin{align}
\label{eq:single rotation 1 a}
    \tilde{\alpha}(\A^{\vn}_{\geq n + 3}) &\subseteq \A^{\vn}_{\geq n + 2}, \\
\label{eq:single rotation 1 b}
    \tilde{\alpha}(\A^{\vn}_{\leq n}) &\subseteq \A^{\vn}_{\leq n + 1}.
\end{align}
In fact, it holds that $\alpha|_{\A^{\vn}_{\leq n}} = \tilde{\alpha}|_{\A^{\vn}_{\leq n}}$.

\item If additionally $\alpha$ is an ALPU with $f(r)$-tails, we can take $u$ such that $\tilde{\alpha}$ is an ALPU with $\bigO(f(r-1))$-tails and such that for $r\to\infty$,
\begin{align}
\label{eq:asymptotically equal 1 right}
  \norm{(\alpha - \tilde{\alpha})|_{\A^{\vn}_{\geq n+r+ 3}}} &= \bigO(f(r)).
\end{align}
\end{enumerate}
\end{lem}
\begin{proof}
(i)~This follows by application of \cref{lem:single rotation 0} to $\beta = \alpha^{-1}$.
Here we use that if $\alpha$ is an $\eps$-nearest neighbour automorphism of $\A_{\ZZ}^{\vn}$, then $\beta$ is a $4\eps$-nearest neighbour automorphism by \cref{lem:alpu to vn}\ref{item:inverse vn auto}.
Now, $\alpha(\A^{\vn}_{\leq n}) \subseteq \A^{\vn}_{\leq n + 1}$ implies that $\A^{\vn}_{\leq n} \subseteq \beta(\A^{\vn}_{\leq n + 1})$ and hence~$\beta(\A^{\vn}_{\geq n + 2}) \subseteq \A^{\vn}_{\geq n+1}$.
Let $\eps_0' := \eps_0/4$ and $C_0' = 4C_0$ for $C_0,\eps_0>0$ the constants from \cref{lem:single rotation 0}.
Thus we may apply \cref{lem:single rotation 0} to~$\beta$ (with $n+2$ in place of $n$) to find an automorphism~$\tilde\beta$ of $\A_{\ZZ}^{\vn}$ that is of the form $\tilde{\beta}(x) = u^*\beta(x)u$ for some unitary $u \in \A^{\vn}_{\geq n + 1}$ with $\norm{u - I} \leq 4C_0\eps = C_0'\eps$ and which satisfies
\begin{align*}
  \tilde{\beta}(\A^{\vn}_{\leq n+1}) &\subseteq \A^{\vn}_{\leq n+2}, \\
  \tilde{\beta}(\A^{\vn}_{\geq n+2}) &\subseteq \A^{\vn}_{\geq n+1}.
\end{align*}
We then see that $\tilde\alpha = \tilde\beta^{-1}$ is given by $\tilde{\alpha}(x) = \alpha(uxu^*)$ and satisfies the desired properties in~(i).
In particular, note that $u \in \A^{\vn}_{\geq n+1}$ immediately implies that~$\alpha|_{\A^{\vn}_{\leq n}} = \tilde{\alpha}|_{\A^{\vn}_{\leq n}}$.

(ii)~If $\alpha$ is an ALPU with $f(r)$-tails, then by \cref{lem:alpu to vn}\ref{item:inverse alpu} $\beta$ is an ALPU with $\bigO(f(r))$-tails.
Hence by part~(ii) of \cref{lem:single rotation 0}, $\tilde\beta$ is an ALPU with $\bigO(f(r-1))$-tails and thus the same is true for~$\tilde\alpha$, again by \cref{lem:alpu to vn}\ref{item:inverse alpu}.
\cref{eq:asymptotically equal 1 right} follows since by \cref{eq:asymptotically equal 0 on image} we have
\begin{align*}
  \norm{uxu^* - x} = \norm{u^*xu - x} = \bigO(f(r)\norm{x})
\end{align*}
for all $x \in \A_{\geq n + r + 3}$.
\end{proof}

We iteratively apply \cref{lem:single rotation 0} and \cref{lem:single rotation 1} to show that for an $\eps$-nearest neighbor automorphism, for any small patch, one can find a nearby $\bigO(\eps)$-nearest neighbor automorphism that is strictly local on that patch.
Below we work with a patch near site $2n$, and the modified automorphism is denoted $\alpha_n$.

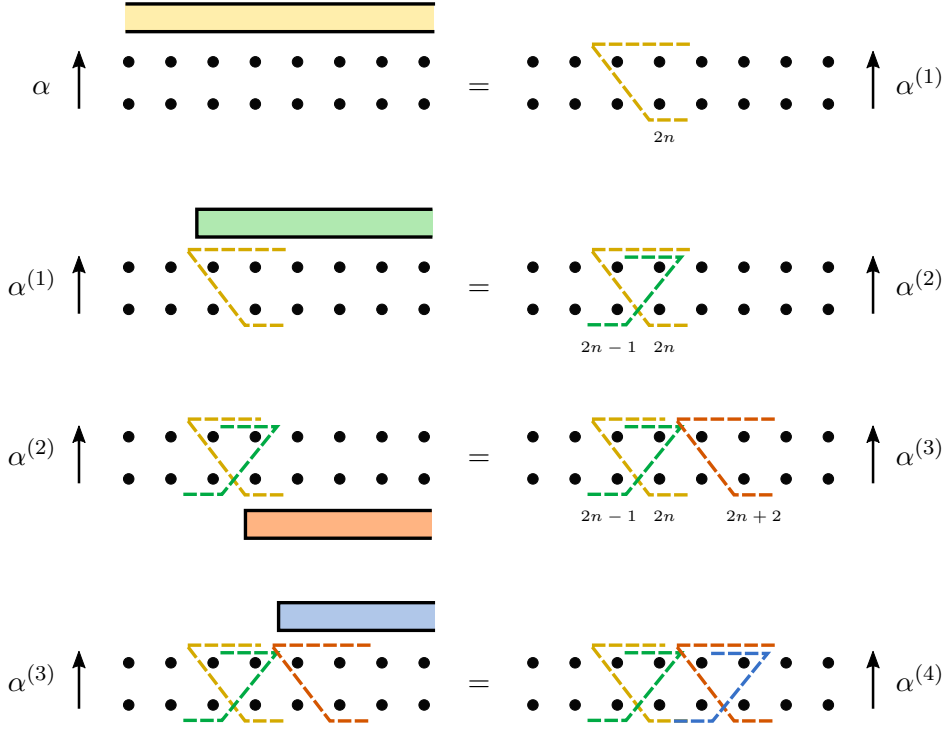
\begin{figure}[t]
\centering
\begin{overpic}[width=0.65\textwidth,grid=false]{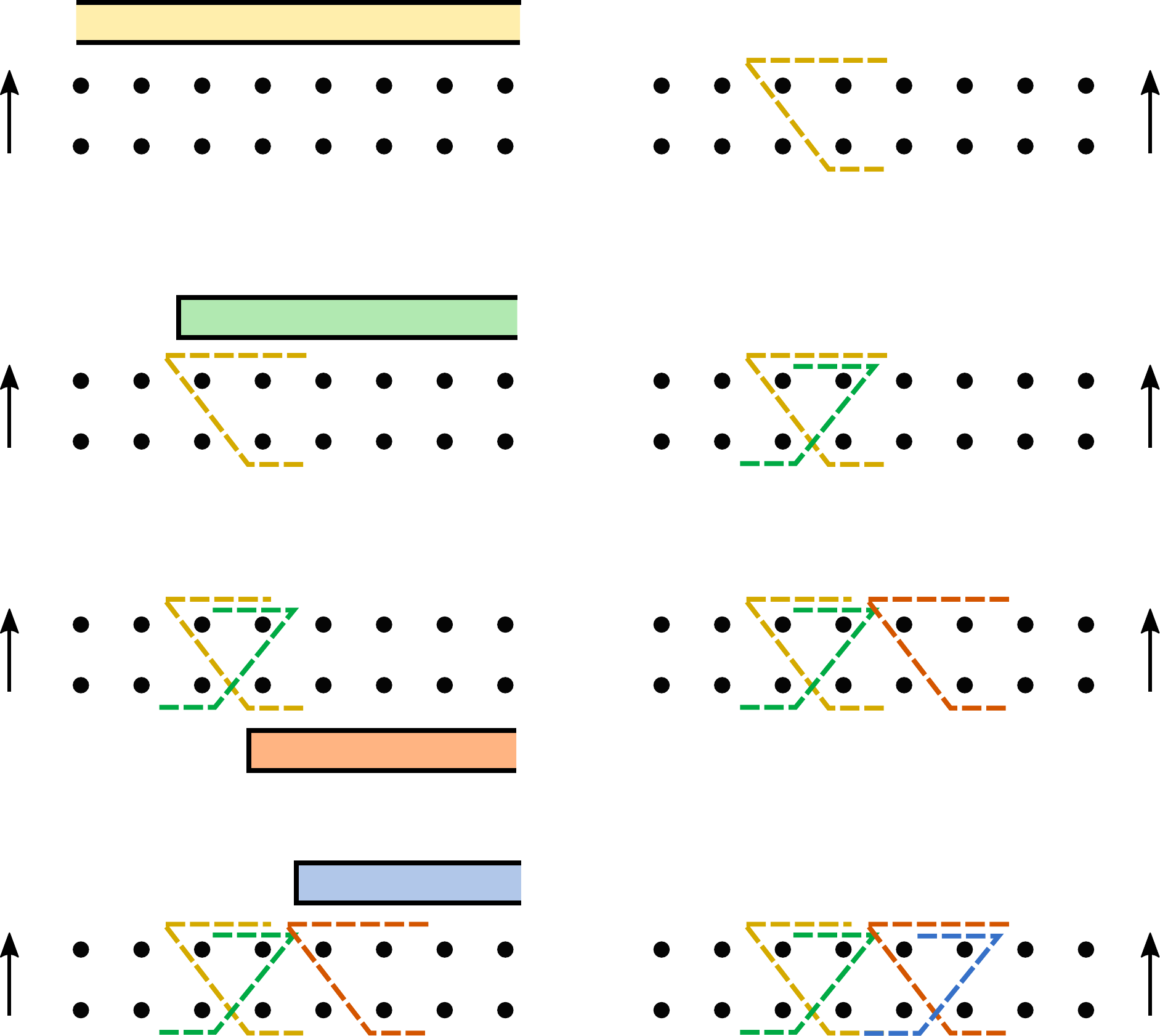}
\put(49,78){$=$} \put(49,53){$=$} \put(49,32){$=$} \put(49,4){$=$}
\put(-5,78){$\alpha$}
\put(-8,53){$\alpha^{(1)}$}
\put(-8,32){$\alpha^{(2)}$}
\put(-8,4){$\alpha^{(3)}$}
\put(102,78){$\alpha^{(1)}$}
\put(102,53){$\alpha^{(2)}$}
\put(102,32){$\alpha^{(3)}$}
\put(102,4){$\alpha^{(4)}$}
\put(72,72){\tiny{$2n$}}
\put(63,46){\tiny{$2n-1$}} \put(72,46){\tiny{$2n$}}
\put(63,25){\tiny{$2n-1$}} \put(72,25){\tiny{$2n$}} \put(81,25){\tiny{$2n+2$}}
\end{overpic}
\caption{Illustration of the construction of $\alpha^{(i)}$ for $i = 1,2,3,4$ in the proof of \cref{prop:single patch localized}.   Each row depicts an equation, and the solid strips on the left depict applications of unitary operators supported on those regions. The second, third, and fourth row use \cref{lem:single rotation 0}, \cref{lem:single rotation 1}, and \cref{lem:single rotation 0} again respectively.
Dashed lines indicate causal cones.}
\label{fig:single patch localized}
\end{figure}

\begin{prop}\label{prop:single patch localized}
\begin{enumerate}
\item There exist universal constants $C_1, \eps_1 > 0$ such that for any $\eps$-nearest neighbor automorphism $\alpha$ of $\A_{\ZZ}^{\vn}$ with $\eps \leq  \eps_1$ and for any site $n \in \ZZ$, there exists an automorphism~$\alpha_n$ of $\A_{\ZZ}^{\vn}$ such that for $k \in \{0, 1, 2, 3\}$,
\begin{align*}
    \alpha_n(\A^{\vn}_{\leq 2n + 2k - 1}) &\subseteq \A^{\vn}_{\leq 2n + 2k},\\
    \alpha_n(\A^{\vn}_{\geq 2n + 2k}) &\subseteq \A^{\vn}_{\geq 2n + 2k - 1},\\
    \norm{\alpha_n - \alpha} &\leq C_1\eps.
\end{align*}
In particular, denoting $\B_m = \A_{\{2m,2m+1\}}$ and $\C_m = \A_{\{2m-1,2m\}}$ as in \cref{eq:B and C regions}, we have
\begin{align*}
  \alpha_n(\B_m) &\subseteq \C_m \ot \C_{m+1} \quad \text{ for } m\in\{n,n+1,n+2\}, \\
  \alpha_n^{-1}(\C_m) &\subseteq \B_{m-1} \ot \B_m \quad \text{ for } m\in\{n+1,n+2\}.
\end{align*}
\item Moreover, if $\alpha$ is an ALPU with $f(r)$-tails, we may take $\alpha_n$ to be an ALPU with $\bigO(f(r - 7))$-tails and such that, for $r \to \infty$,
\begin{align}
\label{eq:single patch tail bound first}
  \norm{(\alpha-\alpha_n)|_{\A^{\vn}_{\leq 2n-r-1}}} & = \mathcal{O}(f(r - 1)),\\
\label{eq:single patch tail bound second}
  \norm{(\alpha-\alpha_n)|_{\A^{\vn}_{\geq 2n+r+6}}} & = \mathcal{O}(f(r - 7)).
\end{align}
\end{enumerate}
\end{prop}
\begin{proof}
(i)~We define a sequence of automorphisms $\alpha_n^{(i)}$, $i = 1, \ldots, 8$ to obtain $\alpha_n := \alpha_n^{(8)}$ with the desired properties.
To begin,
\begin{align*}
  \alpha(\A^{\vn}_{\geq 2n}) \overset{\eps}{\subseteq} \A^{\vn}_{\geq 2n - 1}
\end{align*}
so by \cref{thm:near inclusion}, with $\A_0 = \alpha(\A_{\geq 2n})$, $\A = \A_0'' = \alpha(\A_{\geq 2n}^{\vn})$ and $\B = \A_{\geq 2n-1}^{\vn}$, there exists $u_1 \in (\alpha(\A^{\vn}_{\geq 2n}) \cup \A^{\vn}_{\geq 2n-1})''$ such that
\begin{align*}
  u_1^*\alpha(\A^{\vn}_{\geq 2n})u_1 \subseteq \A^{\vn}_{\geq 2n - 1}
\end{align*}
with $\norm{u_1^* - I} \leq 12\eps$.
We define $\alpha_n^{(1)}(x) = u_1^* \alpha(x)u_1$, which by construction satisfies
\begin{align*}
    \alpha_n^{(1)}(\A^{\vn}_{\geq 2n}) &\subseteq \A^{\vn}_{\geq 2n - 1}, \\
    \norm{\alpha - \alpha_n^{(1)}} & = \mc{O}(\eps).
\end{align*}
Then $\alpha_n^{(1)}$ is an $\mc{O}(\eps)$-nearest neighbor automorphism by the above, and we are in a situation where we can apply \cref{lem:single rotation 0} (but replacing $n$ with $2n$) to obtain an automorphism $\alpha_n^{(2)}(x) = u_2^* \alpha^{(1)}_n (x) u_2$ for unitary $u_2 \in \A^{\vn}_{\geq 2n - 1}$, such that
\begin{align*}
  \alpha_n^{(2)}(\A^{\vn}_{\leq 2n - 1}) &\subseteq \A^{\vn}_{\leq 2n},\\
  \alpha_n^{(2)}(\A^{\vn}_{\geq 2n}) &\subseteq \A^{\vn}_{\geq 2n - 1},\\
  \norm{\alpha_n^{(1)} - \alpha_n^{(2)}} & = \mc{O}(\eps).
\end{align*}
Then $\alpha_n^{(2)}$ is again an $\mc{O}(\eps)$-nearest neighbor automorphism, and we can apply \cref{lem:single rotation 1} (but replacing $n$ with $2n-1$) to obtain an automorphism $\alpha_n^{(3)}(x) =  \alpha^{(2)}_n(u_3 x u_3^*) $ for unitary $u_3 \in \A^{\vn}_{\geq 2n}$, such that
\begin{align*}
  \alpha_n^{(3)}(\A^{\vn}_{\geq 2n+2}) &\subseteq \A^{\vn}_{\geq 2n + 1},\\
  \norm{\alpha_n^{(3)} - \alpha_n^{(2)}} & = \mc{O}(\eps).
\end{align*}
Since $u_3 \in \A^{\vn}_{\geq 2n}$, $\alpha_n^{(3)}$ also satisfies the locality properties listed for $\alpha_n^{(2)}$ above.
See \cref{fig:single patch localized} for an illustration of the construction.

We continue to apply \cref{lem:single rotation 0} and \cref{lem:single rotation 1} alternatingly.
Explicitly, we apply \cref{lem:single rotation 0} (with~$n \to 2n+2$) to define~$\alpha_n^{(4)}$, as illustrated in the figure, and then \cref{lem:single rotation 1} (with~$n \to 2n+1$) to define~$\alpha_n^{(5)}$, followed by \cref{lem:single rotation 0} (with~$n \to 2n+4$) to define~$\alpha_n^{(6)}$ and \cref{lem:single rotation 1} (with~$n \to 2n+3$) to define~$\alpha_n^{(7)}$.
Finally we use \cref{lem:single rotation 0} (with~$n \to 2n+6 $) to obtain~$\alpha_n^{(8)}$.
We take $\alpha_n := \alpha_n^{(8)}$; then $\alpha_n$ has the desired locality properties in the proposition statement.
We must assume $\eps$ is sufficiently small to meet the conditions of these lemmas at each step, determining the universal constant $\epsilon_1$ in the proposition statement.

(ii)~Now we further assume $\alpha$ is an ALPU with $f(r)$-tails, to demonstrate \cref{eq:single patch tail bound first,eq:single patch tail bound second} and prove that~$\alpha_n$ is an ALPU.

We first show that~$\alpha_n^{(2)}$ is an ALPU with $\mathcal{O}(f(r-1))$ tails, using \cref{lem:adjusted alpha is alpu} (with~$n \mapsto 2n$).
Note that $\alpha_n^{(2)}(x) = v^* \alpha(x) v$ for $v = u_1 u_2$, and $\alpha_n^{(2)}$ satisfies the necessary locality properties in \cref{eq:nec locality for alpu lemma} (unlike $\alpha_n^{(1)}$!), so in order to apply the lemma we only need to show that
\begin{align}\label{eq:required for lemma}
  \norm{v^*xv - x} = \mathcal{O}(f(r-1)) \norm{x},
\end{align}
for all $r\geq0$ and $x \in \A_{\leq 2n - r - 2} \cup \A_{\geq 2n + r +1} \cup \alpha(\A_{\leq 2n- r - 1}) \cup \alpha(\A_{\geq 2n+r })$.
To this end, recall that for~$u_1$ we applied \cref{thm:near inclusion} with $\A_0 = \alpha(\A_{\geq 2n})$, $\A = \A_0'' = \alpha(\A_{\geq 2n}^{\vn})$ and $\B = \A_{\geq 2n-1}^{\vn}$, and in the construction of~$u_2\in\A^{\vn}_{\geq 2n - 1}$ in \cref{lem:single rotation 0} (with $n \mapsto 2n$) we applied \cref{thm:near inclusion} with~$\A_0 = \A_{\geq 2n + 1}$, $\A = \A_0'' = \A_{\geq 2n + 1}^{\vn}$ and $\B = \alpha^{(1)}_{n}(\A^{\vn}_{\geq 2n})$.

First consider $x \in \A_{\leq 2n - r - 2}$.
As $\alpha(\A^{\vn}_{\geq 2n}) \overset{f(r+1)}{\subseteq} \A^{\vn}_{\geq 2n-r-1}$, we have $\norm{[x,y]} = 2 f(r + 1) \norm{x}\norm{y}$ for all $y \in \alpha(\A^{\vn}_{\geq 2n})$ by \cref{lem:near inclusion commutator 0}.
Since moreover $[x,y]=0$ for all $y \in \A^{\vn}_{\geq 2n - 1}$, \cref{thm:near inclusion}\ref{it:near inclusion small commutator} shows that $\norm{u_1^*xu_1 - x} = \bigO(f(r + 1)\norm{x})$.
In addition, we have $u_2^*xu_2 = x$ since $u_2\in\A^{\vn}_{\geq 2n - 1}$.
Together we find that $\norm{v^*xv - x} =  \bigO(f(r + 1)\norm{x})$.

Next consider $x \in \A_{\geq 2n + r + 1}$.
By \cref{lem:alpu to vn}, $\alpha^{-1}$ is an ALPU with $\bigO(f(r))$-tails, so we have
\begin{align}\label{eq:inverse argument}
  \alpha^{-1}(x) \overset{\bigO(f(r+1))}{\in} \A_{\geq 2n}
\end{align}
and hence $x \overset{\bigO(f(r+1))}{\in} \alpha(\A_{\geq 2n})$.
Since moreover $x \in \A^{\vn}_{\geq 2n - 1}$, \cref{thm:near inclusion}\ref{it:near inclusion already close} shows that $\norm{u_1^*xu_1 - x} = \bigO(f(r+1)\norm{x})$.
Since $(\alpha^{(1)}_n)^{-1}(x) = \alpha^{-1}(u_1 x u_1^*)$, the latter along with \cref{eq:inverse argument} in turn implies that
$x \overset{\bigO(f(r+1))}\in \alpha^{(1)}_n(\A_{\geq 2n})$.
Also, $x \in \A_{\geq 2n+1}$, hence we obtain $\norm{u_2^* x u_2 - x} = \bigO(f(r+1))$, again by \cref{thm:near inclusion}\ref{it:near inclusion already close}.
Together we find that $\norm{v^*xv - x} = \bigO(f(r+1)\norm{x})$.

Now consider $x \in \alpha(\A_{\leq2n-r-1})$, i.e., $x=\alpha(z)$ for some $z \in \A_{\leq 2n - r - 1}$.
Then~$x$ commutes with~$\alpha(\A^{\vn}_{\geq 2n})$.
Moreover, $x \overset{f(r-1)}\in \A_{\leq2n-2}$, hence $\norm{[x,y]} \leq 2f(r-1)\norm{x}\norm{y}$ for all~$y \in \A^{\vn}_{\geq 2n - 1}$.
Thus we obtain $\norm{u_1^*xu_1 - x} = \bigO(f(r-1) \norm{x})$ by \cref{thm:near inclusion}\ref{it:near inclusion small commutator}.
The preceding in turn implies that for all $y\in\A_{\geq2n+1}^{\vn}$,
\begin{align*}
  \norm{[\alpha_n^{(1)}(z),y]}
= \norm{[u_1^* x u_1,y]}
\leq 2 \norm{u_1^* x u_1 - x} \norm y + \norm{[x,y]}
= \bigO(f(r-1) \norm{x} \norm{y}).
\end{align*}
Also, $\alpha_n^{(1)}(z)$ commutes with~$\alpha_n^{(1)}(\A_{\geq2n}^{\vn})$.
Therefore, again by \cref{thm:near inclusion}\ref{it:near inclusion small commutator} we see that
\begin{align*}
  \norm{v^* x v - u_1^* x u_1}
= \norm{u_2^*\alpha_n^{(1)}(z)u_2 - \alpha_n^{(1)}(z)}
= \bigO(f(r-1)\norm{x}).
\end{align*}
We conclude that $\norm{v^* x v - x} = \bigO(f(r-1)\norm{x})$.

Finally, let $x \in \alpha(\A_{\geq 2n + r})$, i.e., $x=\alpha(z)$ for some $z \in \A_{\geq 2n + r}$.
Then $x \in \alpha(\A_{\geq 2n})$ and~$x \overset{f(r+1)}{\in} \A^{\vn}_{\geq 2n - 1}$.
So, by \cref{thm:near inclusion}\ref{it:near inclusion already close} we find that $\norm{u_1^*xu_1 - x} = \bigO(f(r+1)\norm{x})$.
Using the latter, as well as $\norm{\alpha_n^{(1)}(z) - \alpha(z)} = \norm{u_1^*xu_1 - x}$ and $\alpha(z) \overset{f(r-1)}\in \A_{\geq 2n+1}$, we find $\alpha_n^{(1)}(z) \overset{\bigO(f(r-1))}\in \A_{\geq2n+1}$.
Moreover, $\alpha_n^{(1)}(z) \in \alpha_n^{(1)}(\A_{\geq2n}^{\vn})$, so by \cref{thm:near inclusion}\ref{it:near inclusion already close} we obtain that
\begin{align*}
  \norm{v^*xv - u_1^*xu_1}
= \norm{u_2^*\alpha_n^{(1)}(z)u_2 - \alpha_n^{(1)}(z)}
= \bigO(f(r-1)\norm{x}),
\end{align*}
and hence $\norm{v^*xv - x} = \bigO(f(r-1)\norm x)$.
Altogether we have verified that \cref{eq:required for lemma} holds for all $r\geq0$ and $x \in \A_{\leq 2n - r - 2} \cup \A_{\geq 2n + r +1} \cup \alpha(\A_{\leq 2n- r - 1}) \cup \alpha(\A_{\geq 2n+r })$.
We may therefore apply \cref{lem:adjusted alpha is alpu} and conclude that $\alpha_n^{(2)}$ is an ALPU with $\bigO(f(r-1))$-tails.

For $i=3,\dots,8$, we simply observe that by our applications of \cref{lem:single rotation 0} and \cref{lem:single rotation 1}, the automorphisms $\alpha_n^{(i)}$ are guaranteed to be APLUs with $\bigO(f(r+1-i))$-tails.

To see that \cref{eq:single patch tail bound first} holds, note that \cref{eq:required for lemma} implies that $\norm{\alpha_n^{(2)}(x) - \alpha(x)} = \bigO(f(r-1) \norm x)$ for all~$x \in \A_{\leq2n-r-1}$.
Moreover, $\alpha_n(x) = \alpha_n^{(2)}(x)$ for such~$x$, since $\alpha_n = \alpha_n^{(8)}$ is obtained from $\alpha_n^{(2)}$ by conjugating the input with unitaries in~$\A^{\vn}_{\geq 2n}$ (leaving~$x$ unchanged) and the output by unitaries in~$\A^{\vn}_{\geq2n+1}$ (leaving~$\alpha_n^{(2)}(x) \in \A_{\leq2n}^{\vn}$ unchanged).
Thus \cref{eq:single patch tail bound first}~follows.

Finally, \cref{eq:single patch tail bound second} follows since the $\alpha^{(i)}_n$ for $i=3,\dots,8$ satisfy analogs of \cref{eq:asymptotically equal 0 right,eq:asymptotically equal 1 right} and we have
$\norm{\alpha_n^{(2)}(x) - \alpha(x)} \leq \bigO(f(r-1) \norm x)$
for all $x\in \A_{\geq 2n+r }$, again by \cref{eq:required for lemma}.
\end{proof}


\begin{prop}[QCA approximation of $\eps$-nearest neighbor automorphism]\label{prp:approximation almost nn}
There exists a universal constant~$\eps_2>0$ such that if~$\alpha$ is an $\eps$-nearest neighbor automorphism of $\A_\ZZ$ with $\eps \leq \eps_2$, then there exists a QCA~$\beta$ with radius~$2$ such that
\begin{align*}
  \norm{(\alpha - \beta)|_{\A_X}}= \bigO(\eps|X|)
\end{align*}
for all regions $X$ with $|X|$ sites.
\end{prop}
\begin{proof}
Recall that $\alpha$ extends to a $\eps$-nearest neighbor automorphism of $\A_{\ZZ}^{\vn}$ by \cref{lem:alpu to vn}, which we will denote by the same symbol.
Let $C_1$ and $\eps_1$ be the constants from \cref{prop:single patch localized},
and take~$\eps_2:=\min\{\frac{\eps_1}{2}, \frac{1}{384C_1}\}$.
As usual, we write~$\B_n = \A_{\{2n,2n+1\}}$ and~$\C_n = \A_{\{2n-1,2n\}}$.
Now apply part~(i) of \cref{prop:single patch localized} to find automorphisms~$\alpha_m$, one for each $m \in \ZZ$, which satisfy the locality properties~$\alpha_m(\B_n) \subseteq \C_n \ot \C_{n+1}$ for $n \in \{m,m+1,m+2\}$ as well as $\alpha_m^{-1}(\C_n)\subseteq \B_{n-1} \ot \B_n$ for~$n\in\{m+1,m+2\}$.
Then by \cref{thm:gnvw} and the subsequent \cref{rem:weaker locality properties}, we can define
\begin{align*}
  \L_{n}^{(m)} &= \alpha_{m}(\B_{n}) \cap  \C_{n}, \\
  \R_{n-1}^{(m)} &=  \alpha_{m}(\B_{n-1}) \cap  \C_{n}
\end{align*}
such that, for $m \in \{n-1, n-2\}$,
\begin{align}\label{eq:factorize}
  \C_{n} = \L_{n}^{(m)} \ot \R_{n-1}^{(m)}.
\end{align}
Moreover, again by \cref{thm:gnvw,rem:weaker locality properties}, we have
\begin{align}\label{eq:factorize b}
  \B_n = \alpha_{n-1}^{-1}(\L_{n}^{(n-1)}) \ot \alpha_{n-1}^{-1}(\R_{n}^{(n-1)}),
\end{align}
which we will use below.

Note that $\norm{\alpha_{n-1} - \alpha_{n-2}} \leq \norm{\alpha_{n-1} - \alpha} + \norm{\alpha_{n-2} - \alpha} \leq 2C_1\eps \leq \frac{1}{192}$.  Because $\alpha_{n-1}$ and $\alpha_{n-2}$ are nearby ALPUs with locality properties satisfying \cref{rmk:weaker locality properties 2}, we can apply the argument from \cref{prop:nearby-qcas} to $\alpha_{n-1}$ and $\alpha_{n-2}$, finding that $\L_{n}^{(n-2)}$ and $\L_n^{(n-1)}$ are related by a unitary $u_n \in \C_n$, i.e.\ $u_n \L_n^{(n-1)} u_n^* =\L_{n}^{(n-2)}$, with $\norm{u_n - I} = \bigO(\eps)$.
Finally we define
\begin{align*}
  \beta_n \colon \B_n & \rightarrow \C_n \ot \C_{n+1}, \quad
  \beta_n(x) = u_n\alpha_{n-1}(x)u_n^*.
\end{align*}
Each $\beta_n$ is an injective homomorphism and by~\eqref{eq:factorize b} we obtain
\begin{align}\label{eq:beta_n image}
  \beta_n(\B_n)
= u_n \left( \L_n^{(n-1)} \ot \R_n^{(n-1)} \right) u_n^*
= \L_n^{(n-2)} \ot \R_n^{(n-1)},
\end{align}
where the second equality holds because $u_n \in \C_n$ and $\R_n^{(n-1)} \subseteq \C_{n+1}$ commute.
From \cref{eq:beta_n image} we conclude that $\beta_n(\B_n)$ and $\beta_m(\B_m)$ commute for $n \neq m$.
Hence we can define a global injective homomorphism $\beta$ that acts as $\beta_n$ on each $\B_n$.
By~\eqref{eq:factorize} and~\eqref{eq:beta_n image}, this homomorphism is surjective.
Indeed, $\beta_{n-1}(\B_{n-1}) \ot \beta_n(\B_n) \supseteq \R_{n-1}^{(n-2)} \ot \L_n^{(n-2)} = \C_n$ for all~$n$.
Thus the map $\beta$ is an automorphism. By construction it is clear that this automorphism is a QCA with radius~$2$.
For any single site operator $x$ we have that $x \in \B_n$ for some $n$, so
\begin{align*}
  \norm{\beta(x) - \alpha(x)} &\leq 2\norm{u_n - I} + \norm{\alpha - \alpha_{n-1}}\\
  &\leq \bigO(\eps).
\end{align*}
We showed $\norm{(\beta - \alpha)|_{\A_n}} = \mathcal{O}(\epsilon)$ for all single sites $n$, and the desired result holds by \cref{lem:homomorphism local error}.
\end{proof}

By \cref{prp:approximation almost nn} and coarse-graining, we obtain the main result of this section, which shows that any ALPU in one dimensions can be approximated by a sequence of QCAs.

\begin{thm}[QCA approximations]\label{thm:qca approx}
If $\alpha$ is a one-dimensional ALPU with $f(r)$-tails, then there exists a sequence of QCAs~$\{\beta_j\}_{j=1}^{\infty}$, of radius~$2j$ such that for any finite subset~$X\subset\ZZ$,%
\footnote{Below $\abs X$ denotes the number of sites in the subset~$X$, while $\diam(X) = \max(X) - \min(X) + 1$ denotes the diameter of the subset, and $\ceil{\cdot}$ is the integer ceiling.}
\begin{align}\label{eq:ALPU QCA distance universal}
  \norm{(\alpha - \beta_j)|_{\A_X}} = \bigO\bigl( f(j)\min\left\{|X|,\ceil*{\tfrac{\diam(X)}{j}}\right\} \bigr).
\end{align}
Moreover, there is a constant $C_f>0$, depending only on $f(r)$, such that the following holds for all~$j$ and finite $X \subset \ZZ$:
\begin{align}\label{eq:ALPU QCA distance}
  \norm{(\alpha - \beta_j)|_{\A_X}} \leq C_f \, f(j)\min\left\{|X|,\ceil*{\tfrac{\diam(X)}{j}}\right\}.
\end{align}
In particular, the $\beta_j$ converge strongly to $\alpha$, meaning that $\lim_{j \rightarrow \infty} \norm{\alpha(x) - \beta_j(x)} = 0$ for all~$x \in \A_{\ZZ}$.
\end{thm}
\begin{proof}
By blocking $j$ sites we obtain an $\eps_j$-nearest neighbor QCA on the coarse-grained lattice where~$\eps_j = f(j)$.
For $j > j_0$ sufficiently large, we can apply \cref{prp:approximation almost nn} to obtain a QCA $\beta_j$ of radius~2 on the coarse-grained lattice satisfying  $\norm{(\alpha - \beta)|_{\A_X}}= \bigO(f(j) m)$ for all regions $X$ composed of~$m$ coarse-grained sites.
If we now consider $\beta_j$ as a QCA of radius $2j$ on the original lattice (before coarse-graining), we arrive at~\eqref{eq:ALPU QCA distance universal}.
To obtain~\eqref{eq:ALPU QCA distance}, for smaller $j \leq j_0$, we may choose some arbitrary QCA $\beta_j$ and use that $\norm{\alpha-\beta_j} \leq 2$ at the expense of incurring a tails-dependent constant~$C_f>0$.

We now show that the sequence of QCAs~$\beta_j$ converges strongly to $\alpha$.
For $x \in \A_{\ZZ}$ arbitrary, let~$x_n$ be a sequence of strictly local operators, where $x_n$ is supported on $n$ contiguous sites, such that $\lim_{n \rightarrow \infty} x_n = x$ converges in norm.
Then,
\begin{align*}
  \limsup_{j \rightarrow \infty} \, \norm{\alpha(x) - \beta_j(x)}
&\leq \limsup_{j \rightarrow \infty} \, \Bigl( \norm{\alpha(x) - \alpha(x_n)} + \norm{\alpha(x_n) - \beta_j(x_n)} + \norm{ \beta_{j}(x_n) - \beta_j(x)} \Bigr) \\
&\leq 2 \norm{x - x_n} + \limsup_{j \rightarrow \infty} \, \norm{\alpha(x_n) - \beta_j(x_n)}
= 2 \norm{x - x_n}.
\end{align*}
The second inequality holds since $\alpha$ and the $\beta_j$ are $*$-homomorphisms; the final equality follows by~\eqref{eq:ALPU QCA distance}.
Since the above holds for all $n$, we conclude that $\lim_{j \rightarrow \infty} \norm{\alpha(x) - \beta_j(x)} = 0$.
\end{proof}

\subsection{Definition of the index for ALPUs}
We now use the QCA approximations developed in the preceding to define an index for general ALPUs.
In addition, we give two alternative ways of computing the index for ALPUs with appropriately decaying tails, and we prove that the index is stable also for ALPUs.

\begin{dfn}[Index for ALPUs]\label{dfn:index}
Let $\alpha$ be a one-dimensional ALPU with $f(r)$-tails and let $\beta_j$ be a sequence of QCAs of radius at most $2j$ such that for any finite subset~$X\subset\ZZ$,
\begin{align}\label{eq:ALPU QCA distance in dfn}
  \norm{(\alpha - \beta_j)|_{\A_X}} \leq C_f \, f(j) \ceil*{\tfrac{\diam(X)}{j}},
\end{align}
where $C_f>0$ is the constant from \cref{thm:qca approx}.
We define the \emph{index} of $\alpha$ by
\begin{align}\label{eq:QCA index limit}
\ind(\alpha) := \lim_{j \to \infty} \ind(\beta_j).
\end{align}
\end{dfn}
\noindent
Note that by \cref{thm:qca approx} such a sequence $\beta_j$ always exists.
The following theorem shows that the index is a well-defined, finite quantity.

\begin{thm}[Index for ALPUs]\label{thm:index alpu after reorg}
Let $\alpha$ be a one-dimensional ALPU with $f(r)$-tails and let~$\beta_j$ be a sequence of QCAs as in \cref{dfn:index}.
Then the following hold:
\begin{enumerate}
\item\label{it:index as limit}
There exists $j_0$, depending only on $f(r)$, such that $\ind(\beta_j)$ is constant for $j \geq j_0$.
Accordingly, the limit~\eqref{eq:QCA index limit} exists and is in $\ZZ[\{\log(p_i)\}]$, where the~$p_i$ are the finitely many prime factors of the local Hilbert space dimensions~$d_n$, and $\ZZ[\cdot]$ denotes integer linear combinations.
Moreover, this limit does not depend on the choice of sequence~$\beta_j$.
Thus, $\ind(\alpha)$ is well-defined by~\eqref{eq:QCA index limit}.
\item\label{it:well defined}\label{it:index locally determined}
There is a constant~$r_1$, depending only on $f(r)$, with the following property:
Let $\alpha'$ be another one-dimensional ALPU with $f(r)$-tails.
Then, for any interval $X$ with $\abs X = \diam(X) \geq r_1$,
\begin{align*}
  \norm{(\alpha - \alpha')|_{\A_{X}}} \leq \frac1{384} \implies \ind(\alpha)=\ind(\alpha').
\end{align*}
In particular, the index is completely determined by $\alpha|_{\A_X}$ for any such~$X$.
\item\label{it:index locally}
If $f(r) = \littleO(\frac{1}{r})$ then there exist a constant~$r_2$, depending only on~$f(r)$ and the local Hilbert space dimensions~$d_n$, such that the index may also be computed locally as in~\eqref{eq:mi index local},
\begin{align*}
  \ind(\alpha) = \round_{\ZZ[\{\log(p_i)\}]} \frac12\left( I(L_1':R_1)_{\phi} - I(L_1:R_1')_{\phi} \right),
\end{align*}
where $\phi$ denotes the Choi state, the intervals $L_1,R_1,L_1',R_1'$ must be of size at least~$r_2$, and the notation means that we round to the nearest value in $\ZZ[\{\log(p_i)\}]$.
\item\label{it:index as mi} If $f(r) = \bigO(\frac{1}{r^{1+\delta}})$ for some $\delta > 0$, then the index can also be computed as in~\eqref{eq:mi index for alpus with enough decay}, by
\begin{align*}
\ind(\alpha) = \frac12\left(I(L':R)_{\phi} - I(L:R')_{\phi}\right),
\end{align*}
where both $I(L':R)_{\phi}$ and $I(L:R')_{\phi}$ are finite.
\end{enumerate}
In both calculations~\ref{it:index locally} and~\ref{it:index as mi} of the index, the cut defining the regions $L,R$ may be chosen anywhere on the chain.
\end{thm}
\begin{proof}
Throughout this proof, the implicit constants in the~$\bigO$ notation are allowed to depend on the tails~$f(r)$.

\ref{it:index as limit} and \ref{it:index locally determined}:~To see that $\ind(\beta_j)$ stabilizes at large $j$ and hence the limit~\eqref{eq:QCA index limit} exists, consider~$\beta_j$ and~$\beta_{j+1}$.
After coarse-graining by blocking $2(j+1)$ sites, both $\beta_j$ and $\beta_{j+1}$ are nearest neighbor.
Moreover, on any subset $X_j$ that consists of two neighboring coarse-grained sites,
\begin{align}\label{eq:beta_j close}
  \norm{(\beta_j - \beta_{j+1})|_{\A_{X_j}}} \leq \norm{(\beta_j-\alpha)|_{\A_{X_j}}} + \norm{(\alpha - \beta_{j+1})|_{\A_{X_j}}} =\bigO(f(j))
\end{align}
by \cref{eq:ALPU QCA distance in dfn}.
Since $f(r) = \littleO(1)$ this implies that $\norm{(\beta_j - \beta_{j+1})|_{\A_{X_j}}}$ approaches zero as $j \to \infty$.
By \cref{prop:nearby-qcas} this implies that $\ind(\beta_j) = \ind(\beta_{j+1})$ for sufficiently large $j \geq j_0$, where the constant~$j_0$ can be taken as the minimum~$j$ such that the right-hand side of \cref{eq:beta_j close} remains below~$\frac1{192}$.
Thus we conclude that the limit~\eqref{eq:QCA index limit} exists and equals $\ind(\beta_j)$ for $j\geq j_0$.
Moreover, $\ind(\alpha) \in \ZZ[\{\log(p_i)\}]$, since the same is true for the index of the QCAs~$\beta_j$.

To conclude the proof of~\ref{it:index as limit}, we still need to argue that the index is well-defined.
We will demonstrate this together with~\ref{it:index locally determined}.
Consider an ALPU $\alpha'$ that also has $f(r)$-tails, and let~$\beta'_j$ be a corresponding sequence of QCAs as in \cref{dfn:index}.
Note that $\ind(\beta_{j})$ and $\ind(\beta'_j)$ stabilize for~$j\geq j_0$, with the same constant~$j_0$.
We claim that $\ind(\beta_{j}) = \ind(\beta'_j)$ for some (and hence for all)~$j\geq j_0$.
To see this, we consider~$\beta_j$ and~$\beta'_j$ as nearest-neighbor QCAs on a coarse-grained lattice obtained by blocking~$2j$ sites.
Then by \cref{prop:nearby-qcas}, it is sufficient to show $\norm{(\beta_j-\beta'_j)|_{\A_Y}} \leq \frac{1}{192}$ for a region~$Y$ consisting of two neighboring coarse-grained sites.
Note $Y$ then consists of $4j$ sites on the original lattice.
Now,
\begin{align*}
  \norm{(\beta_j-\beta'_j)|_{\A_Y}}
&\leq \norm{(\beta_j-\alpha)|_{\A_Y}}
+ \norm{(\alpha-\alpha')|_{\A_Y}}
+ \norm{(\alpha'-\beta'_j)|_{\A_Y}} \\
&\leq \bigO(f(j)) + \norm{(\alpha-\alpha')|_{\A_Y}}.
\end{align*}
Since $f(r)=\littleO(1)$, we can find~$j_1 \geq j_0$ large enough such that the $\bigO(f(j))$ term is smaller than~$\frac1{384}$.
Take~$r_1:=8j_1$ to ensure that any interval~$X$ with $r_1$ sites contains two neighboring sites of the coarse-grained lattice, so that $\norm{(\alpha-\alpha')|_{\A_Y}} \leq \frac1{384}$ by assumption.
Then, $\norm{(\beta_{j_1}-\beta'_{j_1})|_{\A_Y}} \leq \frac1{192}$, and \cref{prop:nearby-qcas} implies that $\ind(\beta_j) = \ind(\beta'_j)$ for $j=j_1$ and hence for all $j\geq j_0$.
This implies that the index is well-defined (take $\alpha=\alpha'$), concluding the proof of~\ref{it:index as limit}, and it also establishes~\ref{it:index locally determined}.

\ref{it:index locally}~Let $L_j = \{-2j + 1,\ldots, 0\}$ and $R_j = \{1, \ldots,2j\}$.
Since $\beta_j$ is a QCA of radius $2j$, by \cref{prop:mi formula index} we can compute
\begin{align}\label{eq:ind beta_j via MI}
  \ind(\beta_j) = \frac12\left(I(L_j':R_j)_{\phi_j} - I(L_j:R_j')_{\phi_j}\right)
\end{align}
where $\phi_j = (\beta_j^{\dagger} \ot \id)(\omega)$, with $\omega$ a maximally entangled state on $\A_\ZZ \ot \A_\ZZ$.
We let
\begin{align}\label{eq:ind alpha via MI_j}
  \widetilde{\ind}_j(\alpha) := \widetilde{\ind}_{L_j,R_j}(\alpha) = \frac12\left(I(L_j':R_j)_{\phi} - I(L_j:R_j')_{\phi}\right)
\end{align}
as in~\eqref{eq:tilde index}, where $\phi = (\alpha^{\dagger} \ot \id)(\omega)$.
By \cref{eq:ALPU QCA distance in dfn}, $\norm{(\alpha - \beta_j)|_{A_{X_j}}} = \bigO(f(j))$, where $X_j = L_j \cup R_j$.
Thus \cref{lem:crude index continuity} shows that
\begin{align}\label{eq:index diff from QCA}
  \abs[\big]{\ind(\beta_j) - \widetilde{\ind}_j(\alpha)} &
= \bigO\bigl(j f(j)\log(d) + f(j)\log\tfrac1{f(j)}\bigr)
\end{align}
where $d = \max_n d_n$ is the maximum of the local Hilbert space dimensions associated to~$\A_\ZZ$.
Assuming that $f(j) = \littleO(\frac{1}{j})$ the above approaches zero as $j\to\infty$.
Because the sequence $\ind(\beta_j)$ stabilizes to $\ind(\alpha)$ by definition in~\eqref{eq:QCA index limit}, this implies that
\begin{align}\label{eq:ind_j alpha to alpha}
  \lim_{j\to\infty} \widetilde{\ind}_j(\alpha) = \ind(\alpha).
\end{align}
Since $\ind(\alpha)$ takes values in the nowhere dense set $\ZZ[\{\log(p_i)\}]$, rounding $\widetilde{\ind}_j(\alpha)$ must yield $\ind(\alpha)$ for sufficiently large~$j$, proving~\ref{it:index locally}.

\ref{it:index as mi}~Even though the quantities in \cref{eq:ind beta_j via MI,eq:ind alpha via MI_j} converge with~$j$, we have not yet shown that the individual mutual information terms converge. We will show this next, assuming that~$f(r) = \bigO(\frac{1}{r^{1+\delta}})$ for some $\delta > 0$.
We consider the subsequence $\{\beta_{2^k}\}$.
Then, by \cref{eq:ALPU QCA distance in dfn}
\begin{align}\label{eq:exp beta distance}
  \norm{(\beta_{2^k} - \alpha)|_{\A_{X_{2^{k+1}}}}} = \bigO(f(2^k)),
\end{align}
and thus
\begin{align*}
  \norm{(\beta_{2^k} - \beta_{2^{k+1}})|_{\A_{X_{2^{k+1}}}}} = \bigO(f(2^k)).
\end{align*}
Hence by \cref{lem:crude index continuity}, noting that $I(L_{2^{k+1}}':R_{2^{k+1}})_{\phi_{2^k}} = I(L_{2^{k}}':R_{2^{k}})_{\phi_{2^k}}$ since $\beta_{2^k}$ has radius~$2^{k+1}$, as similarly observed in the proof of \cref{prop:mi formula index}, this implies
\begin{align*}
  \abs{I(L_{2^{k}}':R_{2^{k}})_{\phi_{2^k}} - I(L_{2^{k+1}}':R_{2^{k+1}})_{\phi_{2^{k+1}}}} &= \bigO(2^kf(2^k) \log(d)+ f(2^k)\log\tfrac1{f(2^k)})\\
  &= \bigO(2^{-\delta k}).
\end{align*}
Thus $I(L_{2^{k}}':R_{2^{k}})_{\phi_{2^k}}$ is a Cauchy sequence and hence converges.
Moreover, by \cref{lem:crude index continuity}, \cref{eq:exp beta distance} also implies that
\begin{align*}
  \abs{I(L_{2^{k}}':R_{2^{k}})_{\phi} - I(L_{2^{k}}':R_{2^{k}})_{\phi_{2^{k}}}} &= \bigO(2^kf(2^{k})\log(d) + f(2^{k})\log\tfrac1{f(2^{k})})\\
  &= \bigO(2^{-\delta k}).
\end{align*}
Thus $I(L_{2^{k}}':R_{2^{k}})_{\phi}$ also converges, with the same limit as $I(L_{2^{k}}':R_{2^{k}})_{\phi_{2^k}}$.
Then using \cref{prop:mi vN}, this implies that
\begin{align*}
  I(L':R)_{\phi}
= \lim_{k \rightarrow \infty} I(L_{2^{k}}':R_{2^{k}})_{\phi}
= \lim_{k \rightarrow \infty} I(L_{2^{k}}':R_{2^{k}})_{\phi_{2^k}}
\end{align*}
is finite.
A similar argument shows that $I(L:R')_{\phi}$ is finite and can be computed as
\begin{align*}
  I(L':R)_{\phi}
= \lim_{k \rightarrow \infty} I(L_{2^{k}}:R_{2^{k}}')_{\phi}
= \lim_{k \rightarrow \infty} I(L_{2^{k}}:R_{2^{k}}')_{\phi_{2^k}}.
\end{align*}
It follows that
\begin{align*}
  \ind(\alpha) = \frac12\left(I(L':R)_{\phi} - I(L:R')_{\phi}\right),
\end{align*}
as a consequence either of \cref{eq:ind_j alpha to alpha} or of \cref{eq:QCA index limit}.

In parts~\ref{it:index locally} and~\ref{it:index as mi} we took the cut between $L$ and $R$ to be at $n=0$, but the index may be calculated using regions translated anywhere along the chain, which follows from the same fact for the QCAs $\beta_j$.
\end{proof}

\noindent
The proof of \cref{thm:index alpu after reorg} also shows that in part~\ref{it:index as mi}, the two mutual information quantities can be computed as limits of corresponding mutual information quantities for finite intervals.

\subsection{Properties of the index for ALPUs}
In this section we will show that the index for ALPUs defined in \cref{thm:index alpu after reorg} inherits essentially all properties of the GNVW index for QCAs stated in \cref{thm:properties index}.

We first use \cref{thm:qca approx} to construct a path between any ALPU $\alpha$ with $\ind(\alpha) = 0$ and the identity automorphism $\id$, using a one-parameter family of ALPUs $\beta[t]$ for $t \in [0,1]$, with $\beta[0]=\id$ and $\beta[1] = \alpha$.
The path will be \emph{strongly continuous}, in the sense that for all $x \in \A_{\ZZ}, t_0\in [0,1]$,
\begin{align}\label{eq:strongly continuous}
\lim_{t \to t_0} \norm{\alpha[t](x)-\alpha[t_0](x)}=0.
\end{align}

\begin{figure}
\centering
\begin{overpic}[width=0.6\textwidth,grid=false]{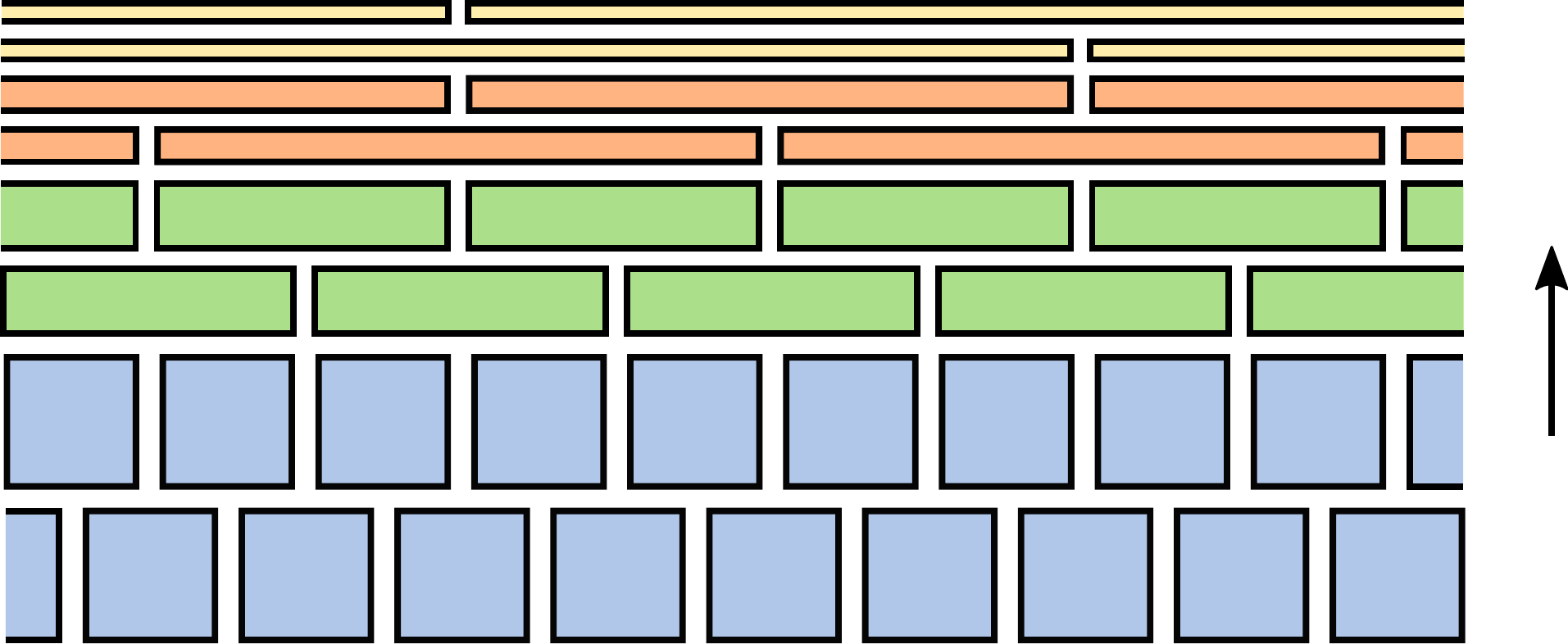}
\put(46,3){$v^{k_0,1}$}
\put(50.5,13){$v^{k_0,2}$}
\put(103,18){$\alpha[t]$}
\end{overpic}
\caption{Illustration of the construction of $\alpha[t]$ in \cref{thm:index 0 path}.  The evolution consists of successive evolutions by different time-independent Hamiltonians, depicted as successive layers, with interaction terms increasing in diameter but decreasing in strength.\label{fig:index 0 path}}
\end{figure}

\begin{thm}[Continuous deformations]\label{thm:index 0 path}
If $\alpha$ is a one-dimensional ALPU with $f(r)$-tails with $\ind(\alpha) = 0$, then there exists a strongly continuous path $\alpha(t)$ with $\alpha[0]=\id$ and $\alpha[1] = \alpha$ such that $\alpha[t]$ has $g(r)$-tails for all $t$, for some $g(r) = \bigO(f(Cr))$ and some universal constant $C>0$.
Moreover, this path may be given by a time evolution using a time-dependent Hamiltonian~$H(t)$ evolving for unit time.
For every $t < 1$ there exists $l$ such that the Hamiltonian~$H(t)$ has only terms $H_X$ on (nonoverlapping) sets $X$ of diameter at most $16l$, and it holds that $\norm{H_{X}(t)} = \bigO(f(l)\log(l))$.
\end{thm}
\noindent The above $H(t)$ is constructed as piecewise-constant for $t \in [0,1)$.
The idea of the proof is that we continuously interpolate between consecutive QCAS $\beta_{2^j}$ as constructed in \cref{thm:qca approx}. For large $j$ we need to use a Hamiltonian with a correspondingly large support to interpolate between $\beta_{2^j}$ and $\beta_{2^{j+1}}$, but on the other hand $\beta_{2^j}$ and $\beta_{2^{j+1}}$ are locally close, so the interaction strength is small.
As we increase $j$, we ``speed up'' the interpolation, so we get to $\alpha$ in unit time.
In particular, the Hamiltonian is piecewise constant on time intervals that decrease in size as $t$ goes to 1, and the support of the Hamiltonian increases as $t$ goes to 1.
This procedure is illustrated in \cref{fig:index 0 path} and leads to the given bound on the terms $H_X(t)$ of the Hamiltonian.  For $f(r)=\bigO((\log r)^{-1})$ the norms $\norm{H_X(t)}$ are uniformly bounded as $t \to 1$; more generally for $f(r)=\littleO(1)$ the terms may diverge in norm as $t \to 1$, but nonetheless the path $H(t)$ is strongly continuous on $[0,1]$.
Of course, the path~$\alpha(t)$ and associated Hamiltonian $H(t)$ are not unique; we just provide one particular construction.

\begin{proof}
We apply \cref{thm:qca approx} to obtain a sequence of QCA approximations $\beta_j$ of radius $2j$ with error $\norm{(\alpha-\beta_j)|_{\A_X}} = \bigO(f(j) \ceil[\big]{\tfrac{\diam(X)}{j}})$ as $j \to \infty$.
Therefore
\begin{align*}
\norm{(\beta_{2j}-\beta_j)|_{\A_X}} \leq \norm{(\alpha-\beta_j)|_{\A_X}} + \norm{(\alpha-\beta_{2j})|_{\A_X}} = \bigO(f(j)\ceil[\big]{\tfrac{\diam(X)}{j}}),
\end{align*}
having used that $f$ is non-increasing, and hence
\begin{align*}
\norm{(\beta_{2j}\beta_j^{-1}- \id)|_{\A_X}} &\leq \norm{(\beta_{2j} - \beta_j)|_{\A_{B(X,2j)}}} = \bigO(f(j)\ceil[\big]{\tfrac{\diam(X)}{j}})
\end{align*}
We can therefore define QCAs
\begin{align*}
\gamma_k & = \beta_{2^{k+1}}\beta_{2^k}^{-1}
\end{align*}
which have at most radius $R_k = 2^{k+3}$ and satisfy $\norm{(\gamma_k- \id)|_{\A_X}} = \bigO(f(2^k))$ for $\diam(X)\leq 2R_k = 2^{k+4}$.
For sufficiently large $k\geq k_0$, $\ind(\beta_{2^k})=\ind(\alpha) = 0$, and hence $\ind(\gamma_k)=0$.

By \cref{thm:gnvw} and \cref{thm:properties index}, any index-0 QCA of radius $R$ can be decomposed as a two-layer circuit with unitaries on blocks of diameter $2R$.
If the QCA is $\epsilon$-near the identity when restricted to intervals of size $2R$, the individual unitaries in the circuit are $\bigO(\epsilon)$-near the identity by \cref{prop:lpu close to identity}.
Therefore $\gamma_k$ may be implemented by a two-layer unitary circuit for $k\geq k_0$.
We proceed to describe this circuit as a Hamiltonian evolution, with a different time-independent Hamiltonian generating each layer, in the straightforward way.
To be precise, from \cref{prop:lpu close to identity} we obtain that
\begin{align*}
  \gamma_k(x) = (v^{(k,2)})^* (v^{(k,1)})^* x v^{(k,1)} v^{(k,2)}
\end{align*}
where for each layer $a\in\{1,2\}$
\begin{align*}
  v^{(k,a)} = \prod_n v^{(k,a)}_n
\end{align*}
and where the $\{v^{(k,a)}_n\}_n$ are unitary gates acting on disjoint regions of diameter $2R_k = 2^{k+4}$.
Moreover, each gate satisfies $\norm{v^{(k,a)}_n-I} = \bigO(f(2^k))$.
Each gate is generated by a Hamiltonian $H^{(k,a)}_n = -i\log(v^{(k,a)}_n)$, defined using the principal logarithm, with $\norm{H^{(k,a)}_n}= \bigO(f(2^k))$.
Let $H^{(k,a)} = \sum_n H^{(k,a)}_n$ denote the total Hamiltonian generating the $a$-th layer.
Then we can define a Hamiltonian evolution $\gamma_k[t]$ for $t \in [0,1]$ with $\gamma_k[0]=I$, $\gamma_k[1] = \gamma_k$:
\begin{align*}
\gamma_k[t](x) = e^{2 i H^{(k,1)} t} (x)  e^{-2i H^{(k,1)} t}
\end{align*}
for $t \in [0,\frac{1}{2}]$ and
\begin{align*}
  \gamma_k[t](x) = e^{2i H^{(k,2)} (t-\frac12)} e^{i H^{(k,1)}}(x)  e^{-i H^{(k,1)} }e^{-2i H^{(k,2)} (t-\frac12)}
\end{align*}
for $t \in (\frac{1}{2},1]$.
Note that the gates implementing $\gamma_k[t]$ are all of the form $(v^{(k,a)}_n)^{s}$ for some $s \in [0,1]$.
From this it is clear that $\gamma_k[t]$ defines a strongly continuous path and the evolution is gentle in the sense that $\gamma_k[t]$ never strays far from $\id$:
\begin{align*}
\norm{(\gamma_k[t] - \id)|_{\A_X}} & = \bigO(f(2^k)).
\end{align*}
for $\diam(X) \leq 2R_k$.
By construction, $\gamma_k[t]$ is a QCA with radius at most $3R_k$ for every $t \in [0,1]$.

We let $\alpha_{k+1}[t] := \gamma_k[t] \beta_{2^k}$, which is a strongly continuous path with $\alpha_{k+1}[0] = \beta_{2^k}$ and $\alpha_{k+1}[1] = \beta_{2^{k+1}}$.
For all $t \in [0,1]$
\begin{align}\label{eq:path error}
  \norm{(\alpha_{k+1}[t] - \alpha)|_{\A_X}} \leq \norm{(\gamma_{k} - \id)|_{\A_{B(X,2^{k+1})}}} + \norm{(\alpha - \beta_{2^k})|_{\A_X}} = \bigO(f(2^k))
\end{align}
for $\diam(X) \leq R_k$.
Moreover $\alpha_{k+1}[t]$ is a QCA with radius $3R_k + 2^{k+1} \leq 4R_k$.

We defined $\alpha_k[t]$ only for $k > k_0$.
Let $\alpha_{k_0}[t]$ be the Hamiltonian evolution implementing the index-0 QCA $\beta_{2^{k_0}}$ for $t \in [0,1]$, in the same way we defined $\gamma_k[t]$.
Let
\begin{align*}
  T = \sum_{k=0}^{\infty} \frac{1}{1 + k^2} , \quad t_k = \sum_{l=0}^{k - 1} \frac{1}{T(1 + l^2)}.
\end{align*}
We define $\alpha[t]$ by gluing together the $\alpha_k[t]$, ``speeding up'' $\alpha_{k_0 + k}$ by a factor $T(k^2 + 1)$ in order to make this a unit time evolution:
\begin{align*}
  \alpha[t] = \alpha_{k_0 + k}\left[ \frac{t - t_k}{T(k^2 + 1)}\right] \text{  if  } t \in (t_k, t_{k+1})
\end{align*}
for $t \in [0,1)$ and $\alpha[1] = \alpha$.
The construction of the path $\alpha[t]$ is illustrated in \cref{fig:index 0 path}.
Going through~$\gamma_{k_0 + k}$ faster by a factor $T(k^2 + 1)$ is equivalent to rescaling the Hamiltonian by $T(k^2 + 1)$, and is still strongly continuous.
Hence $\alpha[t]$ is strongly continuous for $t \in [0,1)$.
The strong continuity at $t=1$ follows from the fact that the sequence $\beta_{2^k}$ converges strongly to $\alpha$.
Indeed, let $x \in \A_X$ for finite $X$.
Then consider $k$ such that $\diam(X) \leq R_{k_0 + k}$, then we see that for $l \geq k_0 + k$ $\norm{\alpha_{l+1}[s](x) - \alpha(x)} = \bigO(f(2^l))$ for $s \in [0,1]$.
Hence, $\norm{\alpha[t](x) - \alpha(x)}$ goes to zero as $t \rightarrow 1$.
As in the proof of \cref{thm:qca approx} we see that the same holds for general $x \in \A_{\ZZ}$.
Moreover, we see that at each point in time the Hamiltonian will have terms~$H_X$ with support of diameter~$16l = 2^{k+4}$ for some~$k$ with~$\norm{H_X} = \bigO(f(2^k)k^2) = \bigO(f(l)\log(l))$.

Finally, we need to show that $\alpha[t]$ has uniform tail bounds for $t \in (0,1)$.
(We already have tail bounds at the initial and final time.)
Let $X\subseteq\ZZ$ be an arbitrary (finite or infinite) interval.
Take some~$r > 4R_{k_0} = r_0$.
There will be some $k$ such that~$4R_{k} \leq r < 4R_{k+1}$, and there will be some~$l$ and~$s \in [0,1]$ such that $\alpha[t] = \alpha_{l+1}[s]$.
If~$k \geq l$, by construction $\alpha[t](\A_X) \subseteq \A_{B(X,4R_{l})} \subseteq \A_{B(X,r)}$.
On the other hand, suppose that $k < l$.
Write~$X = X_1 \cup X_2$ where~$X_1$ is the (possibly empty set) of elements with distance from the boundary larger than $4R_{l}$.
Then $\alpha_{l+1}[s](\A_{X_1}) \subset \A_X$.
Moreover, since $X_2$ consists of at most two intervals of size $4R_{l}$ we have, using \cref{lem:homomorphism local error} and \eqref{eq:path error} that $\norm{(\alpha - \alpha_{l+1}[s])|_{\A_{X_2}}} = \bigO(f(2^{l}))$.
Since $\alpha$ has $f(r)$-tails,
\begin{align*}
  \alpha(\A_{X_2}) \overset{\bigO(f(r))}{\subseteq} \A_{B(X,r)},
\end{align*}
 and since $r < 4R_{l} = 2^{l + 5}$ we see that
\begin{align*}
  \alpha_{l+1}[s](\A_{X_2}) &\overset{\bigO(f(r) + f(2^{l}))}{\subseteq} \A_{B(X,r)}\\
  \alpha_{l+1}[s](\A_{X_2}) &\overset{\bigO(f(\frac{r}{32}))}{\subseteq} \A_{B(X,r)}.
\end{align*}
\cref{lem:simultaneous near inclusions} allows us to conclude that
\begin{align*}
  \alpha[t](\A_{X}) = \alpha_{l+1}[s](\A_{X}) \overset{\bigO(f(\frac{r}{32}))}{\subseteq} \A_{B(X,r)}.
\end{align*}
\end{proof}

\begin{rmk}
If $\alpha$ has $\bigO(\frac{1}{r^{1+a}})$-tails for $a > 0$, then for $0 < b < a$ and reproducing function $F(r) = \frac{1}{(1 + r)^{1 + b}}$ the Hamiltonian constructed in \cref{thm:index 0 path} satisfies the hypotheses in \cref{thm:lr} (Lieb-Robinson).
However, notice that the locality estimates you get from applying the Lieb-Robinson bounds to these bounds are weaker than the original locality bounds on $\alpha[t]$.
\end{rmk}

\begin{rmk}
The Hamiltonian evolution constructed in \cref{thm:index 0 path} cannot always be approximated by a 2-local quantum circuit of constant depth.
Likewise, even QCAs of radius $r$ may have circuit complexity exponential in $r$ when using 2-local gates.
\end{rmk}

\begin{rmk}\label{rmk:exp tails}
Finally, we observe that in \cref{thm:index 0 path} if we have exponential tails, with $f(r) = \bigO(e^{-ar})$, one obtains that $\alpha[t]$ has $\bigO(e^{-aCr})$-tails.
This is not entirely optimal, and for exponential tails one can slightly change the proof, by considering the sequence $\beta_k$ rather than $\beta_{2^k}$ and correspondingly $\gamma_k = \beta_{k+1}\beta_{k}^{-1}$ instead of $\gamma_k = \beta_{2^{k+1}}\beta_{2^k}^{-1}$.
The same arguments as in the proof of \cref{thm:index 0 path} then lead to a path $\alpha[t]$ with $\bigO(f(r + C)) = \bigO(e^{-ar})$-tails, which is implemented by a Hamiltonian $H(t)$.
In this case the Hamiltonian is such that for every $t$, there exists $k$ such that $H(t)$ has only terms $H_X$ on (nonoverlapping) sets $X$ of diameter at most $k$, with $\norm{H_X(t)} = \bigO(k^2e^{-ak})$.
\end{rmk}

Next we discuss blending.
We need a slightly weaker notion than for QCAs.

\begin{dfn}\label{dfn:blending}
Two ALPUs $\alpha_1$ and $\alpha_2$ in one dimension can be \emph{blended} (at the origin) if there exists an ALPU $\beta$ on some $\A^{\vn}_\ZZ$ such that
\begin{align*}
  \lim_{r \to \infty} \norm{(\beta-\alpha_1)|_{\A_{\leq -r}}}& = 0, \\
  \lim_{r \to \infty} \norm{(\beta-\alpha_2)|_{\A_{\geq r}}} &=0.
\end{align*}
\end{dfn}

\begin{prop}\label{prop:blend}
 Two ALPUs $\alpha_1, \alpha_2$ can be blended if and only if $\ind(\alpha_1) = \ind(\alpha_2)$.
\end{prop}

\noindent When $\ind(\alpha_1) = \ind(\alpha_2)$ and both ALPUs have $f(r)$-tails, the approximation requirement of the blending as defined in \cref{dfn:blending} can be refined as in~\eqref{eq:blend error} as discussed in the proof.
The blending proceeds similarly to the construction in \cref{prp:approximation almost nn}.

\begin{proof}
First we assume $\alpha_1$ and $\alpha_2$ can be blended and show $\ind(\alpha_1) = \ind(\alpha_2)$.
Consider a blended ALPU $\beta$ as in \cref{dfn:blending}.
By \cref{thm:index alpu after reorg}\ref{it:index locally determined}, one may compute $\ind(\beta)$ locally on either half of the blended chain.
Both calculations must yield the same index, which does not depend on where it is locally calculated.
By \ref{it:index locally determined} of \cref{thm:index alpu after reorg}, the index computed locally at the sufficiently far left must be $\ind(\alpha_1)$, and the index computed at the far right must be $\ind(\alpha_2)$.
Thus, $\ind(\alpha_1) = \ind(\alpha_2)$.

Next we show that if $\ind(\alpha_1) = \ind(\alpha_2)$, then $\alpha_1$ and $\alpha_2$ can be blended.
We assume both ALPUs are defined on the same $\A_\ZZ$ (i.e., the local dimensions are the same) and address the general case afterward.
Both ALPUs extend to automorphisms of $\A^{\vn}_\ZZ$ as in \cref{rmk:ALPU_for_vN}.
Coarse-grain the lattice until both $\alpha_1$ and $\alpha_2$ are $\eps$-nearest neighbor ALPUs, with $\eps$ smaller than a universal constant determined by the remainder of the proof.
Then we can apply \cref{prop:single patch localized} (if $\eps \leq \eps_1$) separately to $\alpha_1$ and $\alpha_2$ at site $n=0$.
Denote the ALPUs resulting from \cref{prop:single patch localized} as~$\tilde{\alpha}_1$ and~$\tilde{\alpha}_2$, respectively.
Then by construction $\norm{\tilde{\alpha}_i-\alpha_i}\leq C_1 \eps$ for $i=1,2$.
Moreover by \cref{thm:index alpu after reorg}\ref{it:index locally determined}, we can take $\eps$ small enough that $\ind(\alpha_i)=\ind(\tilde{\alpha}_i)$ for $i=1,2$, hence $\ind(\tilde{\alpha}_1)= \ind(\tilde{\alpha}_2)$.

As usual, we write~$\B_n = \A_{\{2n,2n+1\}}$ and~$\C_n = \A_{\{2n-1,2n\}}$.
Then by their construction, $\tilde{\alpha}_i$ for $i=1,2$ both satisfy the locality properties~$\tilde{\alpha}_i(\B_n) \subseteq \C_n \ot \C_{n+1}$ for $n = 0,1,2$, as well as $\tilde{\alpha}^{-1}_i(\C_n)\subseteq \B_{n-1} \ot \B_n$ for $n=1,2$.
Then by \cref{thm:gnvw} and subsequent \cref{rem:weaker locality properties}, for each~$i=0,1$ and~$n=1,2$ we can define
\begin{align*}
  \L_{n}^{(i)} &=\tilde{\alpha}_i(\B_{n}) \cap  \C_{n}, \\
  \R_{n-1}^{(i)} &=  \tilde{\alpha}_i(\B_{n-1}) \cap  \C_{n}
\end{align*}
such that
\begin{align}\label{eq:factorize blend}
  \C_{n} = \L_{n}^{(i)} \ot \R_{n-1}^{(i)}
\end{align}
and, for $n=1$,
\begin{align}\label{eq:factorize blend 2}
  \B_n = \tilde{\alpha}_i^{-1}(\L^{(i)}_{n}) \ot \tilde{\alpha}_i^{-1}(\R^{(i)}_{n}).
\end{align}
Following the structure theory of QCAs in \cref{thm:gnvw}, one can for each $i=1,2$ find a QCA $\beta_i$ of radius 2 such that $\beta_i|_{\A_{\{0,\dots,5\}}}=\tilde{\alpha}_i|_{\A_{\{0,\dots,5\}}}$.%
\footnote{Indeed, $\beta_i$ can be constructed as follows. First we define $\beta_i|_{\A_{\{0,\dots,5\}}}=\tilde{\alpha}_i|_{\A_{\{0,\dots,5\}}}$ and then we define the action of~$\beta_i$ on the remainder of $\A_{\ZZ}$ as follows. We focus on defining $\beta_i$ for $\A_{\leq -1}$; the definition for $\A_{\geq 6}$ is directly analogous.
An easy argument shows that $\L_0^{(i)} := \tilde{\alpha}_i(\B_0) \cap \C_0$ is a factor and moreover \cref{eq:factorize blend 2} also holds for $n=0$.
Then \cref{eq:factorize blend} will also hold for~$n=0$ if we define $\R^{(i)}_{-1}$ alternatively as the complementary factor to $\L_0^{(i)} \subset \C_0$.
For each~$n \leq -1$, choose an arbitrary factorization~$\C_{n} = \L_{n}^{(i)} \ot \R_{n-1}^{(i)}$ with $\L_{n}^{(i)} \cong \L_{1}^{(i)}$ and $\R_{n-1}^{(i)} \cong \R_{-1}^{(i)} \cong \R_{0}^{(i)} \cong \R_{1}^{(i)}$ (using that, by assumption, all local dimensions are the same).
For each $n \leq -1$, choose an arbitrary factorization~$\B_n = \tilde{\L}^{(i)}_n \ot \tilde{\R}^{(i)}_n$ into factors isomorphic to those used for $\B_1$ in \cref{eq:factorize blend 2} for $n=1$.
Then we have~$\tilde{\L}^{(i)}_n \cong \L^{(i)}_n$ and~$\tilde{\R}^{(i)}_n \cong \R^{(i)}_n$, so we can define $\beta_i$ to act as $\beta_i(\tilde{\L}^{(i)}_n) = \L^{(i)}_n$ and $\beta_i(\tilde{\R}^{(i)}_n) = \R^{(i)}_n$ for~$n\leq-1$. This completes the definition of~$\beta_i$ for~$\A_{\leq -1}$.}
By the latter condition, \cref{thm:index alpu after reorg},~\ref{it:index locally determined} implies (if we have sufficiently coarse-grained in our initial step) that $\ind(\tilde{\alpha}_i)=\ind(\beta_i)$.
Recalling that $\ind(\tilde{\alpha}_1)= \ind(\tilde{\alpha}_2)$, we then have $\ind(\beta_1)= \ind(\beta_2)$, so from \cref{eq:LPU index} we conclude that~$\R_0^{(1)}$ and~$\R_0^{(2)}$ have the same dimension and hence are isomorphic finite-dimensional subalgebras of~$\C_1$.
Hence there exists a unitary~$u \in \C_1$ such that $u\R_0^{(1)}u^*=\R_0^{(2)}$.

Now we are in position to define the blended ALPU $\beta$.
Let $\beta(x) = u\tilde{\alpha}_1(x) u^*$ for $x \in \A_{\leq 1}$, and let $\beta|_{\A_{\geq 2 }} = \tilde{\alpha}_2|_{\A_{\geq 2}}$.
Then $\beta(\A_{\leq 1})=(\A_{\leq 0}\cup \R_0^{(2)} )''$ and $\beta(\A_{\geq 2})=(\L_1^{(2)} \cup \A_{\geq 3})''$ commute by construction, so $\beta$ is a well-defined injective unital $*$-homomorphism.
Moreover $\C_1 =  \L_{1}^{(2)} \ot \R_{0}^{(2)}$ from~\eqref{eq:factorize blend}, so $\beta$ is surjective, hence a well-defined ALPU.
By construction of $\tilde{\alpha}_1$ and $\tilde{\alpha}_2$ using \cref{prop:single patch localized}, for~$r \to \infty$,
\begin{equation}\label{eq:blend error}
\begin{aligned}
\norm{(\beta-\alpha_1)|_{\A_{\leq -r-1}}}& = \mc{O}(f(r - 1)), \\
\norm{(\beta-\alpha_2)|_{\A_{\geq r+6}}} &=\mc{O}(f(r - 7)).
\end{aligned}
\end{equation}
Above we assumed both $\alpha_1$ and $\alpha_2$ were defined on the same $\A_\ZZ$ (i.e., that both chains use algebras $\A_n$ of the same dimensions).
If $\alpha_1$ and $\alpha_2$ have different local dimensions, then in the region where we blend them above, we can first pad them with extra tensor factors so that they have identical local dimensions within that region.
\end{proof}

The following theorem extends all properties in \cref{thm:properties index} for QCAs to ALPUs, replacing the role of circuits by Hamiltonian evolutions, and allowing strongly continuous paths through the space of ALPUs with uniform tail bounds.

\begin{thm}[Properties of index for ALPUs]\label{thm:properties index alpu}
Suppose $\alpha$ and $\beta$ are ALPUs in one dimension. Then:
\begin{enumerate}
\item\label{it:index alpu multiplicativity} $\ind(\alpha \ot \beta) = \ind(\alpha) + \ind(\beta)$.
\item\label{it:index alpu composition} If $\alpha$ and $\beta$ are defined on the same algebra, $\ind(\alpha\beta) = \ind(\alpha) + \ind(\beta)$.
\item\label{it:index alpu same eqv} The following are equivalent:
\begin{enumerate}
\item\label{it:alpus have same index} $\ind(\alpha) = \ind(\beta)$.
\item\label{it:alpus can be blended} $\alpha$ and $\beta$ may be blended.
\item\label{it:alpus are related by index 0} There exists an index-0 ALPU $\gamma$ such that $\alpha = \beta \gamma$.
\item\label{it:alpus are strongly connected} There exists $g(r)=\littleO(1)$ and a strongly continuous path $\alpha[t]$ through the space of ALPUs with $g(r)$-tails such that $\alpha[0] = \alpha$ and $\alpha[1]=\beta$.
\end{enumerate}
In~\ref{it:alpus are strongly connected}, if~$\alpha$ and~$\beta$ have $f(r)$-tails, we may take $g(r)=\bigO(f(Cr))$ for a universal constant~$C$.
If they have $\bigO(e^{-ar})$-tails, we may take $g(r) = \bigO(e^{-ar})$.
In general, the path in~\ref{it:alpus are strongly connected} may be implemented by composing $\alpha$ \textup{(}or $\beta$\textup{)} with a Hamiltonian evolution with time-dependent Hamiltonian $H(t)$, with interactions bounded as in \cref{thm:index 0 path,rmk:exp tails}.
\end{enumerate}
\end{thm}
\noindent In~\ref{it:alpus are related by index 0}, if $\alpha$ and $\beta$ do not have the same local dimensions, the statement only holds after separately tensoring $\alpha$ and $\beta$ with appropriate identity automorphisms, such that $\alpha$ and $\beta$ then have the same local dimensions.  The analogous modification is needed for~\ref{it:alpus are strongly connected}.
\begin{proof}
If $\alpha$ and $\beta$ are ALPUs with approximating sequences $\alpha_n$ and $\beta_n$ as in \cref{thm:qca approx,thm:index alpu after reorg}, then $\alpha_n \ot \beta_n$ and $\alpha_{2n} \beta_n$ approximate $\alpha \ot \beta$ and $\alpha\beta$ respectively.
Then~\ref{it:index alpu multiplicativity} and~\ref{it:index alpu composition} follow from the corresponding property for QCAs (\cref{thm:properties index}).
For~\ref{it:index alpu same eqv} the equivalence $\ref{it:alpus have same index} \Leftrightarrow \ref{it:alpus can be blended}$ is stated by \cref{prop:blend}.
The equivalence $\ref{it:alpus have same index} \Leftrightarrow \ref{it:alpus are related by index 0}$ follows from $\ind(\beta \alpha^{-1}) = \ind(\beta)-\ind(\alpha)$, using property~\ref{it:index alpu composition}.
The implication $\ref{it:alpus have same index} \Rightarrow \ref{it:alpus are strongly connected}$ follows from \cref{thm:index 0 path} applied to $\beta \alpha^{-1}$.
The comment about exponential tails follows from the remark after \cref{thm:index 0 path}.
Next we show $\ref{it:alpus are strongly connected} \Rightarrow \ref{it:alpus have same index}$, i.e.\ that the index must remain constant along a strongly continuous path.
Because all ALPUs in the path are assumed to have $g(r)$-tails for some fixed $g(r)$, by \cref{thm:index alpu after reorg}\ref{it:index locally determined} there exists a finite interval~$X$ such that any two ALPUs $\gamma$ and $\gamma'$ with $\norm{(\gamma-\gamma')|_{\A_X}}$ sufficiently small must have $\ind(\gamma)=\ind(\gamma')$.  By the strong continuity~\eqref{eq:strongly continuous} of the path, the index must then be constant along the path.
\end{proof}

\noindent
In the terminology of \cite{hastings2013classifying}, \cref{thm:properties index alpu} shows that an (A)LPU is an LGU (locally generated unitary) if and only if it has index zero.
We can also interpret \cref{thm:properties index alpu} as a converse to the Lieb-Robinson bounds in one dimension.
Again, recall that Lieb-Robinson bounds demonstrate that local Hamiltonian evolution exhibits an approximate causal cone, quantified by the bound.
Conversely, we ask whether evolutions that satisfy Lieb-Robinson-type bounds (i.e. ALPUs) can be generated by some time-dependent Hamiltonian.
We find the following converse, emphasized below.

\begin{cor}[Converse to Lieb-Robinson bounds]\label{cor:LR-converse}
Suppose $\alpha$ is an ALPU in one dimension with $f(r)$-tails.
If (and only if) $\ind(\alpha)=0$, $\alpha$ can be implemented by a strongly continuous path $\alpha[t]$ generated by some time-dependent Hamiltonian $H(t)$, such that $\alpha[0] = \id$, $\alpha[1] = \alpha$, and $\alpha[t]$ has $g(r)$-tails for all $t$, for some $g(r)=\littleO(1)$.
If~$\alpha$ has $f(r)$-tails, we may take $g(r)=\bigO(f(Cr))$ for a universal constant~$C$.
If it has $\bigO(e^{-ar})$-tails, we may take $g(r) = \bigO(e^{-ar})$.
The Hamiltonian $H(t)$ can be taken to have interactions bounded as in \cref{thm:index 0 path,rmk:exp tails}.

More generally, every ALPU in one dimension is a composition of a shift and a Hamiltonian evolution as above.
\end{cor}
\begin{proof}
The equivalence follows immediately from~\ref{it:index alpu same eqv} in \cref{thm:properties index alpu}.
The final statement follows by letting $\sigma$ be a shift with $\ind(\sigma) = \ind(\alpha)$, then $\ind(\alpha\sigma^{-1}) = 0$ by~\ref{it:index alpu composition} in \cref{thm:properties index alpu} so there exists a Hamiltonian evolution $\gamma$ such that $\gamma = \alpha\sigma^{-1}$ and hence $\alpha = \gamma\sigma$.
\end{proof}

\noindent
The use of time-dependent rather than time-independent Hamiltonians is necessary:~\cite{zimboras2020does} shows that there exist QCAs with index 0 that cannot be implemented by any time-independent local Hamiltonian.

\subsection{Finite chains}\label{sec:finite}
We developed the above structure theory of ALPUs on the infinite one-dimensional lattice.  The statements are easily be adapted to the case of a finite one-dimensional chain with non-periodic (``open'') boundary conditions. The statements as well as the proofs essentially hold unchanged, but we make some clarifying remarks.  In summary, the theorems only become nontrivial when the length~$|\Gamma|$ of the chain is taken larger than some finite threshold, but this threshold depends only on the tails and local dimensions of the ALPU.  Meanwhile, the index is always zero.

We work with the algebra $\A_\Gamma$, where $\Gamma$ is now a finite interval $\Gamma \subset \ZZ$.  By non-periodic boundary conditions, we mean that $\Gamma$ is considered as an interval rather than a circle, i.e.\ $\Gamma$ inherits the metric from $\ZZ$, and the sites at either end of the interval are not considered neighbors.  We again consider ALPUs on $\A_\Gamma$ with $f(r)$ tails, where $f(r)$ is only meaningful for $r < |\Gamma|$.   In our arguments, $\A_{\leq n}$ becomes the finite-dimensional algebra corresponding to all sites left of $n+1$, and so on.

With this modification, \cref{lem:single rotation 0} holds as stated, and the proof is identical.  Importantly, all unspecified constants appearing as $\bigO(\cdot)$ in e.g.~\eqref{eq:asymptotically equal 0 left} are independent of the chain length~$|\Gamma|$.

We then arrive at \cref{thm:qca approx} for finite one-dimensional lattices, describing QCA approximations to ALPUs.  Given ALPU $\alpha$ with $f(r)$ tails, the theorem describes an increasing sequence of QCA approximations $\beta_j$ of radius $j$.  For finite $\Gamma$, we restrict attention to $j \leq |\Gamma|$, so that the notion of a QCA of radius $j$ remains meaningful.  Recall the QCA approximations $\beta_j$ were only guaranteed to have the listed properties in \cref{thm:index alpu after reorg} for $j>j_0$, with $j_0$ chosen such that $f(j_0)$ is smaller than some universal constant independent of $|\Gamma|$.  Then we only need $|\Gamma|>j_0$ for \cref{thm:index alpu after reorg} to yield nontrivial QCA approximations, and this threshold size is determined only by the tails $f(r)$.  Finally, the assumption $f(r)=o(\frac{1}{r})$ used for the latter claims of \cref{thm:index alpu after reorg} may expressed more explicitly as the assumption that $f(j_0)j_0$ is smaller than some constant depending only on the local dimensions $d_n$ of $\A_\Gamma$.  This assumption then increases the minimum length $|\Gamma|$ for the theorem to become nontrivial, but with the minimum depending only on the tails and local dimensions, rather the details of $\alpha$.

While \cref{thm:index alpu after reorg} holds as written for finite~$\Gamma$, it also reduces to a special case: the index is always zero.  Calculating the index as the entropy difference in~\eqref{eq:entropy diff}, we see the entropies correspond to complementary regions of a pure state, yielding zero.  In fact, the trivial index was inevitable.  On the infinite lattice, ALPUs with nonzero index implement shifts, and these shifts have no analog on the finite interval with non-periodic boundary conditions.

We can therefore apply \cref{thm:index 0 path} about Hamiltonian evolutions to every ALPU on finite~$\Gamma \subset \ZZ$. As above, the theorem becomes nontrivial lattices of a certain size, using the same threshold discussed above.  We then obtain a local Hamiltonian evolution generating the ALPU, with locality as specified by \cref{thm:index 0 path}.

While finite chains with non-periodic boundary conditions descend as a special case from the infinite lattice, the case of periodic boundary conditions (i.e.\ $\Gamma$ inherits the metric of a circle) appears more difficult.  Many of the tools we develop appear useful there, but the key \cref{lem:single rotation 0} has no obvious analog. Therefore we cannot offer a rigorous index theory of ALPUs on finite chains with periodic boundary conditions.  The question is nonetheless important, and perhaps crucial for a generalization to higher dimensions.  We leave the question to future work.

\section{Many-body physics applications}\label{sec:applications}
In this section, we discuss two specific ways in which our results answer natural questions about quantum many-body systems.  In \cref{sec:translation no-go}, we show there cannot be a local ``momentum density'' on the one-dimensional lattice and discuss examples.  In \cref{sec:floquet} we discuss an application of the ALPU index theory to the classification of topological phases in many-body localized Floquet systems, though this application requires further analysis and poses an interesting question for future work.

\subsection{Translations cannot be implemented by local Hamiltonians}\label{sec:translation no-go}
In quantum many-body systems, local conserved quantities dramatically influence dynamics.  For instance, under local Hamiltonian evolution, energy itself is a local conserved quantity, and after the system has locally equilibrated, the dynamics are often governed by energy diffusion.  More generally, when a system admits more local conserved quantities in addition to energy, the near-equilibrium dynamics are often governed by the hydrodynamics of these quantities~\cite{lux2014hydrodynamic, de2019diffusion, bohrdt2017scrambling}. For translation-invariant systems, one expects momentum is also a local conserved quantity.  For instance, in scalar quantum field theory, the $i$'th component of the total momentum operator may be expressed as $P^i = \int dx\, \pi(x) \partial_i \phi(x)$ which is manifestly local, with local momentum density $\pi(x) \partial_i \phi(x)$.

The long-wavelength, low-energy regime of a lattice system like a spin chain is often described by a field theory, and a local momentum density is well-defined under this approximation.  However, we might also ask for a local momentum operator $P = \sum_x p_x$ on the spin chain that generates translations, yielding $U=e^{iP}$ as the one-site translation operator.  If $P$ were constructed with local terms $p_x$, and if $P$ commuted with some translation-invariant Hamiltonian, this exactly conserved momentum density might play an important role in dynamics.

The existence of such a local $P$ is precisely the question of whether the shift QCA can be generated by a local ``Hamiltonian,'' referring now to $P$ as a Hamiltonian.  We show such a local Hamiltonian cannot exist.  In particular, on the infinite one-dimensional chain, it is impossible to implement the translation operator by time evolution using any time-dependent Hamiltonian satisfying Lieb-Robinson bounds, if the Lieb-Robinson bounds lead to an ALPU with $\littleO(1)$-tails.
This follows immediately from \cref{thm:properties index alpu}\ref{it:index alpu same eqv}\ref{it:alpus are strongly connected}.  For instance, we have:

\begin{cor}[No-go for local momentum densities]\label{cor:momentum-no-go}
If a local Hamiltonian $P = \sum_{X \subseteq  \mathbb{Z}} P_X$ on an infinite one-dimensional spin chain has decaying interactions such that for all $n \in \mathbb{Z}$,
\begin{align*}
   \sum_{\substack{X\subseteq  \mathbb{Z} \\ \text{s.t.\ } n \in X}} \norm{P_X} = \bigO\left(\diam(X)^{-(2+\epsilon)}\right)
\end{align*}
for some $\epsilon>0$, then $e^{iP}$ cannot be the unitary lattice translation operator that translates by one~site.
\end{cor}
\noindent Recall our notation that e.g.\ $P_X$ is a term local to region $X \subset  \mathbb{Z}$ on the lattice $\mathbb{Z}$.  The no-go result is also robust: \cref{thm:index alpu after reorg} constrains how well $e^{iP}$ can approximate the translation operator locally.  (Note that while~\cite{gross2012index} already demonstrated that finite-depth circuits cannot achieve translations, their statements about circuits cannot be easily re-cast as claims about Hamiltonian evolution, at least not without further robustness results such as those developed here.)

Given that the translation operator cannot be generated by a finite-depth circuits, our analogous for claim for sufficiently local Hamiltonians might seem in intuitive.  However the claim is not obvious, as demonstrated by the following example: if we allow evolution generated by Hamiltonians with $\frac{1}{r}$-decaying interaction terms (which then violate Lieb-Robinson bounds), we \textit{can} implement a translation.  The example involves a chain of qubits; we only sketch the construction but the details are easily verified.
A Jordan-Wigner transformation maps the chain of qubits to a chain of fermions (or formally, it maps the quasi-local algebra to the CAR-algebra).
Let  $c^{\dagger}_n$ and $c_n$ be the fermionic creation and annihilation operators at site $n \in \ZZ$.
The Jordan-Wigner transform of the translation automorphism $T$ is again the translation automorphism, $T(c_n) = c_{n-1}$.
Taking a Fourier transform we see that
\begin{align*}
  T(\hat{c}_k) = e^{ik}\hat{c}_k
\end{align*}
Hence time evolution for time $t = 1$ using Hamiltonian
\begin{align*}
  H = \int_{-\pi}^{\pi} \d k \, k \hat{c}_k^{\dagger} \hat{c}_k
\end{align*}
implements $T$.
In real space
\begin{align*}
  H = \sum_{n,m \in \ZZ} h_{n-m} c_{n}^{\dagger}c_m
\end{align*}
where the coefficients $h_r$ (of which the precise form is not important) have magnitude $\frac{1}{r}$, .
Of course, we can also take the inverse Jordan-Wigner transform of this Hamiltonian to obtain a Hamiltonian on the spin chain
\begin{align*}
  \tilde{H} = \sum_{n, m} h_{n-m} \sigma_{n,m}
\end{align*}
where $\sigma_{n,m}$ is a Pauli operator supported on sites $\min\{n,m\}, \ldots, \max\{n,m\}$.
In this way we can construct a Hamiltonian \emph{not} satisfying Lieb-Robinson bounds which does implement $T$.
This shows that our demand that the ALPUs have $\littleO(1)$-tails in our construction of the index is not arbitrary; the classification by index collapses once we allow evolutions such as those generated by $\tilde{H}$ above with $\frac{1}{r}$-decaying interactions.
In fact, by \cref{thm:properties index alpu} we conclude that $e^{-i\tilde{H}t}$ cannot have $\littleO(1)$-tails.

For the case of a single-particle Hamiltonian (i.e.\ a quantum walk), the obstruction to generating the translation operator with a local Hamiltonian hinges on the non-trivial winding of the dispersion relation \cite{gross2012index}.  It has been observed that for quadratic fermion Hamiltonians, every such Hamiltonian that implements the translation operator will need to have a discontinuity in its dispersion relation (in our example at $k = \pm\pi$) and hence at least $\frac{1}{r}$-tails in real space~\cite{zimboras2020does,  wilming2020lieb}.  These single-particle and free fermion results do not permit obvious generalization to the broader many-body case; our results allow us to draw conclusions for all local many-body Hamiltonians satisfying Lieb-Robinson bounds.

\subsection{Floquet phases}\label{sec:floquet}
The GNVW index has been used to define a topological invariant that classifies phases of systems with \emph{dynamical many-body localization} for Floquet systems in two dimensions~\cite{po2016chiral}.
The intuitive idea is that on the two-dimensional lattice, under certain localization assumptions, time evolution of a subsystem with boundary defines an associated evolution on the one-dimensional boundary.
The GNVW index of this boundary automorphism then captures whether the boundary has chiral transport, and relatedly whether the two-dimensional system has vortex-like behavior.

Below we sketch the setup described by~\cite{po2016chiral}, pointing to where our results could make the former rigorous.  However, the first step of reducing the two-dimensional dynamics to a one-dimensional boundary dynamics in a rigorous fashion presents an interesting problem of its own, which we leave to future work.

We consider a (time-dependent) local Hamiltonian $H$ on a two-dimensional lattice, and we let $U$ be the unitary obtained by time evolution for some fixed time $T$.
The system exhibits many-body localization (MBL) if $U$ can be written as a product of commuting unitaries which are all approximately local, i.e. when there exists a complete set of approximately local integrals of motion.  More precisely, one says $U$ is MBL in the sense of~\cite{po2016chiral} when it can be written
\begin{align}\label{eq:many body localized}
  U = \prod_X u_X
\end{align}
where $u_X$ is approximately supported in a set $X$ of some bounded size and $[u_X,u_{X'}] = 0$ for all $X, X'$.
What ``approximately supported'' means here depends on one's definition of many-body localization.  A reasonable definition may be
\begin{align*}
  \norm{u_X - E_{B(X,r)(u_x)}} \leq Ce^{-\gamma}
\end{align*}
for some positive $C$ and $\gamma$, where $E_{B(X,r)}$ denotes the projection onto the algebras supported on $B(X,r)$ as defined in~\eqref{eq:twirling}.
To define the invariant we let $D$ denote the upper half plane (or in fact any simply connected infinite subset of the lattice) and we let $U_D$ denote time evolution for time $T$ using only the terms in the Hamiltonian strictly supported inside $D$.
Also, we use~\eqref{eq:many body localized} to define
\begin{align*}
  U_D' = \prod_{X \subseteq D} U_X.
\end{align*}
Then we let $V = U_D^{-1} U_D'$, so that the map $a \mapsto VaV^{-1}$ approximately preserves the algebra supported on a thick boundary strip $\partial D$ of $D$. More precisely, the Lieb-Robinson bounds and~\eqref{eq:many body localized} together show that for an operator $a$ on a single site in $\partial D$, $VaV^\dagger$ is approximately supported within $\partial D$.

In~\cite{po2016chiral}, one implicitly assumes that $V$ (or some deformation thereof) actually defines an ALPU on $\partial D$.  Given this assumption that MBL dynamics define some APLU on $\partial D$, one could then apply the index theory of ALPUs to obtain a rigorous classification of MBL Floquet evolutions in 2D.  Rigorously justifying that assumption presents an interesting future direction.

\section{Discussion}
We have defined and studied the index for approximately locality-preserving unitaries (ALPUs) on spin chains.
Various open questions remain:
\begin{enumerate}
\item Our results are restricted to the infinitely extended chain, or an open finite chain as in \cref{sec:finite}. One could also investigate what happens with a finite \emph{periodic} chain with an $\eps$-nearest neighbor automorphism for small $\eps$. It appears that our proof technique relies on the fact that the chain is infinite (or open), so probably a different strategy is needed for finite periodic chains.
\item An obvious question of interest is the generalization to higher dimensions.
In that case there is no immediate index theory, but one could still hope that for any ALPU $\alpha$ there exists a sequence of QCAs $\alpha_j$ approximating $\alpha$ as in \cref{thm:qca approx}.
Our constructions of approximating QCAs for an ALPU rely rather heavily on the structure theory (i.e., the GNVW index theory) of one-dimensional QCAs.
Hence, it is not immediately clear how to generalize to higher dimensions.
In fact, we have not even given a definition of what an ALPU is in higher dimensions, where some choices exist.
For two-dimensional QCAs there is also a complete classification (in which any QCA is a composition of a circuit and a generalized shift).
Potentially, this structure theory, as developed in~\cite{freedman2020classification, haah2019clifford} can be used in a similar fashion to construct the $\alpha_j$.
This could involve proving stability results for the notion of a ``visibly simple algebra'' as introduced in~\cite{freedman2020classification}.
However, in higher dimensions there is strong evidence for the existence of ``nontrivial'' QCAs (meaning that they cannot be written as a composition of a circuit and a shift)~\cite{haah2018nontrivial}, so this would require a different approach.
Perhaps more generic, topological arguments (e.g.\ using fixed point theorems) are possible.
A direct physical application would be a rigorous understanding of the index discussed in \cref{sec:floquet}.
\item Another direction to generalize in is to channels which preserve locality but which are not unitary (i.e.\ an automorphism), see~\cite{piroli2020quantum} for definitions and a recent discussion.
In other words, what happens if the dynamics is slightly noisy? Is the index robust under small amounts of noise?
Perhaps the type of algebraic stability results we used can also be applied to prove that any locality preserving channel which is almost unitary can be approximated by a QCA.
\item There is also a notion of \emph{fermionic} QCAs, with a corresponding GNVW index.
It should be possible to use similar arguments to extend the index to fermionic ALPUs.
\item Finally, the algebra stability results of \cref{sec:near inclusion appendix} or the related results \cref{lem:homomorphism local error}, \cref{lem:simultaneous near inclusions}, and \cref{lem:lr from single site} could find application in different aspects of quantum information theory.
An example of recent work using similar techniques for very different purposes is~\cite{chao2017overlapping}.
Other potential applications could include approximate error correction.
\end{enumerate}

\section*{Acknowledgements}
We thank Erik Christensen for helpful discussion of his work.
We are grateful to Yoshiko Ogata for pointing out an error in the original statement of \cref{thm:near inclusion} and helping us improve the presentation of \cref{lem:single rotation 0,lem:single rotation 1}.
We acknowledge helpful conversations with Jonas Helsen and Xiao-Liang Qi.  Finally we thank Alexei Kitaev for discussion of related work on alternative notions of decaying tails.
MW acknowledges support by the NWO through VENI grant no.~680-47-459 and grant OCENW.KLEIN.267, by the Deutsche Forschungsgemeinschaft (DFG, German Research Foundation) under Germany's Excellence Strategy - EXC\ 2092\ CASA - 390781972, and by the BMBF through project Quantum Methods and Benchmarks for Resource Allocation (QuBRA).
DR is supported by the Simons Foundation, as well as in part by the DOE Office of Science, Office of High Energy Physics, grant {DE}-{SC0019380}.

\appendix
\section{Commutator lemmas}
In this appendix we bound commutators $[x,f(y)]$ in terms of commutators $[x,y]$, assuming that~$y$ is near the identity.

\begin{lem}[Commutators with powers]\label{lemma:commutator with power}
Let $\A$ be a $C^*$-algebra and let $y\in\A$ be a normal element with $\norm{I-y} \leq \eps < 1$.
Then, for any $s\in[-1,1]$ and $x,y \in \A$, we have
\begin{align*}
  \norm{[x,y^s]} & \leq \frac{\abs s}{(1-\eps)^{1-s}} \norm{[x,y]}.
\end{align*}
\end{lem}
\noindent
For fractional powers, $y^s$ is defined using the functional calculus, with branch cut on the negative imaginary axis (away from the spectrum because $\norm{y-I}<1$).
\begin{proof}
We assume that $s\not\in\{0,1\}$ since otherwise the claim holds trivially.
Let~$z=I-y$.
The function $t \mapsto (1-t)^s$ is holomorphic on the open unit disk, so we may expand
\begin{align*}
    y^s = (I-z)^s = \sum_{n=0}^\infty c_n z^n.
\end{align*}
The exact form of the coefficients $c_n$ here is unimportant, but note $\sgn(c_n)=-\sgn(s)$ for $n\geq 1$ by our assumption that $s\not\in\{0,1\}$.
\begin{align*}
    \norm{[x,y^s]} & = \norm{[x,(I-z)^s]}
    \leq \sum_{n=1}^\infty \abs{c_n} \, \norm{[x,z^n]}
    \leq -\sgn(s) \sum_{n=1}^\infty c_n \, n \, \norm{z}^{n-1} \norm{[x,z]} \\
    &=  -\sgn(s) \frac{d}{dw}(1-w)^s\Bigr|_{w=\norm{z}} \norm{[x,y]}
    = \frac{\abs s}{(1-\norm{z})^{1-s}} \norm{[x,y]}
    \leq  \frac{\abs s}{(1-\eps)^{1-s}} \norm{[x,y]}
\end{align*}
as desired.
\end{proof}

\begin{lem}[Commutators with polar decompositions]\label{lemma:commutator with polar}
Let $\A$ be a $C^*$-algebra and $y\in\A$ an element with $\norm{y-I} \leq \eps \leq \frac18$.
Let $y = u \abs{y}$ be its polar decomposition, with $\abs y = (y^* y)^{\frac12}$.
Then, for any $x\in \A$,
\begin{align*}
\norm{[x,u]} &
< 3 \norm{[x,y]} + 2 \norm{[x,y^*]}.
\end{align*}
\end{lem}
\noindent
More generally for any $\eps < \sqrt2-1$, an estimate of the above form holds for some choice of constants on the right-hand side depending only on~$\epsilon$.
\begin{proof}
Note $\norm{y-I} \leq \epsilon < 1$ implies that $y$ is invertible, hence the unitary $u$ in the polar decomposition is uniquely given by $u = y \abs y^{-1} = y (y^* y)^{-\frac12}$.
Moreover, we have $\norm{y} \leq 1+\epsilon$ and $\norm{y^* y - I} \leq (2 + \eps) \eps$, which also implies that $\norm{(y^* y)^{-\frac12}} \leq (1 - (2+\eps)\eps)^{-\frac12}$, since $\eps < \sqrt 2 - 1$.
We obtain
\begin{align*}
\norm{[x,u]}
& = \norm{[x, y (y^* y)^{-\frac{1}{2}} ]}
\leq \norm{y} \norm{[x, (y^* y)^{-\frac{1}{2}} } + \norm{ (y^* y)^{-\frac{1}{2}} } \norm{[x,y]}  \\
& \leq (1 + \eps) \norm{[x, (y^* y)^{-\frac{1}{2}} } + \frac1{\left(1 - 2\eps - \eps^2\right)^{\frac12}} \norm{[x,y]}  \\
& \leq \frac{1+\eps}{2(1-2\eps-\eps^2)^{\frac32}} \norm{[x,y^*y]} + \frac1{\left(1 - 2\eps - \eps^2\right)^{\frac12}} \norm{[x,y]} \\
& \leq \frac{(1+\eps)^2}{2(1-2\eps-\eps^2)^{\frac32}} \left( \norm{[x,y^*]} + \norm{[x,y]} \right)
+ \frac1{\left(1 - 2\eps - \eps^2\right)^{\frac12}} \norm{[x,y]} \\
& = \frac{(1+\eps)^2 + 2 (1-2\eps-\eps^2)}{2(1-2\eps-\eps^2)^{\frac32}} \norm{[x,y]}
+ \frac{(1+\eps)^2}{2(1-2\eps-\eps^2)^{\frac32}} \norm{[x,y^*]}.
\end{align*}
Here we use the above comments to bound the relevant norms, as well as \cref{lemma:commutator with power} for~$s=-\frac12$.
Using~$\eps\leq\frac18$ this implies the desired bounds.
\end{proof}

\section{Near inclusions of algebras}\label{sec:near inclusion appendix}
In this appendix, we prove \cref{thm:near inclusion} about near inclusions of von Neumann algebras.  The result is an extension of Theorem 4.1 of Christensen~\cite{christensen1980near}, but we give a self-contained proof.
We follow closely the exposition in~\cite{christensen1980near,christensen1977perturbation}.
Note that in~\cite{christensen1980near} it is assumed that injective von Neumann algebras have a property called $D_1$. However, whether this is true is unknown, see comments in~\cite{perera2014cuntz}.
We slightly adapt the arguments of \cite{christensen1980near} to avoid this issue.

We begin with \cref{prp:inner automorphism}, which generalizes Proposition~4.2 of Christensen~\cite{christensen1977perturbation}.  There Christensen considers two subalgebras~$\A,\B \subseteq B(\mc{H})$ that are isomorphic via an isomorphism $\Phi \colon \A \to \B$.  Note that $\Phi$ is defined only on $\A$, not $B(\mc{H})$.
Roughly speaking, the theorem says that if the isomorphism nearly fixes~$\A$, it is inner and implemented by a unitary near the identity.
Our \cref{prp:inner automorphism} below extends this result to the case of multiple commuting subalgebras $\A_i$.  Our generalization will be useful for \cref{lem:homomorphism local error}.  We also extend Christensen's result with the following observation: for elements of $B(\mc{H})$ that nearly commute with $\A$ and $\B$, the distance these elements are moved by the inner automorphism is controlled by the size of their commutator with $\A$ and $\B$.

\begin{prop}[Making homomorphisms inner]\label{prp:inner automorphism}
Consider $C^*$-algebras $\A_i, \B_i \subseteq B(\mc{H})$ for $i=1,\dots,n$, such that each $\A_i''$ is hyperfinite and $[\A_i, \A_j] = [\B_i , \B_j] = 0$ for $i \neq j$.
Consider unital $*$-homomorphisms $\Phi_i \colon \A_i \to \B_i$, with $\norm{\Phi_i(a_i) - a_i} \leq  \gamma_i \norm{a_i}$ for all $a_i \in \A_i$ and $i=1,\dots,n$.
Denote $\A=(\cup_{i=1}^n \A_i)''$, $\B=(\cup_{i=1}^n \B_i)''$, and $\eps = \sum_{i=1}^n \gamma_i$.
If $\eps < 1$, then there exists a unitary $u \in (\A \cup \B)''$ such that $\Phi_i(a_i) = u^* a_i u $ for all $i$ and $a_i \in \A_i$, with
\begin{align*}
\norm{I-u} & \leq   \sqrt{2} \eps (1 + (1 - \eps^2)^{\frac12})^{-\frac12} \leq \sqrt{2}{\epsilon},
\end{align*}
where we note that the expression in the middle is in fact $\epsilon + \mc{O}(\epsilon^2)$.

Moreover, for $\epsilon \leq  \frac{1}{8}$, $u$ can be chosen such that for any $z \in B(\mc{H})$, if~$\norm{[z,c]} \leq  \delta \norm{z} \norm{c}$ for all~$c \in \A \cup \B$, then $\norm{u^*zu-z} \leq  10 \, \delta \norm{z}$.
\end{prop}

\noindent
Note $c \in \A \cup \B$ refers to the union of sets, i.e.\ $c \in \A$ or $c \in \B$.
The proof extends the proof of Proposition~4.2 in~\cite{christensen1977perturbation}.

\begin{proof}[Proof of \cref{prp:inner automorphism}]
We will define an element $y \in (\A \cup \B)''$ whose polar decomposition yields the desired unitary $u$.
We construct the element $y$ to satisfy the properties  $\norm{I-y} \leq \sum_{i=1}^n \gamma_i$ and $  y \Phi_i(u_i)   = u_i y $ for all $u_i \in U(\A_i)$ and $i=1,\dots,n$.%
\footnote{In finite dimension, one can define the element $y$ using  $y = \int_{U(\A_1)} \d u_1 \cdots \int_{U(\A_n)}  \d u_n \; \, u_n^*\cdots u_1^*  \Phi_1(u_1) \cdots \Phi_n(u_n)$, using the Haar measure of the unitary groups $U(\A_i)$. This~$y$ is easily seen to satisfy the abovementioned properties.}

By \cref{prp:extend phi} further below, each homomorphism $\Phi_i \colon \A_i \to \B_i$ can be extended to a $*$-isomorphism $\Phi'_i \colon \A_i'' \to \Phi_i(\A_i)'' \subseteq \B_i''$.
Moreover, we obtain $\norm{\Phi'_i(a_i) - a_i} \leq  \gamma_i \norm{a_i}$.
Without loss of generality, we may assume $\A_i$ is a hyperfinite von Neumann algebra and $\Phi_i$ is a weak-$*$ continuous unital homomorphism (this can always be achieved by replacing $\A_i$ by $\A_i''$, $\B_i$ by $\B_i''$, $\Phi_i$ by $\Phi'_i$; the latter is weak-$*$ continuous because it is a $*$-isomorphism of von Neumann algebras).

Consider $B(\mathcal{H} \oplus \mathcal{H})$ with pairwise commuting subalgebras
\begin{align*}
\C_i = \left\{ \begin{pmatrix}
a_i & 0\\
0 & \Phi_{i}(a_i)
\end{pmatrix} : a_i \in \A_i \right\} \subseteq B(\mathcal{H} \oplus \mathcal{H}).
\end{align*}
Since $\Phi_i$ is weak-$*$ continuous, the map $a_i \mapsto a_i \oplus \Phi_i(a_i)$ is a weak-$*$-continuous unital $*$-homomorphism.
Therefore, $\C_i$, which is its image, is a von Neumann algebra is isomorphic to $\A_i$ and hence hyperfinite.

Therefore, by \cref{thm:hyperfinite} $\C_i$ has property~P and for
\begin{align*}
x_0  =  \begin{pmatrix}
0 & I\\
0 & 0
\end{pmatrix} \in B(\mathcal{H} \oplus \mathcal{H})
\end{align*}
there exists an element $x_1 \in \C_i'$ that is also in the weak operator closure of the convex hull of $\{c_1^* x_0 c_1 \; : \; c_1 \in U(\C_1)\}$.
Note that unitaries $c_1 \in U(\C_1)$ are of the form
\begin{align*}
    \begin{pmatrix}
u_1 & 0\\
0 & \Phi_{1}(u_1)
\end{pmatrix}
\end{align*}
for $u_1 \in U(\A_1)$, so elements $c_1^* x_0 c_1$ are of the form
\begin{align*}
    \begin{pmatrix}
u_1^* & 0\\
0 & \Phi_1(u_1^*)
\end{pmatrix}
\begin{pmatrix}
0 & I\\
0 & 0
\end{pmatrix}
\begin{pmatrix}
u_1 & 0\\
0 & \Phi_1(u_1)
\end{pmatrix}
=
\begin{pmatrix}
0 &  u_1^* \Phi_{1}(u_1)\\
0 & 0
\end{pmatrix}.
\end{align*}
Hence $x_1$ is of the form
\begin{align}\label{eq:x matrix}
x_1  =  \begin{pmatrix}
0 & y_1\\
0 & 0
\end{pmatrix}
\end{align}
for some $y_1 \in (\A_1 \cup \B_1)''$.
By direct calculation, $x_1 \in \C_1'$ implies $y_1 \Phi(u_1) = u_1 y_1$ for any unitary~$u_1 \in \A_1$, and hence
\begin{align*}
    y_1 \Phi_1(a_1)   = a_1 y_1
\end{align*}
for any $a_1 \in \A_1$.

If $n=1$, we take $y_1=y$.
Otherwise, we repeat the above construction but with $x_1$ taking the place of $x_0$, and applying property~P of $\C_2$.
We obtain $x_2 \in \C_2'$ and associated $y_2$, with $y_2 \Phi_2(a_2) = a_2 y_2$ for all $a_2 \in \A_2$.
Also note $x_2 \in \C_1'$, so $y_2 \Phi_{1}(a_1) = a_1 y_2$ for all $a_1 \in U(\A_1)$.
We continue in this way, until we obtain $y:=y_n$, with the property
\begin{align}\label{eq:y-intertwine}
y \Phi_i(a_i) = a_i y
\end{align}
for all $a_i \in \A_i$ and $i=1,\dots,n$.

By construction, $y_1$ is in the weak operator closure of the convex hull of $\{ u_1^* \Phi_1(u_1) : u_1 \in U(\A_1)\}$, and likewise $y_2$ is in the weak operator closure of the convex hull of $\{ u_2^* y_1 \Phi_2(u_2) : u_2 \in U(\A_2)\}$, and so on.  Then $y$ is in the weak operator closure of the convex hull of
\begin{align} \label{eq:n-twirling}
S := \{ u_n^* \ldots u_1^* \Phi_1(u_1)\ldots \Phi_n(u_n) : u_1 \in U(\A_1), \ldots, u_n \in U(\A_n)\}.
\end{align}
Elements of this form are near the identity,
\begin{align*}
\norm{I- u_n^* \ldots u_1^* \Phi_1(u_1)\ldots \Phi_n(u_n)} & \leq \norm{I- u_n^* \ldots u_2^* \Phi_2(u_2)\ldots \Phi_n(u_n)} + \norm{\Phi_1(u_1)-u_1} \\
& \leq \sum_{i=1}^n  \norm{\Phi_i(u_i)-u_i},
\end{align*}
and thus, by convexity and lower semicontinuity of the norm in the weak operator topology,
\begin{align*}
  \norm{I-y} \leq \sum_{i=1}^n \gamma_i = \eps.
\end{align*}
Define $u=y\abs y^{-1}$ as the unitary in the polar decomposition of $y$.  By the above estimate, it generally follows (Lemma~2.7 of~\cite{christensen1975perturbations}) that
\begin{align*}
  \norm{u-I} \leq \sqrt{2} \eps (1 + (1 - \eps^2)^{\frac12})^{-\frac12} \leq
  \sqrt2 \eps.
\end{align*}

We now show
\begin{align*}
    u^* a_i u = \Phi_i(a_{i})
\end{align*}
for all $a_{i} \in \A_{i}$ and $i=1,\dots,n$.
To see this, first note that~\eqref{eq:y-intertwine} implies $y^*y=\Phi_i(u_i)^* y^* y \Phi(u_i)$ for any~$u_i \in U(\A_i)$, so that $[\Phi_i(u_i),y^*y]=0$.
Then, since any $a_i \in \A_i$ can be written as a linear combination of unitary elements, $[\Phi_i(a_i),y^*y]=0$, hence $[\Phi_i(a_i), \abs y^{-1}]=0$ and
\begin{align*}
    u^* a_i u & =  \abs{y}^{-1} y^* a_i  y \abs{y}^{-1}
    = \abs{y}^{-1} y^* y \Phi_i(a_i) \abs{y}^{-1}
    =  \abs{y}^{-1} y^* y  \abs{y}^{-1} \Phi_i(a_i)
    =  \Phi(a_i)
\end{align*}
where we first used~\eqref{eq:y-intertwine} and then that $[\Phi_i(a_i), \abs y^{-1}]=0$.

Finally, we show the last claim of the theorem.
Consider any $z \in B(\mc{H})$ with the property that $\norm{[z,c]} \leq \delta \norm{z} \norm{c}$ for all $c \in \A \cup \B$.
Then, $\norm{[z,s]} \leq 2 \delta \norm{z}$ for any $s\in S$, since any element of~$S$ is a product of a unitary in $U(\A)$ and a unitary in $U(\B)$, and likewise $\norm{[z,s^*]} \leq 2 \delta \norm z$.
We find that
\begin{align*}
  \norm{[z,y]} \leq 2 \delta \norm{z},
\end{align*}
using that $y$ is in the weak operator closure of the convex hull of~$S$ as defined in~\eqref{eq:n-twirling}.
To see this, let $y_i$ be a net of elements in the convex hull of~$S$ that converges to~$y$ in the weak operator topology.
Since the elements in $S$ have norm at most one, by convexity it holds that $\norm{y_i} \leq  1$ as well and hence $\norm{[z,y_i]} \leq 2 \delta \norm{z}$.
The norm is lower semicontinuous in the weak operator topology which implies that $\norm{[z,y]} \leq \liminf_i \norm{[z,y_i]} \leq 2 \delta \norm{z}$.
The above reasoning holds for $y^*$ as well.
Then we can apply \cref{lemma:commutator with polar}, using that $\norm{I-y} \leq \eps \leq \frac{1}{8}$.
We find
\begin{align*}
  \norm{ u^* z u - z } &= \norm{[z,u]}
  \leq  3 \norm{[z,y]} + 2 \norm{[z,y^*]}
  \leq 10\, \delta \norm{z}
\end{align*}
as desired.
\end{proof}

The above proof is completed by the technical proposition below.  The proof follows from the proof of Theorem 5.4 in~\cite{christensen1977perturbation}.
\begin{prop}\label{prp:extend phi}
Given a $C^*$-algebra $\A \subseteq B(\mc{H})$ with unital $*$-homomorphism $\Phi \colon \A \to B(\mc{H})$ and $\norm{\Phi(a)-a} \;\leq\; \eps \norm{a}$ for all $a \in \A$ and some $\eps<1$, then $\Phi$ can be extended to a $*$-isomorphism $\Phi' \colon \A'' \to \Phi(\A)''$ with $\norm{\Phi(a)-a} \leq \eps\norm{a}$ for all $a \in \A''$.
\end{prop}
\begin{proof}
To extend $\Phi$, consider $B(\mathcal{H} \oplus \mathcal{H})$ with subalgebra
\begin{align*}
\C = \left\{ \begin{pmatrix}
a & 0\\
0 & \Phi(a)
\end{pmatrix} : a \in \A \right\} \subseteq B(\mathcal{H} \oplus \mathcal{H}) .
\end{align*}
We first show that for any $a \in \A''$, there exists unique $b \in \Phi(A)''$ such that
\begin{align*}
c =\begin{pmatrix}
a & 0\\
0 & b
\end{pmatrix} \in \C''.
\end{align*}
For $a \in \A''$, by Kaplansky's density theorem, there exists a net $\{a_i\}$ in $\A$ converging in the strong and hence in the weak operator topology to $a$, with $\norm{a_i}\leq \norm{a}$.
Then $\norm{\Phi(a_i)-a_i} \leq \eps\norm{a}$, and $\norm{\Phi(a_i)} \leq (1+\epsilon)\norm{a}$, so we can define a net
\begin{align*}
c_i = \begin{pmatrix}
a_i & 0\\
0 & \Phi(a_i)
\end{pmatrix}
\end{align*}
within a ball of finite radius in $B(\mc{H})$.
Since such balls are compact in the weak operator topology, 
$c_i$ must have a convergent subnet, which then converges to some
\begin{align*}
c =\begin{pmatrix}
a & 0\\
0 & b
\end{pmatrix} \in \C'',
\end{align*}
as claimed.   To see the uniqueness of $b$ given $a$, suppose otherwise that there exist corresponding $b_1,b_2 \in \Phi(A)''$ with $(a,b_1), (a,b_2) \in \C''$, so that $z=b_1-b_2 \in \Phi(\A)''$ with
\begin{align*}
c = \begin{pmatrix}
0 & 0\\
0 & z
\end{pmatrix} \in \C''.
\end{align*}
By Kaplansky's density theorem, there exists a net $\{c_i\}$ in $\C$ converging strongly to $c$ with $\norm{c_i}\leq \norm{c}=\norm{z}$.
Write $c_i = (a_i, \Phi(a_i))$ for $a_i \in \A$.
Then $\norm{a_i}\leq \norm{z}$, $\{a_i\}$ converges strongly to zero, and~$\{\Phi(a_i)\}$ converges strongly to $z$, so $\Phi(a_i)-a_i$ converges strongly to $z$ and hence also weakly.
By the lower semicontinuity of the norm for the weak operator topology, $\norm{z} \leq \liminf_i \norm{\Phi(a_i)-a_i} \leq \eps \norm{z}$, so that $\norm{z}=0$ and $b_1 = b_2$, demonstrating uniqueness.

A similar argument shows that for any $b \in \Phi(\A)''$, there exists unique $a \in \A''$ such that
\begin{align*}
c =\begin{pmatrix}
a & 0\\
0 & b
\end{pmatrix} \in \C''.
\end{align*}
The above maps $a \mapsto b$ and $b \mapsto a$ define a bijection $\Phi' \colon \A'' \to \Phi(\A)''$.
The linearity, multiplicativity, and $*$-property of $\Phi'$ follow from the above uniqueness.
Thus $\Phi'$ is a $*$-isomorphism.
Finally, we show $\norm{\Phi'(a)-a} \leq \norm{a}\epsilon$ for all $a \in \A''$.
By Kaplansky's density theorem, there exists a net $\{a_i\}$ strongly converging to $a$ for $a_i \in \A$  with $\norm{a_i}\leq \norm{a}$.
By the above constructions, there exists a subnet such that $\Phi(a_i)$ converges in the weak operator topology to $\Phi'(a)$.
Then, again by the lower semicontinuity of the norm, $\norm{\Phi'(a) - a} \leq \liminf_i \norm{\Phi(a_i)-a_i} \leq \eps\norm{a}$, as desired.
\end{proof}

Now we turn to \cref{thm:near inclusion}.  In Theorem~4.1 of~\cite{christensen1980near}, Christensen proves that if a subalgebra~$\A$ is approximately contained in another subalgebra~$\B$ then there exists a unitary near the identity that rotates $\A$ into $\B$.
Our \cref{thm:near inclusion} extends his result with the following observations.
First, elements of $B(\mc{H})$ already close to both~$\A$ and~$\B$ are not moved much by the automorphism.
Second, elements that nearly commute with both~$\A$ and~$\B$ are are not moved much either.
Thus the automorphism ``does no more than it needs.''

For convenience, we recall the notion of near inclusions in \cref{dfn:near inclusion}.  We write $a \overset{\eps}{\in} \B$ when there exists $b \in \B$ such that
\begin{align*}
    \norm{a - b} \leq \eps\norm{a},
\end{align*}
and we write $\A \overset{\eps}{\subseteq} \B$ when $a \overset{\eps}{\in} \B$ for all $a \in \A$.
Also recall the notion of hyperfinite von Neumann algebras, reviewed in \cref{sec:vN prelims}.
Then we are equipped to state \cref{thm:near inclusion}, repeated below.
\nearinclusion*

\noindent
The first item re-states Theorem~4.1 of Christensen~\cite{christensen1980near}, or specifically part~(b) of his Corollary~4.2 (noting that hyperfinite algebras are injective).
The remaining items constitute our extension.

Now we proceed with the proof of \cref{thm:near inclusion}, closely following and elaborating on some technical details and then extending the proof of Theorem~4.1 in~\cite{christensen1980near}.
\begin{proof}[Proof of \cref{thm:near inclusion}]
By \cref{thm:hyperfinite}, since $\B$ is hyperfinite, it is injective and hence there exists a conditional expectation
\begin{align*}
  \EE_\B \colon B(\mc H) \to \B \subseteq B(\mc H).
\end{align*}
This map is completely positive and unital,
and thus it has a Stinespring dilation~\cite{stinespring1955positive}.
That is, there exists a Hilbert space~$\mc{K}$, a unital $*$-homomorphism~$\pi \colon B(\mc{H}) \to B(\mc{K})$, and an isometry $v \colon \mc{H} \to \mc{K}$ such that
\begin{align}\label{eq:dilation}
  \EE_\B(x) =v^* \pi(x) v \quad \forall x \in B(\mc H).
\end{align}
Let $p = v v^* \in B(\mc K)$ be the projection onto the image of~$v$.
Then $p \in \pi(\B)'$, since~$\EE_\B$ restricted to~$\B$ is an isomorphism.%
\footnote{In more detail, to see $p \in \pi(\B)'$, first note $\pi(\B) \to  B(\mc{K}), \, \pi(b) \mapsto p \pi(b) p$ is a $*$-homomorphism. Then note the following general fact: for any algebra $\A \subset B(\mc{H})$ and projection $p \in B(\mc{H})$, if the map $f(a)=pap$ is a $*$-homomorphism, then $p \in \A'$. To see this, note for any $a \in \A$,  $f(a^* a)= f(a^*)f(a) = pa^*p a p$ and $f(a^*a) = pa^*ap=pa^*(p+p^ \perp) a p=pa ^*pap + pa^*p^\perp a p$, so  the difference yields $0=pa^*p^\perp a p=(p^\perp a p)^*( p^\perp a p)$, so $p^\perp a p=0$.  The same is true for $a^*$, so $p a p^\perp=0$ also.  Then $[p,a]=(p+p^\perp)[p,a](p+p^\perp) = pap^\perp - p^\perp a p =0$.}%
\footnote{\label{foot:fin dim}It may be helpful to understand the Stinespring dilation explicitly in finite dimensions where $\mc{H} = \mc H_A \ot \mc H_B$ and~$\B = I_A \ot B(\mc H_B)$, with commutant $\B' = \A = B(\mc H_A) \ot I_B$.
Then the conditional expectation is the normalized partial trace $\EE_\B(x) = \frac1{d_A} \tr_{\A}(x)$.
For a minimal Stinespring dilation we can take the Hilbert space~$\mc K = \mc H_A^1 \ot H_A^2 \ot \mc H_A^3 \ot \mc H_B$, where the $\mc H_A^i$ are three copies of the Hilbert space $\mc H_A$.
We define $\pi \colon B(\mc H) \to B(\mc K)$ by identifying operators on $\mc H$ with operators on $\mc H_A^1 \ot \mc H_B$.
Finally, we take the isometry $v$ as adding a maximally entangled state on $\mc H_A^1 \ot \mc H_A^2$.
Note that the projection~$p$ onto the image of~$v$ commutes with~$\pi(\B) = I_{A^1A^2A^3} \ot B(\mc H_B)$.}

Next we show that~$p$ nearly commutes with~$\pi(\A)$ as well.
For any $a \in \A$, choose $b \in \B$ with $\norm{a-b} \leq \eps \norm{a}$, using $\A \overset{\eps}{\subseteq} \B$.
Then,
\begin{align} \label{eq:pi_a_p_commutator}
  \norm{ [\pi(a),p] }
= \frac12 \norm{[\pi(a-b),2p-I]}
\leq \norm{\pi(a-b)} \, \norm{2p-I}
\leq \eps \norm{a},
\end{align}
noting that for any projection, $\norm{2p-I} = 1$.

Note that although $\A$ itself is hyperfinite, $\pi(\A)''$ is not immediately guaranteed hyperfinite, because $\pi$ is not guaranteed weak-$*$ continuous.
On the other hand, since $\A$ is hyperfinite, there exists an AF $C^*$-algebra $\A_0 \subseteq \A$ such that $\A _0'' = \A$, as in the theorem statement.
Then $\pi(\A_0)$ is also AF, and $\pi(\A_0)''$ is hyperfinite.

Because $\pi(\A_0)''$ hyperfinite, it satisfies property~P, so there exists~$\tilde p$ in the weak operator closure of the convex hull of $\{upu^* : u \in U(\pi(\A_0)'')\}$ such that $\tilde p \in \pi(\A_0)'$.
By convexity of the norm and lower semicontinuity of the norm with respect to the weak operator topology,
\begin{align*}
  \norm{\tilde p-p} \leq \sup_{ u \in U(\pi(\A_0)'')} \norm{\tilde p - upu^*} = \sup_{ u \in U(\pi(\A_0)'')} \norm{[u,p]} \leq \epsilon.
\end{align*}
The final inequality follows from \cref{eq:pi_a_p_commutator} in the following way. 
First note that~$\norm{ [u, p] } \leq \eps$ for any $u \in \pi(U(\A))$.
By the lower semicontinuity of the norm in the weak operator topology, this estimate extends directly to the weak operator closure of $\pi(U(\A))$. 
Accordingly, it suffices to show that any $u \in U(\pi(\A_0)'')$ is contained in the weak operator closure of~$\pi(U(\A))$.
This can be seen as follows.
By Theorem~5.2.5 in \cite{kadison1997fundamentals} there exists self-adjoint~$y \in \pi(\A_0)''$ such that $u = e^{iy}$.
By a version of the Kaplansky density theorem, there exists a net $\{y_n\}$ of self-adjoint elements $y_n \in \pi(\A_0)$ converging strongly to $y$, with $\norm{y_n} \leq \norm{y}$.
Then~$\{e^{i y_n}\}$ is a net of elements in $\pi(U(\A_0))$, since we can always write $y_n = \pi(x_n)$ with self-adjoint $x_n \in \A_0$, hence $e^{i y_n} = e^{i \pi(x_n)} = \pi(e^{i x_n})$ and $e^{i x_n}$ is unitary.
On the other hand, by Proposition~5.3.2 in \cite{kadison1997fundamentals}, $\{e^{i y_n}\}$ converges strongly to $e^{iy} = u$.
We conclude that $u$ is in the strong (and hence in the weak) operator closure of $\pi(U(\A_0))$, hence in particular of~$\pi(U(\A))$.
Note that, by construction, $\norm{\tilde p} \leq \norm{p} = 1$.

Next we would like to project $\tilde p$ onto $(\pi(B(\mc{H})) \cup \{p\})''$, the von Neumann algebra generated by $\pi(B(\mc H))$ and the projection $p$ inside~$B(\mc K)$.\footnote{In the finite-dimensional setting of \cref{foot:fin dim}, where again $\mc{K} = \mc{H}_{A}^1  \ot \mc{H}_{A}^2 \ot \mc{H}_{A}^3 \ot \mc{H}_B$, we have $(\pi(B(\mc{H})) \cup \{p\})' = B(\mc{H}_A^3)$ and hence $(\pi(B(\mc{H})) \cup \{p\})'' = B(\mc{H}_{A}^1 \ot \mc{H}_{A}^2 \ot \mc{H}_B )$.}
By Corollary 1.3.2 in~\cite{arveson1969subalgebras}, $(\pi(B(\mc{H})) \cup \{p\})'$ is isomorphic to $\B'$.
Because the commutant of an injective von Neumann algebra is injective, and $\B$ is injective, $\B'$ is also injective, hence also $(\pi(B(\mc{H})) \cup \{p\})'$ and $(\pi(B(\mc{H})) \cup \{p\})''$.
Thus we can use a conditional expectation to define
\begin{align*}
  x = \EE_{(\pi(B(\mc{H})) \cup \{p\})''}(\tilde p) \in (\pi(B(\mc{H})) \cup \{p\})'', \quad\norm{x-p} \leq \eps,
\end{align*}
where the norm bound follows from $\norm{\tilde p - p}\leq \eps$, because the conditional expectation is a contraction.
Moreover, it holds that $x \in \pi(\A_0)'$.
To see this, compute $[x,z]=0$ for $z \in \pi(\A_0)$, using that~$\tilde p \in \pi(\A_0)'$ and $\pi(\A_0) \subseteq (\pi(B(\mc{H})) \cup \{p\})''$, and the general property of conditional expectations that $\EE_{\mc{Z}}(z_1 y z_2)=z_1 \EE_{\mc{Z}}(y)z_2$ for $z_1,z_2 \in \mc{Z}$.

The next steps follow Lemma~3.3 of~\cite{christensen1977perturbation}.
Note that because $x$ is self-adjoint, $\norm{x-p} \leq \eps$, and~$\norm{x} \leq 1$ (as conditional expectations are contractions), its spectrum is in $[-\eps,\eps] \cup [1-\eps,1]$.
Define the projection $q \in \pi(\A_0)'$ as the spectral projection of~$x$ corresponding to the part of the spectrum in~$[1-\eps,1]$.
Then, $\norm{q-x} \leq \eps$ and $\norm{q-p} \leq 2\eps$.

Using the projection~$p \in \pi(\B)'$ and the nearby projection~$q \in \pi(\A_0)'$, define
\begin{align*}
  y = qp + q^\perp p^\perp,
\end{align*}
where $p^\perp=(I-p)$ denotes the projection onto the orthogonal complement.
Then
\begin{align}\label{eq:y-I}
  \norm{y-I} = \norm{(2q- I)(p-q)} \leq \norm{p-q} \leq 2\eps.
\end{align}
In particular, $y$ is invertible.
Now consider the unitary $w = y \abs y^{-1}$ from the polar decomposition $y = w \abs y$.
Because $y$ is near the identity, $w$ must be as well.
Namely, by Lemma~2.7 of~\cite{christensen1975perturbations}, we find
\begin{align}\label{eq:w-I}
  \norm{w-I} \leq 2 \sqrt{2} \, \eps.
\end{align}
Since $y^* y = pqp + p^\perp q^\perp p^\perp$, we have $[p,y^*y]=0$ and hence $[p,\abs y^{-1}]=0$.
Moreover, $yp = qy$, so
\begin{align}\label{eq:wpw*}
    wpw^* & = y \abs{y}^{-1} p \abs{y}^{-1} y^*
    = yp \abs{y}^{-1}  \abs{y}^{-1} y^*
    = q y \abs{y}^{-1}  \abs{y}^{-1} y^*
    = q.
\end{align}
With the unitary~$w \in (p \cup \pi(B(\mc H)))''$, we can finally define the homomorphism to which we soon apply \cref{prp:inner automorphism}.
Let
\begin{align*}
  \Phi \colon \A_0 \to \B \subseteq B(\mc H), \quad \Phi(a) = v^* w^* \pi(a) w v.
\end{align*}
This is a unital $*$-homomorphism, since it clearly preserves the $*$-operation and we have for all $a_1, a_2 \in \A_0$ that
\begin{align*}
  \Phi(a_1)\Phi(a_2)
&= v^* w^* \pi(a_1) w p w^* \pi(a_2) w v
= v^* w^* \pi(a_1) q \pi(a_2) w v \\
&= v^* w^* \pi(a_1 a_2) q w v
= v^* w^* \pi(a_1 a_2) w p v
= \Phi(a_1 a_2),
\end{align*}
using \eqref{eq:wpw*}, $q \in \pi(\A_0)'$, $p = v v^*$, and that $v$ is an isometry.
To see that its image lies in $\B$, note that $w^* \pi(a) w \in (p \cup \pi(B(\mc H)))''$ and recall the original construction of the Stinespring dilation in~\eqref{eq:dilation}.
Moreover, for any $a\in \A_0$, there exists $b \in \B$ with $\norm{b-a}\leq \eps \norm{a}$, so that
\begin{align}
  \norm{\Phi(a)-a}
& \leq \norm{\Phi(a)-b} + \norm{b-a} \nonumber \\
& = \norm{v^* \bigl( w^* \pi(a) w - \pi(b) \bigr) v} + \norm{b-a} \nonumber \\
& \leq \norm{w^* \pi(a) w - \pi(b)} + \norm{b-a} \nonumber \\
& \leq \norm{w^* \pi(a) w - \pi(a)} + 2 \norm{b-a} \nonumber \\
& = \norm{[\pi(a),w]} + 2 \norm{b-a} \label{eq:Phi(a)-a} \\
& = \norm{[\pi(a),w-I]} + 2 \norm{b-a} \nonumber \\
& \leq 2 \norm{w - I} \norm a + 2 \norm{b-a} \nonumber
\leq 8 \eps \norm a.
\end{align}
using $b = \EE_\B(b) = v^* \pi(b) v$ in the second step and~\eqref{eq:w-I} in the last step.

We can thus apply \cref{prp:inner automorphism} (for $n=1$) to obtain a unitary~$u \in (\A \cup \B)''$ with $\norm{u - I} \leq \sqrt{2} \cdot 8 \eps \leq 12 \eps$ such that $u^* a u = \Phi(a) \in \B$ for all $a \in \A_0$ , and we extend $\Phi \colon \A \to \B$ by $\Phi(a) = u^* a u$ for $a \in \A=\A_0''$.
Moreover, by \cref{prp:inner automorphism}, we are already ensured the desired property of \cref{thm:near inclusion} that if $z \in B(\mc{H})$ satisfies $\norm{[z,c]} \leq  \delta \norm{z} \norm{c}$ for all $ c \in \A \cup \B$, then $\norm{uzu^*-z} \leq 10 \delta \norm{z}$.

Finally, we need to show the additional property that for any $z \in B(\mc{H})$ with $z \overset{\delta}{\in} \A_0$ and $z \overset{\delta}{\in} \B$, we have $\norm{u^*zu-z} \leq 16\delta\norm{z}$. First take $a \in \A_0$ with~$\norm{z-a} \leq \delta \norm{z}$.
Then,
\begin{align*}
   \norm{u^* z u-z}
= \norm{u^*(z-a)u - (z-a) + u^*au - a}
\leq 2 \delta \norm{z} + \norm{\Phi(a) - a}.
\end{align*}
Now take $b \in \B$ with~$\norm{z-b} \leq \delta \norm{z}$, hence also $\norm{a-b} \leq 2 \delta \norm{z}$.
Then we can bound just like in~\eqref{eq:Phi(a)-a} to obtain (note that $a \in \A_0$)
\begin{align*}
  \norm{\Phi(a) - a} &\leq
\norm{[\pi(a),w]} + 2\norm{a - b} \\
&\leq \norm{[\pi(a),w]} + 4 \delta \norm{z} \\
&\leq 3 \norm{[\pi(a),y]} + 2\norm{[\pi(a),y^*]} + 4 \delta \norm{z},
\end{align*}
where we used \cref{lemma:commutator with polar} in the last line, noting that $\norm{y-I} \leq 2\eps \leq \frac18$ by \cref{eq:y-I} and our assumption on~$\eps$.
To bound the commutators' norms, recall that $p\in \pi(\B)'$ and $q \in \pi(\A_{0})'$.
Hence,
\begin{align*}
  \norm{[\pi(a),y]}
&= \norm{[\pi(a),qp + q^\perp p^\perp]}
= \norm{(2q-I)[\pi(a),p]} \\
&\leq \norm{[\pi(a),p]}
= \norm{[\pi(a-b),p]}
= \frac12\norm{[\pi(a-b),2p-I]}
\leq \norm{a-b}
\leq 2 \delta \norm z,
\end{align*}
and likewise for $[\pi(b),y^*]$.
Therefore, $\norm{\Phi(a) - a} \leq 14 \delta\norm{z}$, and hence
\begin{align*}
   \norm{u^* z u - z} \leq
16 \delta\norm{z}
\end{align*}
as desired.
\end{proof}

As another application of \cref{prp:inner automorphism}, \cref{lem:homomorphism local error} controls the distance between homomorphisms using the distance between their local restrictions.  We repeat the statement for convenience.
\homomorphismerror*
\begin{proof}
Since we assume the $\alpha_i$ to be weak-$*$ continuous, $\alpha_1(\A_i)$ and $\alpha_2(\A_i)$ are von Neumann algebras which are isomorphic to $\A_i$ (and in particular are hyperfinite).
Define $*$-isomorphisms $\Phi_i$ between $\alpha_1(\A_i)$ and $\alpha_2(\A_i)$, given by $\alpha_1(a_i) \mapsto \alpha_2(a_i)$ for $a_i \in \A_i$.
Then we apply \cref{prp:inner automorphism} (with $\alpha_1(\A_i)$ as $\A_i$, $\alpha_2(\A_i)$ as $\B_i$, $\gamma_i = \norm{(\alpha_1 - \alpha_2)|_{\A_i}}$) to find a unitary $u \in \B$ such that $\alpha_2(a_i) = u^*\alpha_1(a_i)u$ for all $a_i \in \A_i$ for $i = 1, \ldots n$ with
\begin{align*}
  \norm{I - u} \leq \sqrt{2} \eps \left( 1 + (1 - \eps^2)^{\frac12} \right)^{-\frac12}.
\end{align*}
This implies that $\alpha_2(a) = u^*\alpha_1(a)u$  for all $a \in \A$,
and hence
\begin{align*}
  \norm{\alpha_1 - \alpha_2} \leq 2\sqrt{2} \eps \left( 1 + (1 - \eps^2)^{\frac12} \right)^{-\frac12}.
\end{align*}
\end{proof}

Finally, we mention another result about simultaneous near inclusions.  If several mutually commuting subalgebras $\A_i$ each nearly include into $\B$, then so does the algebra they generate.
We use this lemma to prove \cref{lem:lr from single site} which shows that Lieb-Robinson type bounds for single site operators imply Lieb-Robinson bounds for operators supported on arbitrary sets.

\begin{lem}[Simultaneous near inclusions]\label{lem:simultaneous near inclusions}
Let $\A_i \subseteq B(\mc H)$ for $i=1,...,n$ and $\B \subseteq B(\mc H)$ be von Neumann algebras, where the $\A_i$ are hyperfinite and $[\A_i,\A_j]=0$ for $i\neq j$.
If $\A_i \overset{\eps_i}{\subseteq} \B$ for each~$i$, then for~$\eps := \sum_{i=1}^n \eps_i$ we have
\begin{align}\label{eq:simultaneous near inclusion commutant}
\B' \overset{2\eps}{\subseteq} (\cup_i \A_i)'.
\end{align}
If additionally $\A_i \subseteq \M$ for some von Neumann algebra $\M \subseteq B(\mc H)$ for $i = 1, \ldots, n$, then
\begin{align}\label{eq:simultaneous near inclusion commutant with intersection}
  \B' \cap \M \overset{2\eps}{\subseteq} (\cup_i \A_i)' \cap \M.
\end{align}
Finally, if $\B'$ is hyperfinite then
\begin{align*}
  (\cup_i \A_i)'' \overset{4\eps}{\subseteq} \B.
\end{align*}
\end{lem}
\begin{proof}
First we show $\B'$ nearly includes into $(\cup_i \A_i)'$.
By hyperfiniteness (and therefore property~P) of~$\A_1$, for each $b'_0 \in \B'$ there exists~$b'_1 \in \A_1'$ in the weak operator closure of the convex hull of $\{ u_1^* b'_0 u_1 \; : \; u_{1} \in U(\A_1)\}$.
Then by property~P of $\A_2$, there exists $b'_2 \in \A_{2}'$ in the weak operator closure of the convex hull of $\{ u_2^* b'_1 u_2 \; : \; u \in U(\A_2)\}$.
Note that $b'_2 \in \A_1'$ still, using~$[\A_1, \A_2]=0$.
We continue in this way until we find $b'_n$ in the weak operator closure of the convex hull of%
\footnote{In the finite-dimensional case, we could immediately define $b'_n=  \int_{U(\A_1)}\d u_1  \cdots \int_{U(\A_n)}  \d u_n \; \, u_n^*\cdots u_1^*  b'_0 u_1 \cdots u_n$ using the Haar integral, rather than make use of property~P.}
\begin{align*}
\{ u_n^* \cdots u_1^* b'_0 u_1 \cdots u_n \;:\; u_1 \in U(\A_1), \dots, u_n \in U(\A_n)\}.
\end{align*}
Note $\norm{[u_i,b'_0]} \leq 2\epsilon_i \norm{b'_0}$, by $\A_i \overset{\eps_i}{\subseteq} \B$ and \cref{lem:near inclusion commutator 0}.
Thus, elements in the above set are near $b'_0$, since by a telescoping sum
\begin{align*}
  \norm{b'_0 - u_n^* \cdots u_1^* b'_0 u_1 \cdots u_n}
& \leq \sum_{i=1}^{n} \norm{u_n^* \cdots u_{i+1}^* b'_0 u_{i+1} \cdots u_n - u_n^* \cdots u_i^* b'_0 u_i \cdots u_n} \\
& = \sum_{i=1}^n \norm{b'_0 - u_i^* b'_0 u_i}
= \sum_{i=1}^n \norm{[u_i, b'_0]}
\leq 2\eps \norm{b'_0}
\end{align*}
and hence, using the convexity of the norm and its lower semicontinuity with respect to the weak operator topology,
\begin{align*}
\norm{b'_0 - b'_n} \leq 2\eps \norm{b'_0 }.
\end{align*}
By construction, $b'_n \in \A_i'$ for each $i$, so $b'_n \in (\cup_i \A_i)'$.  The above construction held for any $b'_0 \in \B'$, so \cref{eq:simultaneous near inclusion commutant} follows.
Note that if we assume that each $\A_i \subseteq \M$ and we take $b_0' \in \B' \cap \M$, then also~$b_n' \in \M$, which shows \cref{eq:simultaneous near inclusion commutant with intersection}.
Finally, by \cref{lem:near inclusion commutant} and the assumption that $\B'$ is hyperfinite
we conclude that $(\cup_i \A_i)'' \overset{4\eps}{\subseteq} \B$.
\end{proof}

\bibliographystyle{hunsrtnat}
\bibliography{references}

\begin{thebibliography}{57}
\expandafter\ifx\csname natexlab\endcsname\relax\def\natexlab#1{#1}\fi
\expandafter\ifx\csname url\endcsname\relax
  \def\url#1{{\tt #1}}\fi

\bibitem[Margolus(1986)]{margolus1986quantum}
Norman Margolus.
\newblock Quantum computation.
\newblock 1986.

\bibitem[Schumacher and Werner(2004)]{schumacher2004reversible}
Benjamin Schumacher and Reinhard~F Werner.
\newblock Reversible quantum cellular automata.
\newblock {\em arXiv preprint quant-ph/0405174}, 2004.

\bibitem[Lieb and Robinson(1972)]{lieb1972}
Elliott~H Lieb and Derek~W Robinson.
\newblock The finite group velocity of quantum spin systems.
\newblock {\em Communications in Mathematical Physics}, 28\penalty0
  (3):\penalty0 251--257, 1972.

\bibitem[Po et~al.(2016)Po, Fidkowski, Morimoto, Potter, and
  Vishwanath]{po2016chiral}
Hoi~Chun Po, Lukasz Fidkowski, Takahiro Morimoto, Andrew~C Potter, and Ashvin
  Vishwanath.
\newblock Chiral {Floquet} phases of many-body localized bosons.
\newblock {\em Physical Review X}, 6\penalty0 (4):\penalty0 041070, 2016.

\bibitem[Haah et~al.(2018{\natexlab{a}})Haah, Fidkowski, and
  Hastings]{haah2018nontrivial}
Jeongwan Haah, Lukasz Fidkowski, and Matthew~B Hastings.
\newblock Nontrivial quantum cellular automata in higher dimensions.
\newblock {\em arXiv preprint arXiv:1812.01625}, 2018{\natexlab{a}}.

\bibitem[Gross et~al.(2012)Gross, Nesme, Vogts, and Werner]{gross2012index}
David Gross, Vincent Nesme, Holger Vogts, and Reinhard~F Werner.
\newblock Index theory of one dimensional quantum walks and cellular automata.
\newblock {\em Communications in Mathematical Physics}, 310\penalty0
  (2):\penalty0 419--454, 2012.

\bibitem[Freedman and Hastings(2020)]{freedman2020classification}
Michael Freedman and Matthew~B Hastings.
\newblock Classification of quantum cellular automata.
\newblock {\em Communications in Mathematical Physics}, 376\penalty0
  (2):\penalty0 1171--1222, 2020.

\bibitem[Haah(2019)]{haah2019clifford}
Jeongwan Haah.
\newblock Clifford quantum cellular automata: {T}rivial group in {2D} and
  {W}itt group in {3D}.
\newblock {\em arXiv preprint arXiv:1907.02075}, 2019.

\bibitem[Arrighi et~al.(2020)Arrighi, B{\'e}ny, and
  Farrelly]{arrighi2020quantum}
Pablo Arrighi, C{\'e}dric B{\'e}ny, and Terry Farrelly.
\newblock A quantum cellular automaton for one-dimensional {QED}.
\newblock {\em Quantum Information Processing}, 19\penalty0 (3):\penalty0 88,
  2020.

\bibitem[Bisio et~al.(2018)Bisio, D'Ariano, Perinotti, and
  Tosini]{bisio2018thirring}
Alessandro Bisio, Giacomo~Mauro D'Ariano, Paolo Perinotti, and Alessandro
  Tosini.
\newblock Thirring quantum cellular automaton.
\newblock {\em Physical Review A}, 97\penalty0 (3):\penalty0 032132, 2018.

\bibitem[Alba et~al.(2019)Alba, Dubail, and Medenjak]{alba2019operator}
Vincenzo Alba, Jerome Dubail, and Marko Medenjak.
\newblock Operator entanglement in interacting integrable quantum systems: the
  case of the rule 54 chain.
\newblock {\em Physical Review Letters}, 122\penalty0 (25):\penalty0 250603,
  2019.

\bibitem[Gopalakrishnan et~al.(2018)Gopalakrishnan, Huse, Khemani, and
  Vasseur]{gopalakrishnan2018hydrodynamics}
Sarang Gopalakrishnan, David~A Huse, Vedika Khemani, and Romain Vasseur.
\newblock Hydrodynamics of operator spreading and quasiparticle diffusion in
  interacting integrable systems.
\newblock {\em Physical Review B}, 98\penalty0 (22):\penalty0 220303, 2018.

\bibitem[Stephen et~al.(2019)Stephen, Nautrup, Bermejo-Vega, Eisert, and
  Raussendorf]{stephen2019subsystem}
David~T Stephen, Hendrik~Poulsen Nautrup, Juani Bermejo-Vega, Jens Eisert, and
  Robert Raussendorf.
\newblock Subsystem symmetries, quantum cellular automata, and computational
  phases of quantum matter.
\newblock {\em Quantum}, 3:\penalty0 142, 2019.

\bibitem[Cirac et~al.(2017)Cirac, Perez-Garcia, Schuch, and
  Verstraete]{cirac2017matrix}
J~Ignacio Cirac, David Perez-Garcia, Norbert Schuch, and Frank Verstraete.
\newblock Matrix product unitaries: structure, symmetries, and topological
  invariants.
\newblock {\em Journal of Statistical Mechanics: Theory and Experiment},
  2017\penalty0 (8):\penalty0 083105, 2017.

\bibitem[{\c{S}}ahino{\u{g}}lu et~al.(2018){\c{S}}ahino{\u{g}}lu, Shukla, Bi,
  and Chen]{csahinouglu2018matrix}
M~Burak {\c{S}}ahino{\u{g}}lu, Sujeet~K Shukla, Feng Bi, and Xie Chen.
\newblock Matrix product representation of locality preserving unitaries.
\newblock {\em Physical Review B}, 98\penalty0 (24):\penalty0 245122, 2018.

\bibitem[Gong et~al.(2020)Gong, S{\"u}nderhauf, Schuch, and
  Cirac]{gong2020classification}
Zongping Gong, Christoph S{\"u}nderhauf, Norbert Schuch, and J~Ignacio Cirac.
\newblock Classification of matrix-product unitaries with symmetries.
\newblock {\em Physical Review Letters}, 124\penalty0 (10):\penalty0 100402,
  2020.

\bibitem[Piroli and Cirac(2020)]{piroli2020quantum}
Lorenzo Piroli and J~Ignacio Cirac.
\newblock Quantum cellular automata, tensor networks, and area laws.
\newblock {\em arXiv preprint arXiv:2007.15371}, 2020.

\bibitem[Zhang and Levin(2020)]{zhang2020classification}
Carolyn Zhang and Michael Levin.
\newblock Classification of interacting {F}loquet phases with {$ U(1)$}
  symmetry in two dimensions.
\newblock {\em arXiv preprint arXiv:2010.02253}, 2020.

\bibitem[Farrelly(2019)]{farrelly2019review}
Terry Farrelly.
\newblock A review of quantum cellular automata.
\newblock {\em arXiv preprint arXiv:1904.13318}, 2019.

\bibitem[Arrighi(2019)]{arrighi2019overview}
Pablo Arrighi.
\newblock An overview of quantum cellular automata.
\newblock {\em Natural Computing}, 18\penalty0 (4):\penalty0 885--899, 2019.

\bibitem[Haah et~al.(2018{\natexlab{b}})Haah, Hastings, Kothari, and
  Low]{haah2018quantum}
Jeongwan Haah, Matthew Hastings, Robin Kothari, and Guang~Hao Low.
\newblock Quantum algorithm for simulating real time evolution of lattice
  {Hamiltonians}.
\newblock In {\em 2018 IEEE 59th Annual Symposium on Foundations of Computer
  Science (FOCS)}, pages 350--360. IEEE, 2018{\natexlab{b}}.

\bibitem[Tran et~al.(2019)Tran, Guo, Su, Garrison, Eldredge, Foss-Feig, Childs,
  and Gorshkov]{tran2019locality}
Minh~C Tran, Andrew~Y Guo, Yuan Su, James~R Garrison, Zachary Eldredge, Michael
  Foss-Feig, Andrew~M Childs, and Alexey~V Gorshkov.
\newblock Locality and digital quantum simulation of power-law interactions.
\newblock {\em Physical Review X}, 9\penalty0 (3):\penalty0 031006, 2019.

\bibitem[Hastings(2013)]{hastings2013classifying}
Matthew~B Hastings.
\newblock Classifying quantum phases with the {Kirby} torus trick.
\newblock {\em Physical Review B}, 88\penalty0 (16):\penalty0 165114, 2013.

\bibitem[Christensen(1977{\natexlab{a}})]{christensen1977perturbation}
Erik Christensen.
\newblock Perturbation of operator algebras.
\newblock {\em Inventiones mathematicae}, 43\penalty0 (1):\penalty0 1--13,
  1977{\natexlab{a}}.

\bibitem[Christensen(1980)]{christensen1980near}
Erik Christensen.
\newblock Near inclusions of {$C^*$}-algebras.
\newblock {\em Acta Mathematica}, 144\penalty0 (1):\penalty0 249--265, 1980.

\bibitem[Chao et~al.(2017)Chao, Reichardt, Sutherland, and
  Vidick]{chao2017overlapping}
Rui Chao, Ben~W Reichardt, Chris Sutherland, and Thomas Vidick.
\newblock Overlapping qubits.
\newblock {\em arXiv preprint arXiv:1701.01062}, 2017.

\bibitem[Wilming and Werner(2020)]{wilming2020lieb}
Henrik Wilming and Albert~H Werner.
\newblock {Lieb-Robinson} bounds imply locality of interactions.
\newblock {\em arXiv preprint arXiv:2006.10062}, 2020.

\bibitem[Kitaev(2006)]{kitaev2006anyons}
Alexei Kitaev.
\newblock Anyons in an exactly solved model and beyond.
\newblock {\em Annals of Physics}, 321\penalty0 (1):\penalty0 2--111, 2006.

\bibitem[Duschatko et~al.(2018)Duschatko, Dumitrescu, and
  Potter]{duschatko2018tracking}
Blake~R Duschatko, Philipp~T Dumitrescu, and Andrew~C Potter.
\newblock Tracking the quantized information transfer at the edge of a chiral
  {Floquet} phase.
\newblock {\em Physical Review B}, 98\penalty0 (5):\penalty0 054309, 2018.

\bibitem[Bratteli and Robinson(2012)]{bratteli2012operator}
Ola Bratteli and Derek~W Robinson.
\newblock {\em Operator Algebras and Quantum Statistical Mechanics}, volume~1.
\newblock Springer Science \& Business Media, 2012.

\bibitem[Naaijkens(2013)]{naaijkens2013quantum}
Pieter Naaijkens.
\newblock {\em Quantum spin systems on infinite lattices}.
\newblock Springer, 2013.

\bibitem[Kadison and Ringrose(1997)]{kadison1997fundamentals}
Richard~V Kadison and John~R Ringrose.
\newblock {\em Fundamentals of the theory of operator algebras. {V}olume {I}:
  {E}lementary {T}heory}, volume~1.
\newblock American Mathematical Society, 1997.

\bibitem[Blackadar(2006)]{blackadar2006operator}
Bruce Blackadar.
\newblock {\em Operator algebras: theory of C*-algebras and von {Neumann}
  algebras}, volume 122.
\newblock Springer Science \& Business Media, 2006.

\bibitem[Takesaki(2003)]{takesaki2003theory}
Masamichi Takesaki.
\newblock {\em Theory of operator algebras {III}}, volume 125.
\newblock Springer Science \& Business Media, 2003.

\bibitem[Nachtergaele et~al.(2013)Nachtergaele, Scholz, and
  Werner]{nachtergaele2013local}
Bruno Nachtergaele, Volkher~B Scholz, and Reinhard~F Werner.
\newblock Local approximation of observables and commutator bounds.
\newblock In {\em Operator methods in mathematical physics}, pages 143--149.
  Springer, 2013.

\bibitem[Argerami()]{martinstack}
Martin Argerami.
\newblock Mathematics {S}tack{E}xchange reply.
\newblock URL
  \url{https://math.stackexchange.com/questions/1760817/predual-of-von-neumann-algebra}.

\bibitem[Christensen(1977{\natexlab{b}})]{christensen1977perturbations}
Erik Christensen.
\newblock Perturbations of operator algebras {II}.
\newblock {\em Indiana University Mathematics Journal}, 26\penalty0
  (5):\penalty0 891--904, 1977{\natexlab{b}}.

\bibitem[Christensen et~al.(2012)Christensen, Sinclair, Smith, White, Winter,
  et~al.]{christensen2012perturbations}
Erik Christensen, Allan~M Sinclair, Roger~R Smith, Stuart~A White, Wilhelm
  Winter, et~al.
\newblock Perturbations of nuclear {$C^*$}-algebras.
\newblock {\em Acta mathematica}, 208\penalty0 (1):\penalty0 93--150, 2012.

\bibitem[Ulam(1960)]{ulam1960collection}
Stanislaw~M Ulam.
\newblock A collection of mathematical problems.
\newblock {\em New York}, 29, 1960.

\bibitem[Burger et~al.(2013)Burger, Ozawa, and Thom]{burger2013ulam}
Marc Burger, Narutaka Ozawa, and Andreas Thom.
\newblock On {Ulam} stability.
\newblock {\em Israel Journal of Mathematics}, 193\penalty0 (1):\penalty0
  109--129, 2013.

\bibitem[Johnson(1988)]{johnson1988approximately}
Barry~E Johnson.
\newblock Approximately multiplicative maps between {Banach} algebras.
\newblock {\em Journal of the London Mathematical Society}, 2\penalty0
  (2):\penalty0 294--316, 1988.

\bibitem[Park(2004)]{park2004approximate}
Chun-Gil Park.
\newblock On an approximate automorphism on a {$C^*$}-algebra.
\newblock {\em Proceedings of the American Mathematical Society}, 132\penalty0
  (6):\penalty0 1739--1745, 2004.

\bibitem[Johnston et~al.(2009)Johnston, Kribs, and
  Paulsen]{johnston2009computing}
Nathaniel Johnston, David~W Kribs, and Vern~I Paulsen.
\newblock Computing stabilized norms for quantum operations via the theory of
  completely bounded maps.
\newblock {\em Quantum Information \& Computation}, 9\penalty0 (1):\penalty0
  16--35, 2009.

\bibitem[Nachtergaele et~al.(2019)Nachtergaele, Sims, and
  Young]{nachtergaele2019quasi}
Bruno Nachtergaele, Robert Sims, and Amanda Young.
\newblock Quasi-locality bounds for quantum lattice systems. {I}.
  {Lieb-Robinson} bounds, quasi-local maps, and spectral flow automorphisms.
\newblock {\em Journal of Mathematical Physics}, 60\penalty0 (6):\penalty0
  061101, 2019.

\bibitem[Hastings(2010)]{hastings2010locality}
Matthew~B Hastings.
\newblock Locality in quantum systems.
\newblock {\em Quantum Theory from Small to Large Scales}, 95:\penalty0
  171--212, 2010.

\bibitem[Ohya and Petz(2004)]{ohya2004quantum}
Masanori Ohya and D{\'e}nes Petz.
\newblock {\em Quantum entropy and its use}.
\newblock Springer Science \& Business Media, 2004.

\bibitem[Alicki and Fannes(2004)]{alicki2004continuity}
Robert Alicki and Mark Fannes.
\newblock Continuity of quantum conditional information.
\newblock {\em Journal of Physics A: Mathematical and General}, 37\penalty0
  (5):\penalty0 L55, 2004.

\bibitem[Winter(2016)]{winter2016tight}
Andreas Winter.
\newblock Tight uniform continuity bounds for quantum entropies: conditional
  entropy, relative entropy distance and energy constraints.
\newblock {\em Communications in Mathematical Physics}, 347\penalty0
  (1):\penalty0 291--313, 2016.

\bibitem[Wilde(2013)]{wilde2013quantum}
Mark~M Wilde.
\newblock {\em Quantum information theory}.
\newblock Cambridge University Press, 2013.

\bibitem[Zimbor{\'a}s et~al.(2020)Zimbor{\'a}s, Farrelly, Farkas, and
  Masanes]{zimboras2020does}
Zolt{\'a}n Zimbor{\'a}s, Terry Farrelly, Szil{\'a}rd Farkas, and Lluis Masanes.
\newblock Does causal dynamics imply local interactions?
\newblock {\em arXiv preprint arXiv:2006.10707}, 2020.

\bibitem[Lux et~al.(2014)Lux, M{\"u}ller, Mitra, and
  Rosch]{lux2014hydrodynamic}
Jonathan Lux, Jan M{\"u}ller, Aditi Mitra, and Achim Rosch.
\newblock Hydrodynamic long-time tails after a quantum quench.
\newblock {\em Physical Review A}, 89:\penalty0 053608, 2014, 1311.7644.

\bibitem[De~Nardis et~al.(2019)De~Nardis, Bernard, and Doyon]{de2019diffusion}
Jacopo De~Nardis, Denis Bernard, and Benjamin Doyon.
\newblock Diffusion in generalized hydrodynamics and quasiparticle scattering.
\newblock {\em SciPost Physics}, 6:\penalty0 049, 2019.

\bibitem[Bohrdt et~al.(2017)Bohrdt, Mendl, Endres, and
  Knap]{bohrdt2017scrambling}
Annabelle Bohrdt, Christian~B Mendl, Manuel Endres, and Michael Knap.
\newblock Scrambling and thermalization in a diffusive quantum many-body
  system.
\newblock {\em New Journal of Physics}, 19\penalty0 (6):\penalty0 063001, 2017.

\bibitem[Perera et~al.(2014)Perera, Toms, White, and Winter]{perera2014cuntz}
Francesc Perera, Andrew Toms, Stuart White, and Wilhelm Winter.
\newblock The {C}untz semigroup and stability of close {$C^*$}-algebras.
\newblock {\em Analysis \& PDE}, 7\penalty0 (4):\penalty0 929--952, 2014.

\bibitem[Christensen(1975)]{christensen1975perturbations}
Erik Christensen.
\newblock Perturbations of type {I} von {Neumann} algebras.
\newblock {\em Journal of the London Mathematical Society}, 2\penalty0
  (3):\penalty0 395--405, 1975.

\bibitem[Stinespring(1955)]{stinespring1955positive}
W~Forrest Stinespring.
\newblock Positive functions on {$C^*$}-algebras.
\newblock {\em Proceedings of the American Mathematical Society}, 6\penalty0
  (2):\penalty0 211--216, 1955.

\bibitem[Arveson(1969)]{arveson1969subalgebras}
William~B Arveson.
\newblock Subalgebras of {$C^*$}-algebras.
\newblock {\em Acta Mathematica}, 123\penalty0 (1):\penalty0 141--224, 1969.

\end{thebibliography}

\end{document}